\pdfoutput=1
\documentclass[sigconf]{acmart}

\usepackage{amsmath,amsfonts,dsfont,multirow,multicol,epsfig,url,array,makecell,balance,color,epstopdf}
\usepackage{algorithmic}
\usepackage[ruled, vlined, linesnumbered]{algorithm2e}
\usepackage{hhline}
\usepackage{array}
\usepackage{enumerate}
\usepackage{enumitem}
\usepackage{subfigure}
\usepackage{booktabs}
\usepackage{xcolor,colortbl}
\usepackage{bm}
\usepackage{balance}
\usepackage{footnote}
\usepackage{threeparttable}
\usepackage{balance}

\newtheorem{assumption}{Assumption}[section]
\newtheorem{definition}{Definition}[section]
\def\header{\vspace{0.8mm} \noindent}

\def\tblcapup{\vspace{0mm}}
\def\tblcapdown{\vspace{1mm}}
\def\tbldown{\vspace{-0mm}}

\newcommand{\pushright}[1]{\ifmeasuring@#1\else\omit\hfill$\displaystyle#1$\fi\ignorespaces}
\newcommand{\pushleft}[1]{\ifmeasuring@#1\else\omit$\displaystyle#1$\hfill\fi\ignorespaces}

\def\e{\varepsilon}

\def\epi{\bm{\hat{\pi}}}
\def\vpi{\bm{\pi}}
\def\z{\bm{z}}
\def\q{\bm{q}}
\def\eq{\bm{\hat{q}}}
\def\r{\bm{r}}
\def\er{\bm{\hat{r}}}

\def\A{\mathbf{A}}
\def\D{\mathbf{D}}
\def\X{\mathbf{X}}
\def\Z{\mathbf{Z}}

\def\E{\mathrm{E}}
\def\Var{\mathrm{Var}}

\renewcommand{\vec}{\bm}

\AtBeginDocument{%
  \providecommand\BibTeX{{%
    \normalfont B\kern-0.5em{\scshape i\kern-0.25em b}\kern-0.8em\TeX}}}

\copyrightyear{2021}
\acmYear{2021}
\setcopyright{acmlicensed}\acmConference[KDD '21]{Proceedings of the 27th ACM SIGKDD Conference on Knowledge Discovery and Data Mining}{August 14--18, 2021}{Virtual Event, Singapore}
\acmBooktitle{Proceedings of the 27th ACM SIGKDD Conference on Knowledge Discovery and Data Mining (KDD '21), August 14--18, 2021, Virtual Event, Singapore}
\acmPrice{15.00}
\acmDOI{10.1145/3447548.3467243} \acmISBN{978-1-4503-8332-5/21/08}

\begin{document}
\fancyhead{}
\title{Approximate Graph Propagation}
\subtitle{[Technical Report]}

\author{Hanzhi Wang}
%\author{Mingguo He}
%\author{Zhewei Wei}
%\authornote{Zhewei Wei is the corresponding author. Work partially done at Gaoling School of Artificial Intelligence, Beijing Key Laboratory of Big Data Management and Analysis Methods, MOE Key Lab DEKE, Renmin University of China, and Pazhou Lab, Guangzhou, 510330, China.}
%\email{hanzhi_wang,mingguo@ruc.edu.cn}
\email{hanzhi_wang@ruc.edu.cn}
\affiliation{%
  \institution{Renmin University of China}
  \city{Beijing}
  \country{China}
}

\author{Mingguo He}
\email{mingguo@ruc.edu.cn}
 \affiliation{%
   \institution{Renmin University of China}
   \city{Beijing}
   \country{China}
}

\author{Zhewei Wei}
\email{zhewei@ruc.edu.cn}
\authornote{Zhewei Wei is the corresponding author. Work partially done at Gaoling School of Artificial Intelligence, Beijing Key Laboratory of Big Data Management and Analysis Methods, MOE Key Lab DEKE, Renmin University of China, and Pazhou Lab, Guangzhou, 510330, China.}
\affiliation{%
  \institution{Renmin University of China}
  \city{Beijing}
  \country{China}
}

\author{Sibo Wang}
\email{swang@se.cuhk.edu.hk}
\affiliation{%
 \institution{The Chinese University of Hong Kong}
 \city{Hong Kong}
 \country{China}
}

\author{Ye Yuan}
\email{yuan-ye@bit.edu.cn}
\affiliation{%
  \institution{Beijing Institute of Technology}
  \city{Beijing}
  \country{China}}

\author{Xiaoyong Du}
\author{Ji-Rong Wen}
%\email{duyong@ruc.edu.cn}
\email{duyong, jrwen@ruc.edu.cn}
\affiliation{%
  \institution{Renmin University of China}
  \city{Beijing}
  \country{China}
}

%\author{Ji-Rong Wen}
%\email{jrwen@ruc.edu.cn}
%\affiliation{%
%  \institution{Renmin University of China}
%  \city{Beijing}
%  \country{China}
%}




\begin{abstract}
Efficient computation of node proximity queries such as transition probabilities, Personalized PageRank, and Katz are of fundamental importance in various graph mining and learning tasks. In particular, several recent works leverage fast node proximity computation to improve the scalability of Graph Neural Networks (GNN). However, prior studies on proximity computation and GNN feature propagation are on a case-by-case basis, with each paper focusing on a particular proximity measure. % Moreover, existing feature propagation methods in GNNs rarely leverage approximate algorithms for proximity computation, limiting their scalability on billion-scale graphs. 
%Efficient computation of node proximity queries such as transition probabilities, Personalized PageRank, and Katz are of fundamental importance in various graph mining and learning tasks, including graph clustering, Graph Neural Network(GNN), and label propagation. However, prior studies on proximity computation are on a case-by-case basis, with each paper focusing on a particular proximity measure. Moreover, several recent works leverage fast node proximity computation to improve scalability. In particular, existing feature propagation methods in GNNs rarely leverage approximate algorithms for proximity computation, limiting their scalability on billion-scale graphs. 

In this paper, we propose {Approximate Graph Propagation (AGP)}, a unified randomized algorithm that computes various proximity queries and GNN feature propagation, including transition probabilities, Personalized PageRank, heat kernel PageRank, Katz, SGC, GDC, and APPNP. Our algorithm provides a theoretical bounded error guarantee and runs in almost optimal time complexity. We conduct an extensive experimental study to demonstrate AGP's effectiveness in two concrete applications: local clustering with heat kernel PageRank and node classification with GNNs. Most notably, we present an empirical study on a  billion-edge graph Papers100M, the largest publicly available  GNN dataset so far. The results show that AGP can significantly improve various existing GNN models' scalability without sacrificing prediction accuracy. %Codes for our experiments and the technical report with detailed proofs can be found at~\cite{TechnicalReport}. 

%In this paper, we propose {\it Approximate Graph Propagation (AGP)}, a unified randomized algorithm that computes various proximity queries. Our algorithm provides a theoretical bounded error guarantee and runs in almost optimal time complexity. Besides, we conduct an extensive experimental study to demonstrate AGP's effectiveness in two concrete applications: local clustering and node classification with GNNs. Most notably, we present an empirical study on a billion-edge graph Papers100M, the largest publicly available GNN dataset so far. The results show that AGP can significantly improve various existing GNN models' scalability without sacrificing prediction accuracy. 

\end{abstract}

%%% Local Variables:
%%% mode: latex
%%% TeX-master: "paper"
%%% End:

\begin{CCSXML}
<ccs2012>
   <concept>
       <concept_id>10002950.10003624.10003633.10010917</concept_id>
       <concept_desc>Mathematics of computing~Graph algorithms</concept_desc>
       <concept_significance>500</concept_significance>
       </concept>
   <concept>
       <concept_id>10002951.10003227.10003351</concept_id>
       <concept_desc>Information systems~Data mining</concept_desc>
       <concept_significance>500</concept_significance>
       </concept>
 </ccs2012>
\end{CCSXML}

\ccsdesc[500]{Mathematics of computing~Graph algorithms}
\ccsdesc[500]{Information systems~Data mining}

\keywords{node proximity queries, local clustering, Graph Neural Networks}


\pagestyle{plain}

\maketitle

%\vspace{-2mm}
\section{Introduction} \label{sec:intro}
Recently, significant research effort has been devoted to compute {\em node proximty queries} such as Personalized PageRank~\cite{page1999pagerank,jung2017bepi,Wang2017FORA,wei2018topppr}, heat kernel PageRank~\cite{chung2007HKPR,yang2019TEA} and the Katz score~\cite{katz1953Katz}. %Node proximity queries find numerous applications in the area of graph mining, such as link prediction in social networks~\cite{backstrom2011supervised, Liben2003link}, personalized graph search techniques~\cite{jeh2003scaling,ivan2011cell}, fraud detection~\cite{andersen2008robust,Spirin2011spam}, and collaborative filtering in recommender networks~\cite{gupta2013wtf,liu2017related}. 
Given a node $s$ in an undirected graph $G=(V,E)$ with $|V|=n$ nodes and $|E|=m$ edges, a node proximity query  returns an $n$-dimensional vector $\vec{\pi}$ such that $\vec{\pi}(v)$ represents the importance of node $v$ with respect to $s$.  
% Some of the most widely used proximity measures include transition probabilities, PageRank and Personalized PageRank~\cite{page1999pagerank}, the Katz score~\cite{katz1953Katz}, and heat kernel PageRank~\cite{chung2007HKPR}. 
For example, a widely used proximity measure is the $ L $-th transition probability vector. It captures the $ L $-hop neighbors' information by computing the probability that a $ L $-step random walk from a given source node $s$ reaches each node in the graph. 
The vector form is given by  $	\vec{\pi}= \left(\mathbf{A} \mathbf{D}^{-1} \right)^L \cdot \vec{e}_s,$
% \begin{equation}\label{eqn:transition_definition}
% 	\begin{aligned}
% 		\vec{\pi}= \left(\mathbf{A} \mathbf{D}^{-1} \right)^L \cdot \vec{e}_s,
% 	\end{aligned}
% \end{equation}
where $\A $ is the adjacency matrix, $\D$ is the diagonal degree matrix with $\D(i,i) = \sum_{j \in V} \A(i,j)$, and $\vec{e}_s$ is the one-hot vector with $\vec{e}_s(s)=1$ and $\vec{e}_s(v)=0, v\neq s$. 
Node proximity queries find numerous applications in the area of graph mining, such as link prediction in social networks~\cite{backstrom2011supervised}, personalized graph search techniques~\cite{jeh2003scaling}, fraud detection~\cite{andersen2008robust}, and collaborative filtering in recommender networks~\cite{gupta2013wtf}. 
% Recently, there is a trends to use node proximities to model and scale up the Graph Neural Networks (GNNs). In particular, 

%In particular, a recent trend in Graph Neural Networks (GNNs)  research~\cite{wu2019SGC,klicpera2019GDC,Klicpera2018APPNP} employs node proximity query to build scalable GNN models. 
In particular, a recent trend in Graph Neural Network (GNN)  researches~\cite{wu2019SGC,klicpera2019GDC,Klicpera2018APPNP} is to employ node proximity queries to build scalable GNN models.
A typical example is SGC~\cite{wu2019SGC}, which simplifies the original Graph Convolutional Network (GCN)~\cite{kipf2016GCN} with a linear propagation process. More precisely, given a self-looped graph and an  $n \times d$ feature matrix $\mathbf{X}$, 
%SGC takes the multiplication of $L$-th power of the normalized adjacency matrix $\mathbf{D}^{-\frac{1}{2}}  \mathbf{A} \mathbf{D}^{-\frac{1}{2}} $ and feature matrix  $\mathbf{X}$ 
SGC takes the multiplication of the $L$-th normalized transition probability matrix $\left(\mathbf{D}^{-\frac{1}{2}}  \mathbf{A}  \mathbf{D}^{-\frac{1}{2}} \right)^L$ and the feature matrix $\mathbf{X}$
to form the representation matrix %$\mathbf{Z}\hspace{-0.5mm}= \hspace{-0.5mm}\left(\mathbf{D}^{-\frac{1}{2}}  \mathbf{A}  \mathbf{D}^{-\frac{1}{2}} \right)^L \hspace{-1mm}\cdot \hspace{-0.5mm}\X$. 
\begin{equation}
\begin{aligned}\label{eqn:SGC}
\mathbf{Z}=\left(\mathbf{D}^{-\frac{1}{2}} \cdot \mathbf{A} \cdot \mathbf{D}^{-\frac{1}{2}} \right)^L \cdot \vec{X}.
\end{aligned}	
\end{equation}
If we treat each column of the feature matrix $\X$ as a graph signal vector $\bm{x}$, then the representation matrix $\mathbf{Z}$ can be derived by the augment of $d$ vectors $\bm{\pi}\hspace{-0.5mm}= \hspace{-0.5mm}\left(\mathbf{D}^{-\frac{1}{2}} \mathbf{A}  \mathbf{D}^{-\frac{1}{2}} \right)^L \hspace{-1mm}\cdot \bm{x}$. SGC feeds $\mathbf{Z}$ into a logistic regression or a standard neural network for downstream machine learning tasks such as node classification and link prediction. 
%Similarly, APPNP~\cite{Klicpera2018APPNP}, GDC~\cite{klicpera2019GDC} and GBP~\cite{chen2020GBP} utilize Personalized PageRank and Heat Kernel PageRank to capture multi-hop neighborhood information. 
The $L$-th normalized transition probability matrix $\left(\mathbf{D}^{-\frac{1}{2}}  \mathbf{A}  \mathbf{D}^{-\frac{1}{2}} \right)^L$ can be easily generalized to other node proximity models, such as PPR used in APPNP~\cite{Klicpera2018APPNP}, PPRGo~\cite{bojchevski2020scaling} and GBP~\cite{chen2020GBP}, and HKPR used in~\cite{klicpera2019GDC}. 
Compared to the original GCN~\cite{kipf2016GCN} which uses a full-batch training process and stores the representation of each node in the GPU memory, these proximity-based GNNs decouple prediction and propagation and thus allows mini-batch training to improve the scalability of the models. 
%Note that even though the ideas to employ PPR and HKPR models in the feature propagation process are borrowed from APPNP, PPRGo, GBP and GDC, the original papers of APPNP, PPRGo and GDC use extra complex structures to propagate node features. With a slight abuse of notation, we use APPNP and GDC to denote the linear propagation process by substituting $\left(\mathbf{D}^{-\frac{1}{2}}  \mathbf{A}  \mathbf{D}^{-\frac{1}{2}} \right)^L$ in Equation~\eqref{eqn:SGC} to PPR and HKPR models, respectively. 

\header{\bf Graph Propagation.} 
%\header{\bf Problem Definitions. } 
  To model various proximity measures and GNN propagation formulas, we consider the following unified {\em graph propagation equation}:   
\vspace{-2mm}
\begin{equation}\label{eqn:pi_gen}
\vspace{-1mm}
	\begin{aligned}
		\vec{\pi}=\sum_{i=0}^\infty w_i \cdot \left(\mathbf{D}^{-a} \mathbf{A}  \mathbf{D}^{-b} \right)^i \cdot \vec{x}, 
	\end{aligned}
\end{equation}
where $\mathbf{A}$  denotes the adjacency matrix, $\mathbf{D}$ denotes the diagonal degree matrix, $a$ and $b$ are the Laplacian parameters that take values in $[0,1]$, the sequence of $w_i$ for $i=0,1,2,...$ is the weight sequence and $\bm{x}$ is an $n$ dimensional vector. Following the convention of Graph Convolution Networks~\cite{kipf2016GCN}, we refer to $\mathbf{D}^{-a}\mathbf{A}\mathbf{D}^{-b}$ as the {\em propagation matrix}, and $\vec{x}$ as the {\em graph signal vector}.

A key feature of the graph propagation equation~\eqref{eqn:pi_gen} is that we can manipulate parameters $a$, $b$, $w_i$ and $\vec{x}$ to obtain different proximity measures. For example, if we set  $a=0, b=1, w_L=1,   w_i=0$ for $i=0,\ldots, L-1$, and $\vec{x}=\vec{e}_s$, then $\vec{\pi}$ becomes the $L$-th transition probability vector from node $s$. Table~\ref{tbl:propagation} summarizes the proximity measures and GNN models that can be expressed by Equation~\eqref{eqn:pi_gen}.

\header{\bf Approximate Graph Propagation (AGP). } 
In general, it is computational infeasible to compute Equation~\eqref{eqn:pi_gen}
exactly as the summation goes to infinity. Following~\cite{bressan2018sublinear,wang2020RBS},   we will consider an approximate  version of the graph propagation equation~\eqref{eqn:pi_gen}:
\begin{definition}[Approximate propagation with relative error]\label{def:pro-relative}
	%\caption{Node Incode, Edge Saving and Edge Expense in forwardPush}
	Let $\vec{\pi}$ be the graph propagation vector defined in Equation~\eqref{eqn:pi_gen}.
% 	graph $G=(V,E)$, a vector $\vec{x}$ over $n$ nodes, two laplacian parameters $a$ and $b$,  a weight sequence $w_i$ for $i\ge 0$, and 
	Given an error threshold $\delta$, an approximate propagation algorithm has to return an estimation vector $\hat{\vec{\pi}}$, such that  for any $v \in V$, with  $|\vec{\pi}(v)|>\delta$, we have 
	\vspace{-2mm}
	\begin{align}\nonumber
	\vspace{-2mm}
	\left|\vec{\pi}(v)-\hat{\vec{\pi}}(v)\right| \leq \frac{1}{10} \cdot \vec{\pi}(v)
	\end{align}
	 with at least a constant probability (e.g. $99\%$).
\end{definition}
%The above relative error guarantee follows from previous work on PageRank and heat kernel PageRank~\cite{bressan2018sublinear}. 
We note that some previous works~\cite{yang2019TEA,Wang2017FORA} consider the guarantee $\left|\vec{\pi}(v)-\hat{\vec{\pi}}(v)\right| \leq \varepsilon_r \cdot \vec{\pi}(v)$ with probability at least $1-p_f$, where $\varepsilon_r$ is the relative error parameter and $p_f$ is the fail probability.  However, $\varepsilon_r$ and $p_f$ are set to be constant in these works. For sake of simplicity and readability, we set $\varepsilon_r = 1/10$ and $p_f = 1\%$ and introduce only one error parameter $\delta$, following the setting in~\cite{bressan2018sublinear,wang2020RBS}.

% For simplicity, we first consider the undirected graph. We will show how to extend the graph propagation equation~\eqref{eqn:pi_gen} on directed graphs in Section~\ref{sec:other}.

% As real-world graphs, such as social networks and citation networks, have grown in popularity at an extraordinary pace, efficient computation of the proximity queries has also been extensively studied over the last few years. Exact computations of various proximity queries often rely on iterative methods, which generally takes at least $O((m+n)L)$ time, where $n$ is the number of nodes, $m$ is the number of edges, and $L$ is the number of iterations. Such high time complexity limits the algorithm's scalability and makes them infeasible to support real-time proximity queries on large graphs.  Fortunately, many graph mining applications can tolerate a small error of the proximity vector. Therefore, many works focus on the approximate computation of various proximity vectors~\cite{lofgren2014fast,lofgren2015personalized,LofgrenBG15,Teng2004Nibble,yang2019TEA,chung2018computing}.

\header{\bf Motivations.} 
Existing works on proximity computation and GNN feature propagation are on a case-by-case basis, with each paper focusing on a particular proximity measure. For example, despite the similarity between Personalized PageRank and heat kernel PageRank, the two proximity measures admit two completely different sets of algorithms (see~\cite{wei2018topppr,Wang2016HubPPR,jung2017bepi,Shin2015BEAR,coskun2016efficient} for Personalized PageRank and~\cite{yang2019TEA,chung2007HKPR,chung2018computing,kloster2014heat} for heat kernel PageRank). 
% Moreover, 
% existing feature propagation methods in GNN do not leverage the approximate algorithms for proximity computation, limiting their scalability.  
Therefore, a natural question is 
\begin{quote}
	{\em Is there a universal algorithm that computes the approximate graph propagation with near optimal cost? }
\end{quote}

\begin{table*}[t]
	%\vspace{-5mm}
	\centering
	\renewcommand{\arraystretch}{1.5}
	\tblcapup
	\caption{Typical graph propagation equations.}
	\vspace{-4mm}
	\tblcapdown
	\begin{small}
		\begin{tabular}{|c|c|c|c|c|c|c|} \hline

					~&{\bf Algorithm} & {\bf $\boldsymbol{a}$}& {\bf $\boldsymbol{b}$} & {\bf $\boldsymbol{w_i}$} & {\bf $\boldsymbol{\vec{x}}$} & {\bf Propagation equation} \\ \cline{1-7}
			\multirow{6}{*}{Proximity} &	{$L$-hop transition probability} & $0$ & $1$ & $w_i=0 (i \neq L), w_L=1$ & one-hot vector $\vec{e}_s$ & $\vec{\pi}=\left(\mathbf{A} \mathbf{D}^{-1} \right)^L \cdot \vec{e}_s $\\ \cline{2-7}
			~&{PageRank}~\cite{page1999pagerank} & $0$ & $1$& $\alpha \left( 1-\alpha \right)^i $ & $\frac{1}{n}\cdot \vec{1}$ & $\vec{\pi}=\sum_{i=0}^\infty \alpha \left( 1-\alpha \right)^i \hspace{-1mm}\cdot\hspace{-0.5mm} \left(\mathbf{A} \mathbf{D}^{-1} \right)^i \hspace{-1mm}\cdot \hspace{-0.5mm} \frac{\vec{1}}{n}$ \\ \cline{2-7}
			~&{Personalized PageRank}~\cite{page1999pagerank} &$0$ & $1$ & $\alpha \left( 1-\alpha \right)^i $ & teleport probability distribution $x$ & $\vec{\pi}=\sum_{i=0}^\infty \alpha \left( 1-\alpha \right)^i \cdot \left(\mathbf{A} \mathbf{D}^{-1} \right)^i \cdot \bm{x} $ \\ \cline{2-7}
			~&{single-target PPR}~\cite{lofgren2013personalized} & $1$ & $0$ & $\alpha \left( 1-\alpha \right)^i $ & one-hot vector $\vec{e}_v$ & $\vec{\pi}=\sum_{i=0}^\infty \alpha \left( 1-\alpha \right)^i \cdot \left(\mathbf{D}^{-1} \mathbf{A} \right)^i \cdot \vec{e}_v $\\ \cline{2-7}
			~&{heat kernel PageRank}~\cite{chung2007HKPR}  &$0$ & $1$ & $e^{-t} \cdot \frac{t^i}{i!} $ & one-hot vector $\vec{e}_s$ & $\vec{\pi}=\sum_{i=0}^\infty e^{-t} \cdot \frac{t^i}{i!}\cdot \left(\mathbf{A} \mathbf{D}^{-1} \right)^i \cdot \vec{e}_s $\\ \cline{2-7}
			~&{Katz}~\cite{katz1953Katz} & $0$ & $0$ & $\beta^i$ & one-hot vector $\vec{e}_s$ & $\vec{\pi}=\sum_{i=0}^\infty \beta^i  \mathbf{A}^i \cdot \vec{e}_s$\\ \hline
			
			\multirow{3}{*}{GNN} & SGC~\cite{wu2019SGC} & $\frac{1}{2}$ & $\frac{1}{2}$ & $w_i=0 (i \neq L), w_L=1$ & the graph signal $\bm{x}$ & $\vec{\pi}=\left(\mathbf{D}^{-\frac{1}{2}} \mathbf{A} \mathbf{D}^{-\frac{1}{2}} \right)^L \cdot \vec{x} $\\ \cline{2-7}
			~& APPNP~\cite{Klicpera2018APPNP}  &  $\frac{1}{2}$ & $\frac{1}{2}$ & $\alpha \left( 1-\alpha \right)^i$ & the graph signal $\bm{x}$ & $\vec{\pi}=\sum_{i=0}^L \alpha \left( 1-\alpha \right)^i \hspace{-1mm}\cdot \hspace{-0.5mm} \left(\mathbf{D}^{-\frac{1}{2}}\mathbf{A} \mathbf{D}^{-\frac{1}{2}} \right)^i \hspace{-1.5mm}\cdot \hspace{-0.5mm} \vec{x} $ \\ \cline{2-7}
			~& GDC~\cite{klicpera2019GDC}  &  $\frac{1}{2}$ & $\frac{1}{2}$ & $e^{-t} \cdot \frac{t^i}{i!} $  & the graph signal $\bm{x}$ & $\vec{\pi}=\sum_{i=0}^L e^{-t} \cdot \frac{t^i}{i!}\hspace{-1mm}\cdot \hspace{-0.5mm} \left(\mathbf{D}^{-\frac{1}{2}}\mathbf{A}\mathbf{D}^{-\frac{1}{2}} \right)^i \hspace{-1.5mm}\cdot \hspace{-0.5mm} \vec{x} $ \\ \hline
		\end{tabular}
	\end{small}
	\label{tbl:propagation}
	%\tbldown
	%\vspace{-3mm}
\end{table*}

\header{\bf Contributions.} In this paper, we present AGP, a UNIFIED randomized algorithm that computes Equation~\eqref{eqn:pi_gen} with almost optimal computation time and theoretical error guarantee. AGP naturally generalizes to  various proximity measures, including transition probabilities, PageRank and Personalized PageRank, heat kernel PageRank, and Katz. We conduct an extensive experimental study to demonstrate the effectiveness of AGP on real-world graphs. We show that AGP outperforms the state-of-the-art methods for local graph clustering with heat kernel PageRank.  We  also show that AGP can scale various GNN models (including SGC~\cite{wu2019SGC}, APPNP~\cite{Klicpera2018APPNP}, and GDC~\cite{klicpera2019GDC}) up on the billion-edge graph Papers100M, which is the largest publicly available  GNN dataset.

\section{Preliminary and Related Work} \label{sec:pre}

\begin{table} [t]
	\centering
	\renewcommand{\arraystretch}{1.3}
	\begin{small}
		\tblcapup
		\caption{Table of notations.}\label{tbl:def-notation}
		\vspace{-4mm}
		%\tblcapdown
		%p{2.3in}
		\begin{tabular} {|l|p{2.3in}|} \hline
			{\bf Notation} &  {\bf Description}  \\ \hline
			$G=(V,E)$ & undirected graph with vertex and   edge sets $V$ and $E$ \\ \hline
			$n, m$      &   the numbers of nodes and edges in $G$                            \\ \hline
			%$N_{in}(u), N_{out}(u)$	& 	the in/out neightbor set of node $u$	\\ \hline
			%$d_{in}(u), d_{out}(u)$ & the in/out degree of node $u$ \\ \hline
			$\mathbf{A}$, $\mathbf{D}$  & the adjacency matrix and degree matrix of $G$\\ \hline
			%$\mathbf{D}$ & the diagonal degree matrix of $G$\\ \hline
			$N_u$, $d_u$	& 	the neighbor set and the degree of node $u$\\ \hline
			%$N_{i}(u), N_{o}(u)$	& 	the in-/out-neighbor set of node $u$	\\ \hline
			%$d_u$ & the degree of node $u$ on undirected graph\\ \hline
			%$d_{i}(u), d_{o}(u)$ & the in-/out-degree of node $u$ \\ \hline
			$a,b$ & the Laplacian parameters\\ \hline
			$\vec{x}$ & the graph signal vector in $\mathcal{R}^n$, $\left\| \vec{x} \right\|_{2}=1$ \\ \hline
			$w_i, Y_i$ & the $i$-th weight  and partial sum $Y_i=\sum_{k=i}^\infty w_k$\\ \hline
			%	$L_E = \sum_{i=0}^\infty i w_i $  & the average propagation length\\ \hline
			%$\vec{w}$ & the weighted vector with $w_i$ as the $i_{th}$ entry\\ \hline
			$\vec{\pi},\hat{\vec{\pi}}$ & the true and estimated propagation vectors in $\mathcal{R}^n$\\ \hline
			$\vec{r}^{(i)},\hat{\vec{r}}^{(i)}$ & the true and estimated $i$-hop residue vectors in $\mathcal{R}^n$\\ \hline
			$\vec{q}^{(i)},\hat{\vec{q}}^{(i)}$ & the true and estimated $i$-hop reserve vectors in $\mathcal{R}^n$\\ \hline
			%$\e_r,\delta$ & the relative error and threshold \\ \hline
			$\delta$ & the relative error threshold \\ \hline
			$\tilde{O}$ & the Big-Oh natation ignoring the log factors \\ \hline
		\end{tabular}
		%\vspace{-3mm}
	\end{small}
\end{table}

In this section, we provide a detailed discussion on how the graph propagation equation~\eqref{eqn:pi_gen} models various node proximity measures. %We will also discuss how approximate graph propagation improves the scalability of  various applications, including node proximity computation, local clustering, and graph neural networks. 
Table~\ref{tbl:def-notation} summarizes the notations used in this paper.

\header{\bf  Personalized PageRank (PPR) }~\cite{page1999pagerank} is  developed by Google to rank web pages on the world wide web, with the intuition that "a page is important if it is referenced by many pages, or important pages". Given an undirected graph $G=(V,E)$ with $n$ nodes and $m$ edges and a {\em teleporting probability distribution} $\vec{x}$ over the $n$ nodes, the PPR  vector $\vec{\pi}$ is the solution to the following equation:
\vspace{-1mm}
\begin{equation}\label{eqn:PageRank_definition}
\vspace{-2mm}
	\begin{aligned}
	\vspace{-2mm}
		\vec{\pi}=\left( 1-\alpha \right) \cdot \mathbf{A}\mathbf{D}^{-1} \cdot \vec{\pi}+ \alpha\vec{x}.
	\end{aligned}
	\hspace{-2mm}
\end{equation}
The unique solution to Equation~\eqref{eqn:PageRank_definition} is given by $\vec{\pi}=\sum_{i=0}^\infty \alpha \left( 1-\alpha \right)^i \cdot \left(\mathbf{A} \mathbf{D}^{-1} \right)^i \cdot \vec{x}.$
% \begin{equation}\nonumber
% 	\begin{aligned}
% 		\vec{\pi}=\sum_{i=0}^\infty \alpha \left( 1-\alpha \right)^i \cdot \left(\mathbf{A} \cdot\mathbf{D}^{-1} \right)^i \cdot \vec{x},
% 	\end{aligned}
% \end{equation}
%  which can be expressed by the graph propagation if we set  $a=0, b=1, w_i=\alpha \cdot \left(1-\alpha\right)^i$, and $\vec{x}=\vec{e}_s $ (or $\frac{1}{n}\cdot \vec{1}$ for PageRank) in equation~\eqref{eqn:pi_gen}.

%  PageRank and Personalized PageRank reflects the structural importance of each node, and  is widely used in graph mining tasks~\cite{Liben-NowellK03,gupta2013wtf} and graph representation learning tasks~\cite{ou2016asymmetric,tsitsulin2018verse,zhou2017scalable,Klicpera2018APPNP}. Consequently, the efficient computation of the PPR vector has been extensively studied for the past decades. Due to space limits, we only discuss a few works that are closely related to our AGP framework. 
 
The efficient computation of the PPR vector has been extensively studied for the past decades. A simple algorithm to estimate PPR is the Monte-Carlo sampling~\cite{Fogaras2005MC}, which stimulates adequate random walks from the source node $s$ generated following $\bm{x}$ and uses the percentage of the walks that terminate at node $v$ as the estimation of $\vec{\pi}(v)$.  Forward Search~\cite{FOCS06_FS} conducts deterministic local pushes from the source node $s$ to find nodes with large PPR scores. FORA~\cite{Wang2017FORA} combines Forward Search with the Monte Carlo method to improve the computation efficiency. TopPPR~\cite{wei2018topppr} combines Forward Search, Monte Carlo, and Backward Search~\cite{lofgren2013personalized} to obtain a better error guarantee for  top-$k$ PPR estimation. ResAcc~\cite{lin2020index} refines FORA by accumulating probability masses before each push.  However, these methods only work for graph propagation with transition matrix $\mathbf{A} \mathbf{D}^{-1}$. For example, the Monte Carlo method simulates random walks to obtain the estimation, which is not possible with the general propagation matrix $\mathbf{D}^{-a}\mathbf{A}\mathbf{D}^{-b}$. 
 
Another line of research~\cite{lofgren2013personalized,wang2020RBS} studies the {\em single-target PPR}, which asks for the PPR value of every node to a given target node $v$ on the graph.  The single-target PPR vector for a given node $v$ is defined by the slightly different formula: %$\vec{\pi} =\left( 1-\alpha \right) \cdot \left(\mathbf{D}^{-1}\mathbf{A} \right)\cdot \pi_v+\alpha \vec{e}_v.$
\vspace{-2mm}
\begin{equation}\label{eqn:st_PPR_definition}
\vspace{-2mm}
 	\begin{aligned}
 		\vec{\pi} =\left( 1-\alpha \right) \cdot \left(\mathbf{D}^{-1} \mathbf{A} \right)\cdot \vec{\pi}+\alpha \vec{e}_v.  
	\end{aligned}
 \end{equation}
Unlike the single-source PPR vector, the single-target PPR vector $\vec{\pi}$ is not a probability distribution, which means $\sum_{s\in V}\vec{\pi}(s)$ may not equal to $1$. 
We can also derive the unique solution to Equation~\eqref{eqn:st_PPR_definition} by $\vec{\pi} =\sum_{i=0}^\infty \alpha \left( 1-\alpha \right)^i \cdot \left(\mathbf{D}^{-1}\mathbf{A} \right)^i \cdot \vec{e}_v$.
% \begin{equation}\nonumber
% 	\begin{aligned}
% 		\vec{\pi} =\sum_{i=0}^\infty \alpha \left( 1-\alpha \right)^i \cdot \left(\mathbf{D}^{-1} \cdot \mathbf{A} \right)^i \cdot \vec{e}_v, 
% 	\end{aligned}
% \end{equation}
%  which can be formulated as the graph propagation~\eqref{eqn:pi_gen} by setting $a=1,b=0, w_i=\alpha \left( 1-\alpha \right)^i, \vec{x}=\vec{e}_v$. 

\header{\bf Heat kernal PageRank} is proposed by~\cite{chung2007HKPR} for high quality community detection. For each node $v\in V$ and the seed node $s$, the heat kernel PageRank (HKPR) $\vpi(v)$ equals to the probability that a heat kernal random walk starting from node $s$ ends at node $v$. The length  $L$ of the random walks follows the Poisson distribution with parameter $t$, i.e. $\Pr[L = i] = \frac{e^{-t}t^i}{i!}, i=0,\ldots, \infty$. Consequently, the HKPR vector of a given node $s$ is defined as $\vec{\pi}=\sum_{i=0}^\infty \frac{e^{-t}t^i}{i!}\cdot \left(\mathbf{A} \mathbf{D}^{-1}\right)^i\cdot \vec{e}_s$, 
%\vspace{-1mm}
%\begin{equation}
%\begin{aligned}\label{eqn:definition_HKPR}
%\vspace{-1mm}
%	\vec{\pi}=\sum_{i=0}^\infty \frac{e^{-t}t^i}{i!}\cdot \left(\mathbf{A} \mathbf{D}^{-1}\right)^i\cdot \vec{e}_s,
%\end{aligned}	
%\end{equation}
where $\vec{e}_s$ is the one-hot vector with $\vec{e}_s(s)=1$.  
%Compared with Personalized PageRank, the length of random heat kernel PageRank
%Equation~\eqref{eqn:definition_HKPR} 
This equation fits in the framework of our generalized propagation equation~\eqref{eqn:pi_gen} if we set $a=0,b=1, w_i=\frac{e^{-t}t^i}{i!}$, and $\vec{x}=\vec{e}_s$. Similar to PPR, HKPR can be estimated by the Monte Carlo method~\cite{chung2018computing,yang2019TEA} that simulates random walks of Possion distributed length. HK-Relax~\cite{kloster2014heat} utilizes Forward Search to approximate the HKPR vector. TEA~\cite{yang2019TEA} combines Forward Search with  Monte Carlo for a more accurate estimator.

\header{\bf Katz index}~\cite{katz1953Katz} is another popular proximity measurement to evaluate relative importance of nodes on the graph. Given two node $s$ and $v$, the Katz score between $s$ and $v$ is characterized by the number of reachable paths from $s$ to $v$. Thus, the Katz vector for a given source node $s$ can be expressed as $\vec{\pi}=\sum_{i=0}^\infty \mathbf{A}^{i} \cdot \vec{e}_s,$ where $\mathbf{A} $ is the adjacency matrix and $\vec{e}_s$ is the one-hot vector with $\vec{e}_s(s)\hspace{-1mm} =\hspace{-1mm} 1$. However, this summation may not converge due to the large spectral span of $\mathbf{A}$. A commonly used fix up is to apply a penalty of $\beta$ to each step of the path, leading to the following definition: $\vec{\pi}=\sum_{i=0}^\infty \beta^i \cdot \mathbf{A}^{i} \cdot \vec{e}_s$. 
%\vspace{-1mm}
%\begin{align}~\label{eqn:Katz_defintion}
%\vspace{-1mm}
%	\vec{\pi}=\sum_{i=0}^\infty \beta^i \cdot \mathbf{A}^{i} \cdot \vec{e}_s. 
%\end{align}
To guarantee convergence, $\beta$ is a constant that set to be smaller than $\frac{1}{\lambda_1}$, where $\lambda_1$ is  the largest eigenvalue of the adjacency matrix $\mathbf{A}$. 
% Let $a=b=0, w_i=\beta^i, \vec{x}=\vec{e}_s$, our generalized propagation equation~\eqref{eqn:pi_gen} can also be used to compute the Katz index.
Similar to PPR, the Katz vector can be computed by iterative multiplying $\vec{e}_s$ with $\mathbf{A}$, which runs in $\Tilde{O}(m+n)$ time~\cite{foster2001faster}.  Katz has been widely used in graph analytic and learning tasks such as link prediction~\cite{Liben2003link} and graph embedding~\cite{ou2016asymmetric}. However, the $\Tilde{O}(m+n)$ computation time limits its scalability on large graphs.

\header{\bf Proximity-based Graph Neural Networks.}
Consider an undirected graph $G=(V,E)$, where $V$ and $E$ represent the set of vertices and edges. Each node $v$ is associated with a numeric feature vector of dimension $d$. The $n$ feature vectors form an $n \times d$ matrix $\mathbf{X}$. Following the convention of graph neural networks~\cite{kipf2016GCN,hamilton2017graphSAGE}, we assume each node in $G$ is also attached with a self-loop.  The goal of graph neural network is to obtain an $n\times d'$ representation matrix $\mathbf{Z}$, which encodes both the graph structural information and the feature matrix $\mathbf{X}$.
%We then feed $\mathbf{Z}$ into a neural network for downstream machine learning tasks such as node classification and link prediction. 
Kipf and Welling~\cite{kipf2016GCN} propose the vanilla Graph Convolutional Network (GCN), of which the $\ell$-th representation $\mathbf{H}^{(\ell)}$ is defined as  $\mathbf{H}^{(\ell)}=\sigma\left(\mathbf{{D}^{-\frac{1}{2}}} \mathbf{{A}}\mathbf{{D}^{-\frac{1}{2}}}\mathbf{H}^{(\ell-1)} \mathbf{W}^{(\ell)}\right)$,
%\begin{equation}
%\begin{aligned}\label{eqn:definition_GCN}
%	\mathbf{H}^{(\ell)}=\sigma\left(\mathbf{{D}^{-\frac{1}{2}}} \cdot \mathbf{{A}} \cdot \mathbf{{D}^{-\frac{1}{2}}} \cdot \mathbf{H}^{(\ell-1)} \cdot \mathbf{W}^{(\ell)}\right),
%\end{aligned}	
%\end{equation}
where $\mathbf{A}$ and $\mathbf{D}$ are the adjacency matrix and the diagonal degree matrix of $G$, $\mathbf{W}^{(\ell)}$ is the learnable weight matrix, and $\sigma(.)$ is a non-linear activation function (a common choice is the Relu function). Let $L$ denote the number of layers in the GCN model. The $0$-th representation $\mathbf{H}^{(0)}$ is set to the feature matrix $\mathbf{X}$, and the final representation matrix $\mathbf{Z}$ is the $L$-th representation $\mathbf{H}^{(L)}$. Intuitively, GCN aggregates the neighbors' representation vectors from the $(\ell-1)$-th layer to form the representation of the $\ell $-th layer. Such a simple paradigm is proved to be effective in various graph learning tasks~\cite{kipf2016GCN,hamilton2017graphSAGE}.

% where $\mathbf{\tilde{A}}$ is adjacency matrix of $G$ with added self-connections, $\mathbf{\tilde{D}}$ is the corresponding degree matrix, $\mathbf{W}^{(l)}$ is the learnable weight matrix and $\sigma(.)$ is the activation function.

A major drawback of the vanilla GCN is the lack of ability to scale on graphs with millions of nodes. Such limitation is caused by the fact that the vanilla GCN uses a full-batch training process and stores each node's representation in the GPU memory. 
To extend GNN to large graphs, a line of research focuses on decoupling prediction and propagation, which removes the non-linear activation function $\sigma(.)$ for better scalability. These methods first apply a proximity matrix to the feature matrix $\mathbf{X}$ to obtain the representation matrix $\mathbf{Z}$, and then feed $\mathbf{Z}$ into logistic regression or standard neural network for predictions. 
%Among them, SGC~\cite{wu2019SGC} simplifies the vanilla GCN by taking the multiplication of $L$-th power of the normalized adjacency matrix $\mathbf{A}$ and feature matrix  $\mathbf{X}$ to form the final presentation $\mathbf{Z}=\P\cdot \X=\left(\mathbf{D}^{-\frac{1}{2}}\mathbf{A} \mathbf{D}^{-\frac{1}{2}} \right)^L \cdot \X$.
Among them, SGC~\cite{wu2019SGC} simplifies the vanilla GCN by taking the multiplication of the $L$-th normalized transition probability matrix $\left(\mathbf{D}^{-\frac{1}{2}}\mathbf{A} \mathbf{D}^{-\frac{1}{2}} \right)^L$ and the feature matrix $\mathbf{X}$ to form the final presentation $\mathbf{Z}=\left(\mathbf{D}^{-\frac{1}{2}}\mathbf{A} \mathbf{D}^{-\frac{1}{2}} \right)^L \cdot \X$.  %(see Equation~\eqref{eqn:SGC}).
% state-of-art scalable GNNs use various techniques to decouple prediction and propagation and approximate the full neighbor propagation. The main models are: 
% \begin{equation}
% \begin{aligned}\label{eqn:SGC}
% \mathbf{Z}= \left(\mathbf{D}^{-\frac{1}{2}} \cdot \mathbf{A} \cdot \mathbf{D}^{-\frac{1}{2}} \right)^L \cdot \vec{X},
% \end{aligned}	
% \end{equation}
% To make prediction, SGC performs standard logistic regression $\mathbf{y} = SoftMax\left(\mathbf{Z}\cdot \mathbf{W}\right)$, where $\mathbf{W}$ is a trainable weight matrix.  It is easy to see that if we set $a=\frac{1}{2}$, $b=\frac{1}{2}$, $w_i=0 (i =0,\ldots, L-1), w_L=1$ and  $\vec{x}$ to be a column of the feature matrix $\mathbf{X}$, then SGC also falls into the framework of AGP. 
The proximity matrix can be generalized to PPR used in APPNP~\cite{Klicpera2018APPNP} and HKPR used in GDC~\cite{klicpera2019GDC}. 
Note that even though the ideas to employ PPR and HKPR models in the feature propagation process are borrowed from APPNP, PPRGo, GBP and GDC, the original papers of APPNP, PPRGo and GDC use extra complex structures to propagate node features. For example, the original APPNP~\cite{Klicpera2018APPNP} first applies an one-layer neural network to $\mathbf{X}$ before the propagation that $\mathbf{Z}^{(0)}=f_\theta(\mathbf{X})$. Then APPNP propagates the feature matrix $\mathbf{X}$ with a truncated Personalized PageRank matrix $\mathbf{Z}^{(L)} =\sum_{i=0}^L \alpha \left( 1-\alpha \right)^i \cdot \left(\mathbf{D}^{-\frac{1}{2}}\mathbf{A} \mathbf{D}^{-\frac{1}{2}} \right)^i \cdot f_\theta(\mathbf{X})$, where $L$ is the number of layers and $\alpha$ is a constant in $(0,1)$. The original GDC~\cite{klicpera2019GDC} follows the structure of GCN and employs the heat kernel as the diffusion kernel. 
%For the sake of scalability, with a slight abuse of notation, we use APPNP and GDC to denote the propagation process by substituting the proximity matrix to PPR and HKPR models, respectively. More precisely, unless specificed otherwise, we use APPNP and GDC to denote the linear propagation process that the representation matrix $\mathbf{Z} =\sigma\left(\sum_{i=0}^L w_i \cdot \left(\mathbf{D}^{-\frac{1}{2}}  \mathbf{A}  \mathbf{D}^{-\frac{1}{2}} \right)^{i} \cdot \X\right)$, where $L$ is the number of layers. Specifically, $w_i=\alpha \left( 1-\alpha \right)^i$ for APPNP and $w_i=\frac{e^{-t} t^i}{i!}$ for GDC. 
For the sake of scalability, we use APPNP to denote the propagation process $\mathbf{Z} =\sum_{i=0}^L \alpha \left( 1-\alpha \right)^i \cdot \left(\mathbf{D}^{-\frac{1}{2}}  \mathbf{A}  \mathbf{D}^{-\frac{1}{2}} \right)^{i} \cdot \X$, and GDC to denote the propagation process $\mathbf{Z} =\sum_{i=0}^L \frac{e^{-t} t^i}{i!} \cdot \left(\mathbf{D}^{-\frac{1}{2}}  \mathbf{A}  \mathbf{D}^{-\frac{1}{2}} \right)^{i} \cdot \X$.

A recent work PPRGo~\cite{bojchevski2020scaling} improves the scalability of APPNP by employing the Forward Search algorithm~\cite{FOCS06_FS} to perform the propagation. However, PPRGo only works for APPNP, which, as we shall see in our experiment, may not always achieve the best performance among the three models.  Finally,
a recent work GBP~\cite{chen2020GBP} proposes to use deterministic local push and the Monte Carlo method to approximate GNN propagation of the form $\Z=\sum_{i=0}^L w_i \cdot \left(\mathbf{D}^{-(1-r)}\mathbf{A} \mathbf{D}^{-r} \right)^i \cdot \X.$ However, GBP suffers from two drawbacks: 1) it requires $a+b = 1$ in Equation~\eqref{eqn:pi_gen} to utilize the Monte-Carlo method, and 
%2) it may incur unbounded error in the context of proximity queries (i.e. Definition~\ref{def:pro-relative}). In particular, consider a graph propagation process  with the reverse transition matrix $\D^{-1}\A$. GBP will amplify the error by a factor of $d_u$ for the estimator $\epi(u)$ of any node $u\in V$. The detailed illustrations can be found in the technical report~\cite{TechnicalReport}.
%2) it requires $O(n)$ space cost in the Monte-Carlo process, which limits its scalability on large graphs with billions of edges. As we shall see in Section~\ref{sec:exp}, to run GBP on Friendster and Papers100M, we have to reduce the forward Monte-Carlo phase and only conduct the reverse push to avoid the out-of-memory problems. 
2) it requires a large memory space to store the random walk matrix generated by the Monte-Carlo method.

\vspace{-3mm}
\section{Basic Propagation}
\vspace{-1mm}
\label{sec:BPA}
In the next two sections, we present two algorithms to compute the graph propagation equation~\eqref{eqn:pi_gen} with the theoretical relative error guarantee in Definition~\ref{def:pro-relative}.   
%  Algorithm~\ref{alg:AGP-deter} computes $\vec{\hat{\pi}} = \sum_{i=0}^L w_i \cdot \left(\bm{D}^{-a}\cdot \bm{A} \cdot \bm{D}^{-b} \right)^i \cdot \vec{x}$, which is the propagation equation~\eqref{eqn:pi_gen} truncated at level $L$. The value of $L$ depends on the approximation quality. 
%Let $L_E = \sum_{i=0}^\infty w_i\cdot i = \sum_{i=0}^\infty Y_i$ be the average length of the propagation. 

\vspace{-0.5mm}
\header{\bf Assumption on graph signal $\vec{x}$.} 
For the sake of simplicity, we assume the graph signal $\vec{x}$ is non-negative. We can deal with the negative entries in $\vec{x}$ by decomposing it into $\vec{x} \hspace{-0.5mm}=\hspace{-0.5mm} \vec{x}^{+} \hspace{-0.5mm}+\hspace{-0.5mm} \vec{x}^-$, where $\vec{x}^+$ only contains the non-negative entries of $\vec{x}$ and $\vec{x}^-$ only contains the negative entries of  $\vec{x}$. After we compute $\vec{\pi}^+\hspace{-0.5mm}=\hspace{-0.5mm}\sum_{i=0}^\infty w_i \cdot \left(\mathbf{D}^{-a}\mathbf{A} \mathbf{D}^{-b} \right)^i \hspace{-1mm} \cdot \hspace{-1mm}\vec{x}^+$ and $\vec{\pi}^{-}\hspace{-1mm}=\hspace{-1mm}\sum_{i=0}^\infty w_i \cdot \left(\mathbf{D}^{-a}\mathbf{A} \mathbf{D}^{-b} \right)^i\hspace{-0.5mm} \cdot\hspace{-0.5mm} \vec{x}^-$, we can combine  $\vec{\pi}^+$ and $\vec{\pi}^-$ to form $\vec{\pi}\hspace{-0.5mm} = \hspace{-0.5mm} \vec{\pi}^+ \hspace{-0.5mm}+\hspace{-0.5mm} \vec{\pi}^-$. We also assume $\vec{x}$ is normalized, that is $\|\vec{x}\|_1 \hspace{-1mm}=\hspace{-0.5mm} 1$.

%\vspace{-0.5mm}
\header{\bf Assumptions on $w_i$.} To make the computation of Equation~\eqref{eqn:pi_gen} feasible, we first introduce several assumptions on the weight sequence $w_i$ for $i\in \{0,1,2, ...\}$. We assume $\sum_{i=0}^\infty w_i =1$. If not, we can perform propagation with $w_i' = w_i/\sum_{i=0}^\infty w_i$ and rescale the result by  $\sum_{i=0}^\infty w_i$. We also note that to ensure the convergence of Equation~\ref{eqn:pi_gen}, the weight sequence $w_i$ has to satisfy $\sum_{i=0}^\infty w_i \lambda_{max}^i < \infty$, where $\lambda_{max}$ is the maximum singular value of the propagation matrix $\mathbf{D}^{-a} \mathbf{A}\mathbf{D}^{-b}$. Therefore, we assume that for sufficiently large $i$, $w_i \lambda_{max}^i$ is upper bounded by a geometric distribution: 

%\vspace{-1mm}
\begin{assumption}\label{asm:L}
%\vspace{-4mm}
There exists a constant $L_0$ and $\lambda < 1$, such that for any $i \ge L_0$, $w_i \cdot \lambda_{max}^i  \le \lambda^i$.
%\vspace{-2mm}
\end{assumption}
According to Assumption~\ref{asm:L}, to achieve the relative error in Definition~\ref{def:pro-relative}, we only need to compute the prefix sum $\vec{\pi}=\sum_{i=0}^L w_i \cdot \left(\mathbf{D}^{-a}\mathbf{A}\mathbf{D}^{-b} \right)^i \cdot \vec{x}$, where $L$ equals to $ \log_{\lambda} {\delta} = O\left( \log {1 \over \delta}\right)$.
This property is possessed by all proximity measures discussed in this paper. For example, PageRank and Personalized PageRank set $w_i=\alpha \left(1-\alpha\right)^i$, where $\alpha$ is a constant. Since the maximum eigenvalue of $ \mathbf{A}\mathbf{D}^{-1}$ is $1$, we have $\|\sum_{i=L+1}^\infty w_i \left(\mathbf{A}\mathbf{D}^{-1}\right)^i \hspace{-1mm} \cdot \vec{x}\|_2 \le \|\sum_{i=L+1}^\infty w_i \cdot \vec{x}\|_2 =\sum_{i=L+1}^\infty w_i \cdot \|\vec{x}\|_2 \le \sum_{i=L+1}^\infty w_i \cdot \|\vec{x}\|_1= \sum_{i=L+1}^\infty \alpha \cdot \left(1-\alpha\right)^i =(1-\alpha)^{L+1}$. In the second inequality, we use the fact that $\|\vec{x}\|_2\le \|\vec{x}\|_1$ and the assumption on $\vec{x}$ that $\|\vec{x}\|_1=1$. If we set $L=\log_{1-\alpha} \delta=O\left(\log{\frac{1}{\delta}} \right)$, the remaining sum $\|\sum_{i=L+1}^\infty w_i\cdot \left(\mathbf{A}\mathbf{D}^{-1}\right)^i \cdot \vec{x}\|_2$ is bounded by $\delta$. By the assumption that $\vec{x}$ is non-negative, we can terminate the propagation at the $L$-th level without obvious error increment. We can prove similar bounds for HKPR, Katz, and transition probability as well. Detailed explanations are deferred to the appendix.

\header{\bf Basic Propagation.} 
As a baseline solution, we can compute the graph propagation equation~\eqref{eqn:pi_gen} by iteratively updating the propagation vector $\vec{\pi}$ via matrix-vector multiplications. Similar approaches have been used for computing PageRank, PPR, HKPR and Katz, under the name of Power Iteration or Power Method.

In general, we employ matrix-vector multiplications to compute the summation of the first $L =O\left(\log {1\over \delta}\right)$ hops of Equation~\eqref{eqn:pi_gen}: $\vec{\pi}=\sum_{i=0}^L w_i \cdot \left(\mathbf{D}^{-a}\mathbf{A}\mathbf{D}^{-b} \right)^i \cdot \vec{x}$.  To avoid the $O(nL)$ space of storing vectors $\left(\mathbf{D}^{-a}\mathbf{A} \mathbf{D}^{-b} \right)^i \cdot \vec{x}, i=0,\ldots, L$, we only use two vectors: the {\em residue and reserve vectors}, which are defined as follows.

\begin{definition}\label{def:RQ-relation} [{\bf residue and reserve}]
Let $Y_i$ denote the partial sum of the weight sequence that $Y_i=\sum_{k=i}^\infty w_k, i=0,\ldots, \infty$. Note that $Y_0 = \sum_{k=0}^\infty w_k = 1$ under the assumption: $\sum_{i=0}^\infty w_i=1$.  At level $i$, the residue vector is defined as $\vec{r}^{(i)}=Y_i\cdot \left(\D^{-a}\A\D^{-b} \right)^i \cdot \vec{x}$;  The reserve vector is defined as $\vec{q}^{(i)}=\frac{w_i}{Y_i}\cdot \vec{r}^{(i)}=w_i \cdot \left(\D^{-a}\A\D^{-b} \right)^i \cdot \vec{x}$. 
\end{definition}

% \begin{definition}\label{def:RQ-relation} [{\bf residue and reserve vectors}]
% At level $i$, the residue vector $\vec{r}^{(i)}$ and reserve vector $\vec{q}^{(i)}$ are defined as follows:
% \begin{itemize}
%   \item $\vec{r}^{(i)}=Y_i\cdot \left(\bm{D}^{-a}\cdot \bm{A} \cdot \bm{D}^{-b} \right)^i \cdot \vec{x};$
%   \item $\vec{q}^{(i)}=\frac{w_i}{Y_i}\cdot \vec{r}^{(i)}=w_i \cdot \left(\bm{D}^{-a}\cdot \bm{A} \cdot \bm{D}^{-b} \right)^i \cdot \vec{x}, $
% \end{itemize}
% \end{definition}

%Intuitively, for each node $u\in V$ and level $i \ge 0$, the residue $\vec{r}^{(i)}(u)$ denotes the energy to be propagated to $u$'s neighbors in the next level, and the reserve $\vec{q}^{(i)}(u)$ denotes the energy that stays at node $u$ in level $i$. 
Intuitively, for each node $u\in V$ and level $i \ge 0$, the residue $\vec{r}^{(i)}(u)$ denotes the probability mass to be distributed to node $u$ at level $i$, and the reserve $\vec{q}^{(i)}(u)$ denotes the probability mass that will stay at node $u$ in level $i$ permanently. 
By Definition~\ref{def:RQ-relation}, the graph propagation equation~\eqref{eqn:pi_gen} can be expressed as $\vec{\pi}=\sum_{i=0}^\infty \vec{q}^{(i)}.$
Furthermore, the residue vector $\vec{r}^{(i)}$ satisfies the following recursive formula:
%\vspace{-2mm}
\begin{align}\label{eqn:iteration}
\vspace{-8mm}
	\vec{r}^{(i+1)}=\frac{Y_{i+1}}{Y_i}\cdot\left(\D^{-a}\A \D^{-b} \right)\cdot \vec{r}^{(i)}. 
%\vspace{-8mm}
\end{align} 
%\vspace{-1mm}
%$\vec{r}^{(i+1)}=\frac{Y_{i+1}}{Y_i}\cdot\left(\bm{D}^{-a}\cdot \bm{A} \cdot \bm{D}^{-b} \right)\cdot \vec{r}^{(i)}$. 
We also observe that the reserve vector  $\vec{q}^{(i)}$ can be derived from the residue vector $\vec{r}^{(i)}$ by $\vec{q}^{(i)}=\frac{w_i}{Y_i}\cdot \vec{r}^{(i)}$.  Consequently, given a predetermined level number $L$, we can compute the graph propagation equation~\eqref{eqn:pi_gen} by iteratively computing the residue vector $\vec{r}^{(i)}$ and reserve vector $\vec{q}^{(i)}$ for $i=0,1,...,L$.

\begin{algorithm}[t]%[ht]
\begin{small}
    \caption{Basic Propagation Algorithm\label{alg:AGP-deter}}
	\KwIn{Undirected graph $G=(V,E)$, graph signal vector $\vec{x}$, weights $w_i$,  number of levels $L$\\}
	\KwOut{the estimated propagation vector $\hat{\vec{\pi}}$\\}
	$\vec{r}^{(0)} \gets \vec{x}$\;
	%$\hat{\pi} \gets w_0\cdot \hat{P}^{(0)}$\;
	\For{$i=0$ to $L-1$}{

		\For{each $u \in V$ with nonzero $\vec{r}^{(i)}(u)$}{
			\For{each $v\in N_u$ }{
				$\vec{r}^{(i+1)}(v) \gets \vec{r}^{(i+1)}(v) + \left(\frac{Y_{i+1}}{Y_i} \right) \cdot \frac{\vec{r}^{(i)}(u)}{d_v^{a} \cdot d_u^b}$;
			}
			$\vec{q}^{(i)}(u) \gets \vec{q}^{(i)}(u)+\frac{w_i}{Y_i}\cdot \vec{r}^{(i)}(u)$\;
		}
		$\hat{\vec{\pi}} \gets \hat{\vec{\pi}} +\vec{q}^{(i)}$ and empty $\vec{r}^{(i)},\vec{q}^{(i)}$\;
	}	
	 $\vec{q}^{(L)}=\frac{w_L}{Y_L}\cdot \vec{r}^{(L)}$ and $\hat{\vec{\pi}} \gets \hat{\vec{\pi}} +\vec{q}^{(L)}$\;
	\Return $\hat{\vec{\pi}}$\;
	%\vspace{-3mm}
\end{small}
\end{algorithm}

Algorithm~\ref{alg:AGP-deter} illustrates the pseudo-code of the basic iterative propagation algorithm.  We first set $\vec{r}^{(0)}=\vec{x}$ (line 1). 
%For $i$ from $0$ to $L-1$, we compute $\vec{r}^{(i+1)}=\frac{Y_{i+1}}{Y_i}\cdot\left(\D^{-a}\A\D^{-b} \right)\cdot \vec{r}^{(i)}$ by pushing an energy of $\left(\frac{Y_{i+1}}{Y_i} \right) \cdot \frac{\vec{r}^{(i)}(u)}{d_v^{a} \cdot d_u^b}$ to each neighbor $v$ of each node $u$ (lines 2-5). 
For $i$ from $0$ to $L-1$, we compute $\vec{r}^{(i+1)}=\frac{Y_{i+1}}{Y_i}\cdot\left(\D^{-a}\A\D^{-b} \right)\cdot \vec{r}^{(i)}$ by pushing the probability mass  $\left(\frac{Y_{i+1}}{Y_i} \right) \cdot \frac{\vec{r}^{(i)}(u)}{d_v^{a} \cdot d_u^b}$ to each neighbor $v$ of each node $u$ (lines 2-5). 
Then, we set $\vec{q}^{(i)}=\frac{w_i}{Y_i}\cdot \vec{r}^{(i)}$ (line 6), and aggregate $\vec{q}^{(i)}$ to $\hat{\vec{\pi}}$ (line 7). We also empty $\vec{r}^{(i)},\vec{q}^{(i)}$ to save memory. After all $L$ levels are processed, we transform the residue of level $L$ to the reserve vector by updating $\hat{\vec{\pi}}$ accordingly (line 8). We return $\hat{\vec{\pi}}$  as an estimator for the graph propagation vector $\vec{\pi}$ (line 9). 

Intuitively, each iteration of Algorithm~\ref{alg:AGP-deter}  computes the matrix-vector multiplication $\vec{r}^{(i+1)}=\frac{Y_{i+1}}{Y_i}\cdot\left(\D^{-a}\A\D^{-b}\right)\cdot \vec{r}^{(i)}$, where $\D^{-a}\A\D^{-b}$ is an  $n\times n$ sparse matrix with $m$ non-zero entries. Therefore, the cost of each iteration of Algorithm~\ref{alg:AGP-deter}  is $O(m)$. To achieve the relative error guarantee in Definition~\ref{def:pro-relative}, we need to set $L=O\left(\log{\frac{1}{\delta}} \right)$, and thus the total cost becomes $O\left(m\cdot \log{\frac{1}{ \delta}}\right)$.  Due to the logarithmic dependence on $\delta$, we use Algorithm~\ref{alg:AGP-deter} to compute high-precision proximity vectors as the ground truths in our experiments. 
%Finally, in the setting of Graph Neural Network, there are $d$ different graph signals $\vec{x}$ to be propagated, where $d$ is the dimension of the node feature vector. 
However, the linear dependence on the number of edges $m$ limits the scalability of Algorithm~\ref{alg:AGP-deter} on large graphs. In particular, in the setting of Graph Neural Network, we treat each column of the feature matrix $\X \in \mathcal{R}^{n \times d}$ as the graph signal $\bm{x}$ to do the propagation. 
%However, in the setting of Graph Neural Network, we treat each column of the feature matrix $\X \in \mathcal{R}^{n \times d}$ as the graph signal $\bm{x}$ to do the propagation. 
Therefore, Algorithm~\ref{alg:AGP-deter} costs $O\left(md \log  \frac{1}{\delta}\right)$ to compute the representation matrix $\mathbf{Z}$. Such high complexity limits the scalability of the existing GNN models.

\section{Randomized PROPAGATION}
\label{sec:RPA}

% \begin{quote}
% 	{\em "Is there a more general way to introduce randomization into graph propagation and reduce the time complexity?"}
% %	, like the $O(m \cdot \log \frac{1}{\lmd})$-bound of $\powitr$?}
% \end{quote}

\header{\bf A failed attempt: pruned propagation.}
The $O\left(m \log \frac{1}{\delta}\right)$ running time is undesirable in many applications. To improve the scalability of the basic propagation algorithm, a simple idea is to prune the nodes with small residues in each iteration. This approach has been widely adopted in local clustering methods such as Nibble and PageRank-Nibble~\cite{FOCS06_FS}. In general, there are two schemes to prune the nodes: 1) we can ignore a node $u$ if its  residue $\hat{\vec{r}}^{(i)}(u)$ is smaller than some threshold $\varepsilon$ in line 3 of Algorithm~\ref{alg:AGP-deter}, or 2) in line 4 of Algorithm~\ref{alg:AGP-deter}, we can somehow ignore an edge $(u,v)$ if $\left(\frac{Y_{i+1}}{Y_i} \right) \cdot \frac{\vec{r}^{(i)}(u)}{d_v^{a} \cdot d_u^b}$, the  residue to be propagated from  
$u$ to $v$, is smaller than some  threshold $\varepsilon'$. Intuitively, both pruning schemes can reduce the number of operations in each iteration.

However, as it turns out, the two approaches suffer from either unbounded error or large time cost. More specifically, consider the toy graph shown in Figure~\ref{fig:special-case}, on which the goal is to estimate $\vec{\pi}=\left(\mathbf{A}\mathbf{D}^{-1} \right)^2 \cdot \vec{e}_s $, the transition probability vector of a 2-step random walk from node $s$. It is easy to see that $\vec{\pi}(v)=1/2, \vec{\pi}(s) = 1/2$, and $\vec{\pi}(u_i)=0, i=1,\ldots, n$. 
%We focus on the approximation quality of $\vec{\pi}(v)$. In particular, we set $\delta = 1/4$ so that the approximate propagation algorithm has to return a constant approximation of $\vec{\pi}(v)$. 
By setting the relative error threshold $\delta$ as a constant (e.g. $\delta = 1/4$), the approximate propagation algorithm has to return a constant approximation of $\vec{\pi}(v)$. 
\begin{figure}[t]%[h]
	\begin{small}
		\centering
		%\vspace{-4mm}
		%    \begin{footnotesize}
		\begin{tabular}{c}
			%\multicolumn{4}{c}{\hspace{-4mm} \includegraphics[height=5mm]{./Figs/legend_large.eps}} \vspace{-1mm} \\
			%\hspace{-3mm} 
			\includegraphics[height=40mm]{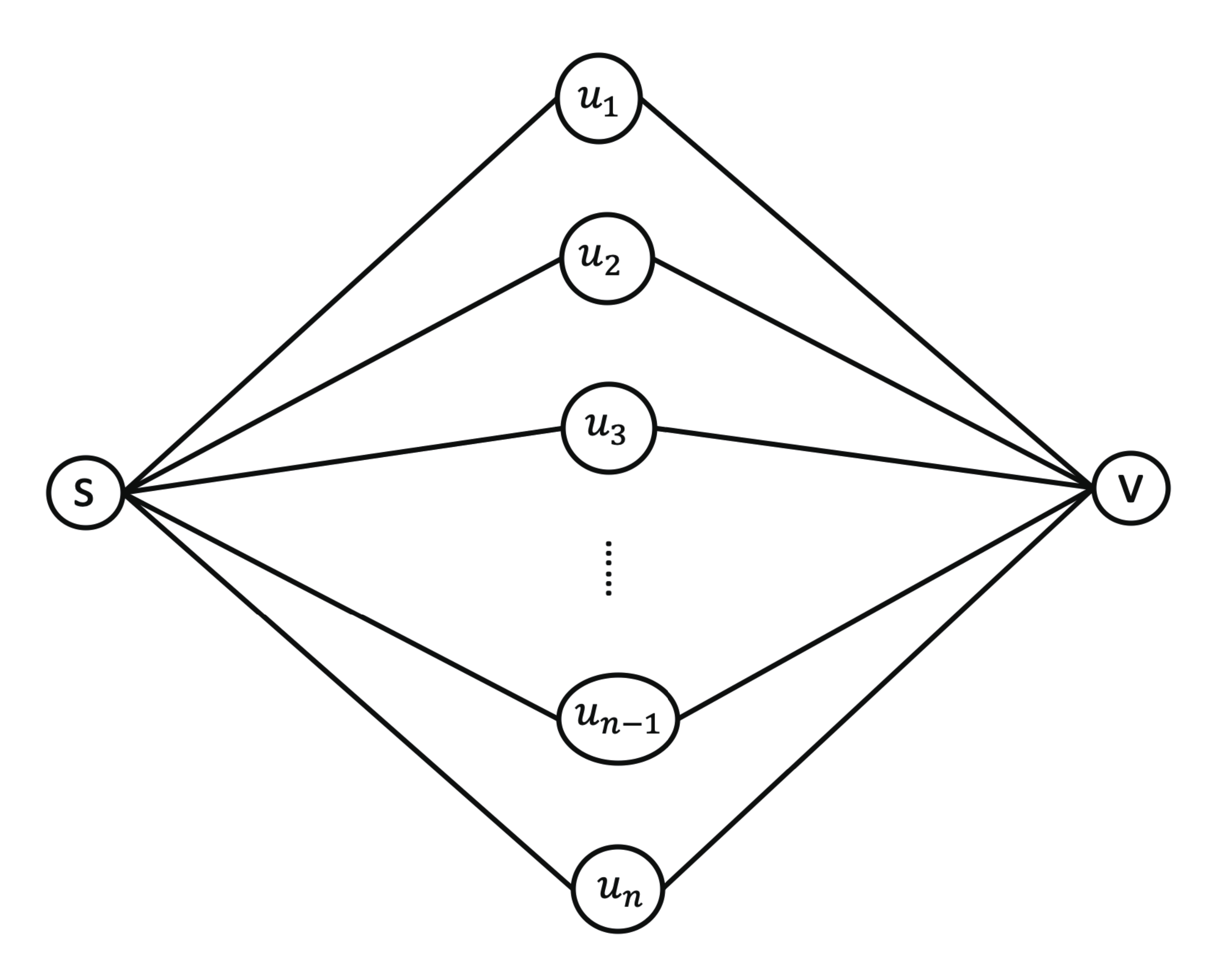} 
		\end{tabular}
		\vspace{-5mm}
		\caption{A bad-case graph for pruned propagation.} 
		\label{fig:special-case}
		\vspace{-3mm}
	\end{small}
	%\vspace{-2mm}
\end{figure}
We consider the first iteration, which pushes the residue $\vec{r}^{(0)}(s)=1$ to $u_1, \ldots, u_n$. If we adopt the first pruning scheme that performs push on $s$ when the residue is large, then we will have to visit all $n$ neighbors of $s$, leading to an intolerable time cost of $O(n)$. On the other hand, we observe that the residue transforms from $s$ to any neighbor $u_i$ is $ \frac{\vec{r}^{(0)}(s)}{d_s} = \frac{1}{n}$. Therefore, if we adopt the second pruning scheme which only performs pushes on edges with large residues to be transformed to, we will simply ignore all pushes from $s$ to $u_1, \ldots, u_n$ and make the incorrect estimation that $\hat{\vec{\pi}}(v)=0$. The problem becomes worse when we are dealing with the general graph propagation equation~\eqref{eqn:pi_gen}, where the Laplacian parameters $a$ and $b$ in the transition matrix $\mathbf{D}^{-a} \mathbf{A}  \mathbf{D}^{-b}$ may take arbitrary values. For example, to the best of our knowledge, no sub-linear approximate algorithm exists for Katz index where $a=b=0$.

% Consequently, a natural question arises: {\em is there a more general way to introduce randomization into graph propagation and reduce the time complexity?}

% In the specific application of estimating the transition probability vector, the above dilemma can be solved by the Monte Carlo method, which simulates a number of 2-step-random-walks from $s$ and uses the percentage of random walks that reaches $v$ as an estimator of $\pi(v)$. The standard probabilistic analysis shows that one only needs to simulate a constant number of random walks to obtain a constant approximation of $\pi(v)$. Therefore, by introducing randomization into the propagation, we can bring the time complexity from $O(n)$ to $O(1)$ and still make the correct estimation with high probability. However, the Monte Carlo method requires that the propagation matrix to be $\mathbf{A} \cdot \mathbf{D}^{-1}$, which is merely a special case of the general graph propagation. Consequently, a natural question arises: is there a more general way to introduce randomization into graph propagation and reduce the time complexity?

%\vspace{-1mm}
%\subsection{Randomized propagation} 
\header{\bf Randomized propagation.} We solve the above dilemma by presenting a simple {\em randomized} propagation algorithm that achieves both theoretical approximation guarantee and near-optimal running time complexity. Algorithm~\ref{alg:AGP-RQ} illustrates the pseudo-code of the Randomized Propagation Algorithm, which only differs from Algorithm~\ref{alg:AGP-deter} by a few lines. Similar to Algorithm~\ref{alg:AGP-deter},  Algorithm~\ref{alg:AGP-RQ} takes in an undirected graph $G=(V,E)$, a graph signal vector $\vec{x}$, a level number $L$ and a weight sequence $w_i$ for $i\in [0,L]$. In addition, Algorithm~\ref{alg:AGP-RQ}  takes in an extra parameter $\varepsilon$, which specifies the relative error guarantee. As we shall see in the analysis, $\varepsilon$ is roughly of  the same order as the relative error threshold $\delta$ in Definition~\ref{def:pro-relative}. Similar to Algorithm~\ref{alg:AGP-deter}, we start with $\vec{\hat{r}}^{(0)} = \vec{x}$ and iteratively perform propagation through level 0 to level $L$. Here we use $\vec{\hat{r}}^{(i)}$ and $\vec{\hat{q}}^{(i)}$ to denote the estimated residue and reserve vectors at level $i$, respectively. 
The key difference is that, on a node $u$ with non-zero residue $\hat{\vec{r}}^{(i)}(u)$, instead of pushing the residue to the whole neighbor set $N_u$, we only perform pushes to the neighbor $v$ with small degree $d_v$. More specifically, for each neighbor $v\in N(u)$ with degree $d_v \le \left( \frac{1}{\varepsilon} \cdot\frac{Y_{i+1}}{Y_i} \cdot \frac{\hat{\vec{r}}^{(i)}(u)}{d_u^b}\right)^{1/a}$, we increase $v$'s residue by $\frac{Y_{i+1}}{Y_i} \cdot \frac{\hat{\vec{{r}}}^{(i)}(u)}{d_v^a\cdot d_u^b}$, which is the same value as in Algorithm~\ref{alg:AGP-deter}. We also note that the condition  $d_v \le \left( \frac{1}{\varepsilon} \cdot\frac{Y_{i+1}}{Y_i} \cdot \frac{\hat{\vec{r}}^{(i)}(u)}{d_u^b}\right)^{1/a}$ is equivalent to $ \frac{Y_{i+1}}{Y_i} \cdot \frac{\hat{\vec{{r}}}^{(i)}(u)}{d_v^a\cdot d_u^b} > \varepsilon$, which means we push the residue from $u$ to $v$ only if it is larger than $\varepsilon$. For the remaining nodes in $N_u$, we sample each neighbor $v\in N_u$ with probability $p_v= \frac{1}{\varepsilon}\cdot \frac{Y_{i+1}}{Y_i}\cdot \frac{\hat{\vec{r}}^{(i)}(u)}{d_v^a\cdot d_u^b}$. Once a node $v$ is sampled, we increase the residue of $v$ by $\varepsilon$. The choice of $p_v$ is to ensure that $p_v\cdot \varepsilon$, the expected residue increment of $v$, equals to $\frac{Y_{i+1}}{Y_i} \cdot \frac{\hat{\vec{r}}^{(i)}(u)}{d_v^a\cdot d_u^b}$, the true residue increment if we perform the actual propagation from $u$ to $v$ in  Algorithm~\ref{alg:AGP-deter}.

\begin{algorithm}[t]
	\caption{Randomized Propagation Algorithm\label{alg:AGP-RQ}}
	\KwIn{undirected graph $G=(V,E)$, graph signal vector $\vec{x}$ with $\|\vec{x}\|_1 \le 1$, weighted sequence $w_i(i=0,1,...,L)$, error parameter $\varepsilon$, number of levels $L$\\}
	\KwOut{the estimated propagation vector $\hat{\vec{\pi}}$\\}
	%$L\gets \log{\frac{1}{\varepsilon}}$\;
	%$\vec{\hat{r}}^{(i)} \gets \vec{0}$, $\vec{\hat{q}}^{(i)} \gets \vec{0}$, for $i=0,1,...,L$\;
	$\vec{\hat{r}}^{(0)} \gets \vec{x}$\;
	%$\hat{\pi} \gets w_0\cdot \hat{P}^{(0)}$\;
	\For{$i=0$ to $L-1$}{
		\For{each $u \in V$ with non-zero residue $\hat{\vec{r}}^{(i)}(u)$}{
			\For{each $v\in N_u$ and $d_v \hspace{-0.5mm} \le \hspace{-0.5mm} \left(\frac{1}{\varepsilon} \cdot\frac{Y_{i+1}}{Y_i} \cdot \frac{\hat{\vec{r}}^{(i)}(u)}{d_u^b}\right)^{\frac{1}{a}}$}{
				$\vec{\hat{r}}^{(i+1)}(v) \gets \vec{\hat{r}}^{(i+1)}(v)+ \frac{Y_{i+1}}{Y_i} \cdot \frac{\hat{\vec{r}}^{(i)}(u)}{d_v^a\cdot d_u^b} $;
			}
			%$ran \gets rand(0,1)$\;
			%\For{each $v\in N(u)$ and $\frac{\hat{R}^{(i)}(u)}{\varepsilon} \cdot \left( 1- \frac{w_i}{Y_i}\right) < d^{1-r}_v\cdot d^r_u \le \frac{\hat{R}^{(i)}(u)}{ran \cdot \varepsilon} \cdot \left( 1- \frac{w_i}{Y_i}\right)$ }{
			{\bf Subset Sampling}: Sample each remaining neighbor $v\in N_u$ with probability $p_v = \frac{1}{\varepsilon}\cdot \frac{Y_{i+1}}{Y_i}\cdot \frac{\hat{\vec{r}}^{(i)}(u)}{d_u^b} \cdot \frac{1}{d_v^a}$\;
			\For{each sampled neighbor $v\in N(u)$}{
			    %Sample $v$ with probability $\frac{1}{\varepsilon}\cdot \frac{Y_{i+1}}{Y_i}\cdot %\frac{\hat{\vec{r}}^{(i)}(u)}{d_v^a\cdot d_u^b}$\;
			    %\If{$v$ is sampled}{$\vec{\hat{r}}^{(i+1)}(v) \gets \vec{\hat{r}}^{(i+1)}(v)+ \varepsilon$;}
			    $\vec{\hat{r}}^{(i+1)}(v) \gets \vec{\hat{r}}^{(i+1)}(v)+ \varepsilon$\;
			}
			$\vec{\hat{q}}^{(i)}(u) \gets \vec{\hat{q}}^{(i)}(u)+\frac{w_i}{Y_i}\cdot \hat{\vec{r}}^{(i)}(u)$\;
			%$\hat{\vec{r}}^{(i)}(u) \gets 0$ \;
		}
		$\hat{\vec{\pi}} \gets \hat{\vec{\pi}} +\hat{\vec{q}}^{(i)}$ and empty $\hat{\vec{r}}^{(i)},\hat{\vec{q}}^{(i)}$\;
	}	
	 $\vec{q}^{(L)}=\frac{w_L}{Y_L}\cdot \vec{r}^{(L)}$ and $\hat{\vec{\pi}} \gets \hat{\vec{\pi}} +\vec{q}^{(L)}$\;
	\Return $\vec{\hat{\pi}}$\;
%\vspace{-1mm}
\end{algorithm}

There are two key operations in Algorithm~\ref{alg:AGP-RQ}. First of all, we need to access the neighbors with small degrees. Secondly, we need to sample each (remaining) neighbor $v \in N_u$ according to some probability $p_v$. Both operations can be supported by scanning over the neighbor set $N_u$. However, the cost of the scan is asymptotically the same as performing a full propagation on $u$ (lines 4-5 in Algorithm~\ref{alg:AGP-deter}), which means Algorithm~\ref{alg:AGP-RQ} will lose the benefit of randomization and essentially become the same as Algorithm~\ref{alg:AGP-deter}. 

\header{\bf Pre-sorting adjacency list by degrees.} To access the neighbors with small degrees, we can pre-sort each adjacency list $N_u$ according to the degrees. More precisely, we assume that $N_u = \{v_1, \ldots,v_{d_u}\}$ is stored in a way that $d_{v_1} \le \ldots \le d_{v_{d_u}}$. Consequently, we can implement lines 4-5 in Algorithm~\ref{alg:AGP-RQ} by sequentially scanning through $N_u = \{v_1, \ldots,v_{d_u}\}$ and stopping at the first $v_j$ such that $d_{v_j} > \left( \frac{1}{\varepsilon} \cdot\frac{Y_{i+1}}{Y_i} \cdot \frac{\hat{\vec{r}}^{(i)}(u)}{d_u^b}\right)^{1/a}$. With this implementation, we only need to access the neighbors with degrees that exceed the threshold. We also note that we can pre-sort the adjacency lists when reading the graph into the memory, without increasing the asymptotic cost. In particular, we construct a tuple $(u,v,d_v)$ for each edge $(u,v)$ and use counting sort to sort $(u,v,d_v)$ tuples in the ascending order of $d_v$. Then we scan the tuple list. For each $(u,v,d_v)$, we append $v$ to the end of $u$'s
adjacency list $N_u$. Since each $d_v$ is bounded by $n$, and
there are $m$ tuples, the cost of counting sort is bounded by
$O(m+n)$, which is asymptotically the same as reading the graphs.

% Therefore, we need to solve the following two problems: 1) How can we access the small-degree neighbors  without scanning over the neighbor set $N_u$? 2) How can we sample each neighbor $v \in N_u$ according to probability $p_v$ without touching all neighbors in $N_u$?

\header{\bf Subset Sampling.} The second problem, however, requires a more delicate solution. Recall that the goal is to sample each neighbor $v_j \in N_u = \{v_1, \ldots, v_{d_u}\}$ according to the probability $p_{v_j} = \frac{Y_{i+1}}{\e \cdot Y_i}\cdot \frac{\hat{\vec{r}}^{(i)}(u)}{d_{v_j}^a \cdot d_u^b}$ without touching all the neighbors in $N_u$. This problem is known as the {\em Subset Sampling problem} and has been solved optimally in~\cite{bringmann2012efficient}. For ease of implementation, we employ a simplified solution: for each node $u\in V$, we partition $u$'s adjacency list $N_u = \{v_1, \ldots, v_{d_u}\}$ into $O(\log n)$ groups, such that the $k$-th group $G_k$ consists of the neighbors $v_j\in N_u$ with degrees $d_{v_j} \in [2^{k},2^{k+1})$. Note that this can be done by simply sorting $N_u$ according to the degrees. Inside the $k$-th group $G_k$, the sampling probability $p_{v_j} = \frac{Y_{i+1}}{\e \cdot Y_i}\cdot \frac{\hat{\vec{r}}^{(i)}(u)}{d_{v_j}^a \cdot d_u^b}$ differs by a factor of at most $2^a \le 2$. Let $p^*$ denote the maximum sampling probability in $G_k$. We generate a random integer $\ell$ according to the Binomial distribution $B(|G_k|, p^*)$, and randomly selected $\ell$ neighbors from $G_k$. For each selected neighbor $v_j$, we reject it with probability $1-p_{v_j}/p^*$. 
%Note that the sampling complexity for $G_k$ is $O\left(\sum_{j\in G_k} p_j+1\right)$. Consequently, the total sampling complexity becomes $O\left(\sum_{k=1}^{\log n} \left(\sum_{j\in G_k} p_j+1 \right) \right)= O\left(\sum_{j=0}^{d_u} p_j + \log n \right)$. 
Hence, the sampling complexity for $G_k$ can be bounded by $O\left(\sum_{j\in G_k} p_{v_j}+1\right)$, where $O(1)$ corresponds to the cost to generate the binomial random integer $\ell$ and $O\left(\sum_{j\in G_k} p_j\right)=O\left(\sum_{j\in G_k} p^*\right)$ denotes the cost to select $\ell$ neighbors from $G_k$. Consequently, the total sampling complexity becomes $O\left(\sum_{k=1}^{\log n} \left(\sum_{j\in G_k} p_j+1 \right) \right)= O\left(\sum_{j=0}^{d_u} p_j + \log n \right)$.
Note that for each subset sampling operation, we need to return $O\left(\sum_{j=0}^{d_u} p_j \right)$ neighbors in expectation, so this complexity is optimal up to the $\log n$ additive term.

 \begin{figure*}[t]
 	\begin{small}
 		\centering
 		%\vspace{-5mm}
 		%    \begin{footnotesize}
 		\begin{tabular}{cccc}
 			%\multicolumn{4}{c}{\hspace{-4mm} \includegraphics[height=5mm]{./Figs/legend_large.eps}} \vspace{-1mm} \\
 			\hspace{-4mm} \includegraphics[height=34mm]{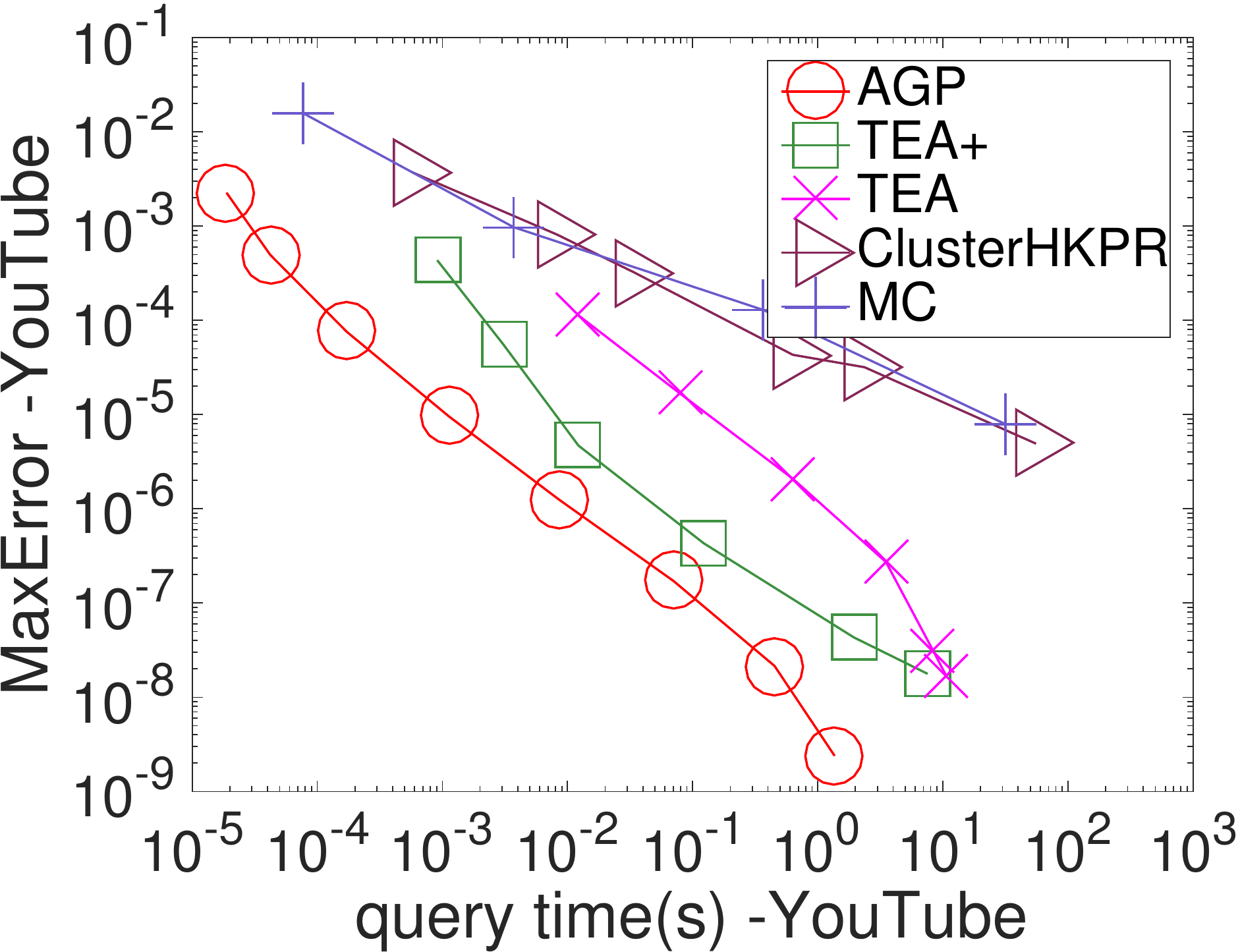} &
 			%\hspace{-3mm} \includegraphics[height=25mm]{./Figs/HKPR-maxerr-query-DB.eps} &
 			\hspace{-4mm} \includegraphics[height=34mm]{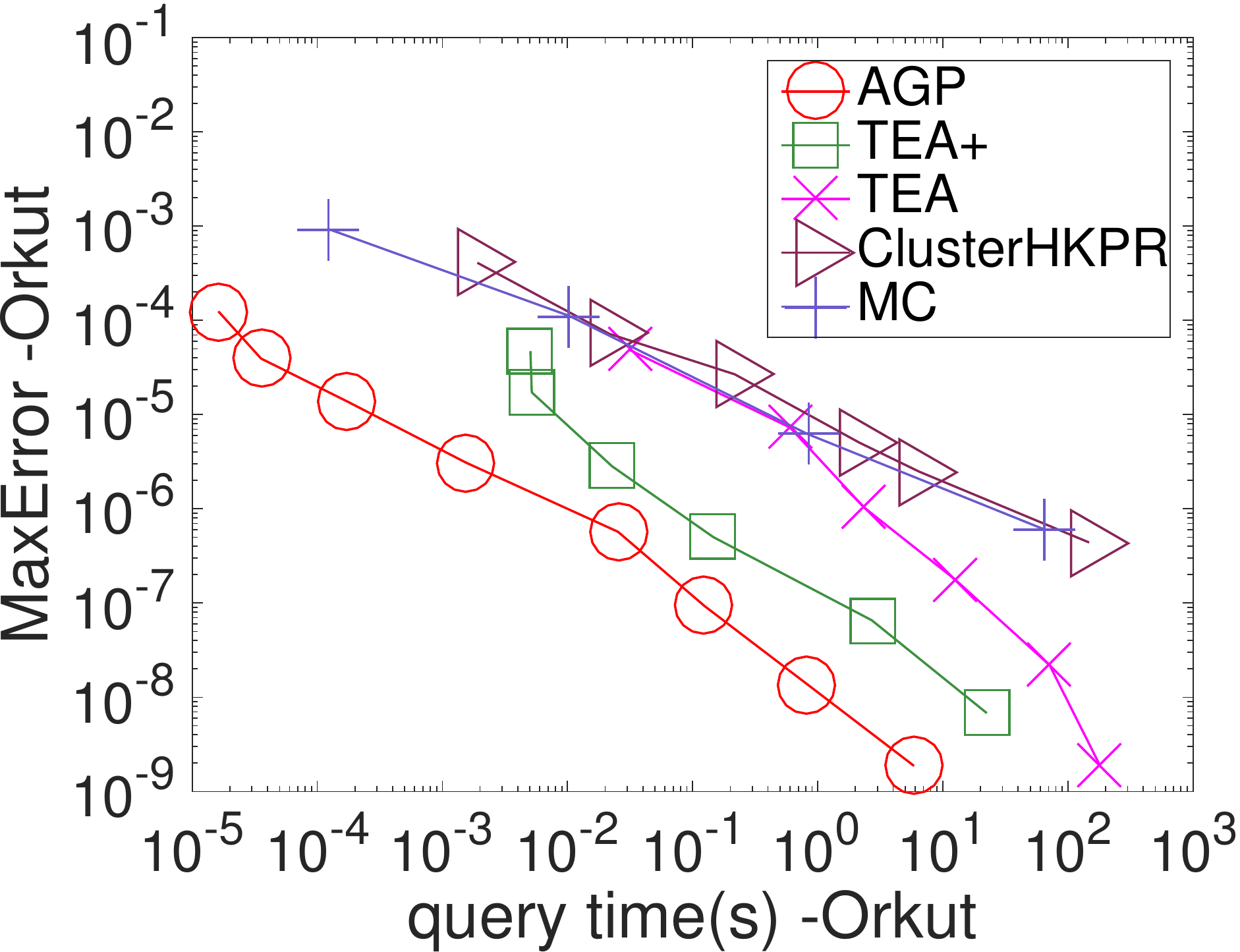} &
 			\hspace{-4mm} \includegraphics[height=34mm]{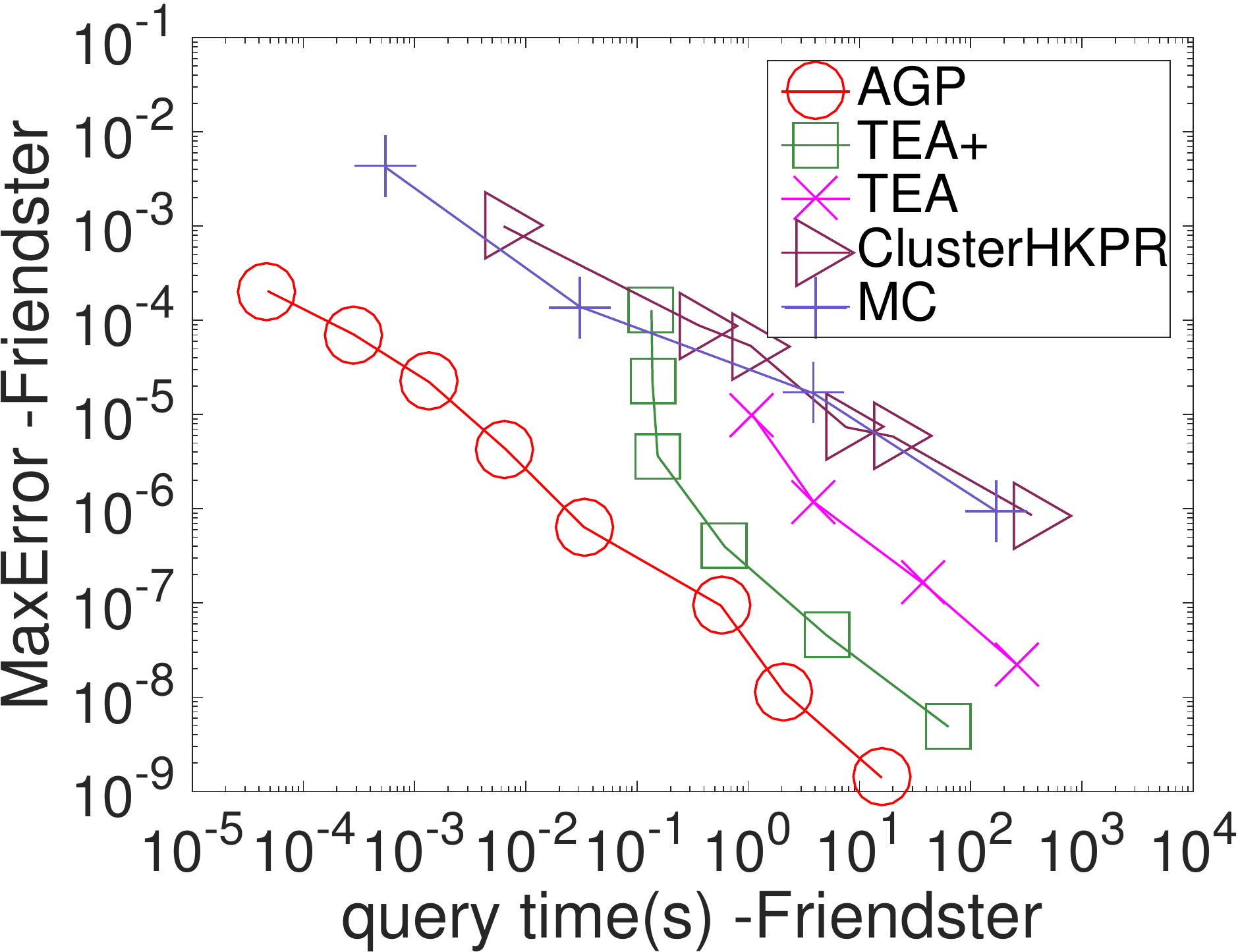} &
 			\hspace{-4mm} \includegraphics[height=34mm]{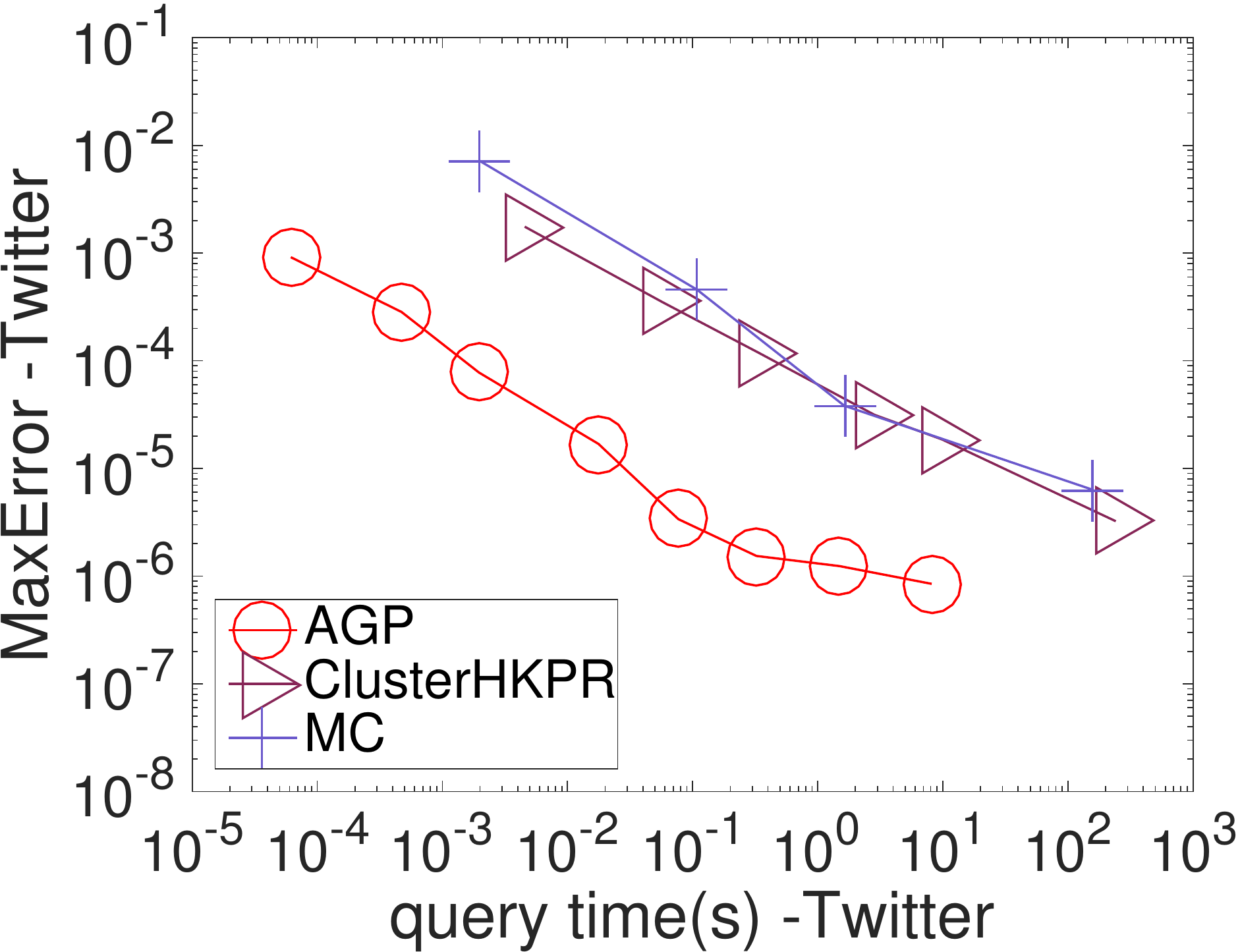} 
 		\end{tabular}
 		\vspace{-5mm}
 		\caption{Tradeoffs between {\em MaxError} and query time in local clustering.}
 		\label{fig:HKPR-maxerror-query}
 		\vspace{-2mm}
 	\end{small}
 \end{figure*}

 \begin{figure}[t]
 	\begin{small}
 		\centering
 		%\vspace{-2mm}
 		%    \begin{footnotesize}
 		\begin{tabular}{cccc}
 			%\multicolumn{4}{c}{\hspace{-4mm} \includegraphics[height=5mm]{./Figs/legend_large.eps}} \vspace{-1mm} \\
 			% 			\hspace{-2mm} \includegraphics[height=34mm]{./Figs/HKPR-conductance-query-YT.eps} &
 			%\hspace{-3mm} \includegraphics[height=25mm]{./Figs/HKPR-conductance-query-DB.eps} &
 			\hspace{-4mm} \includegraphics[height=32mm]{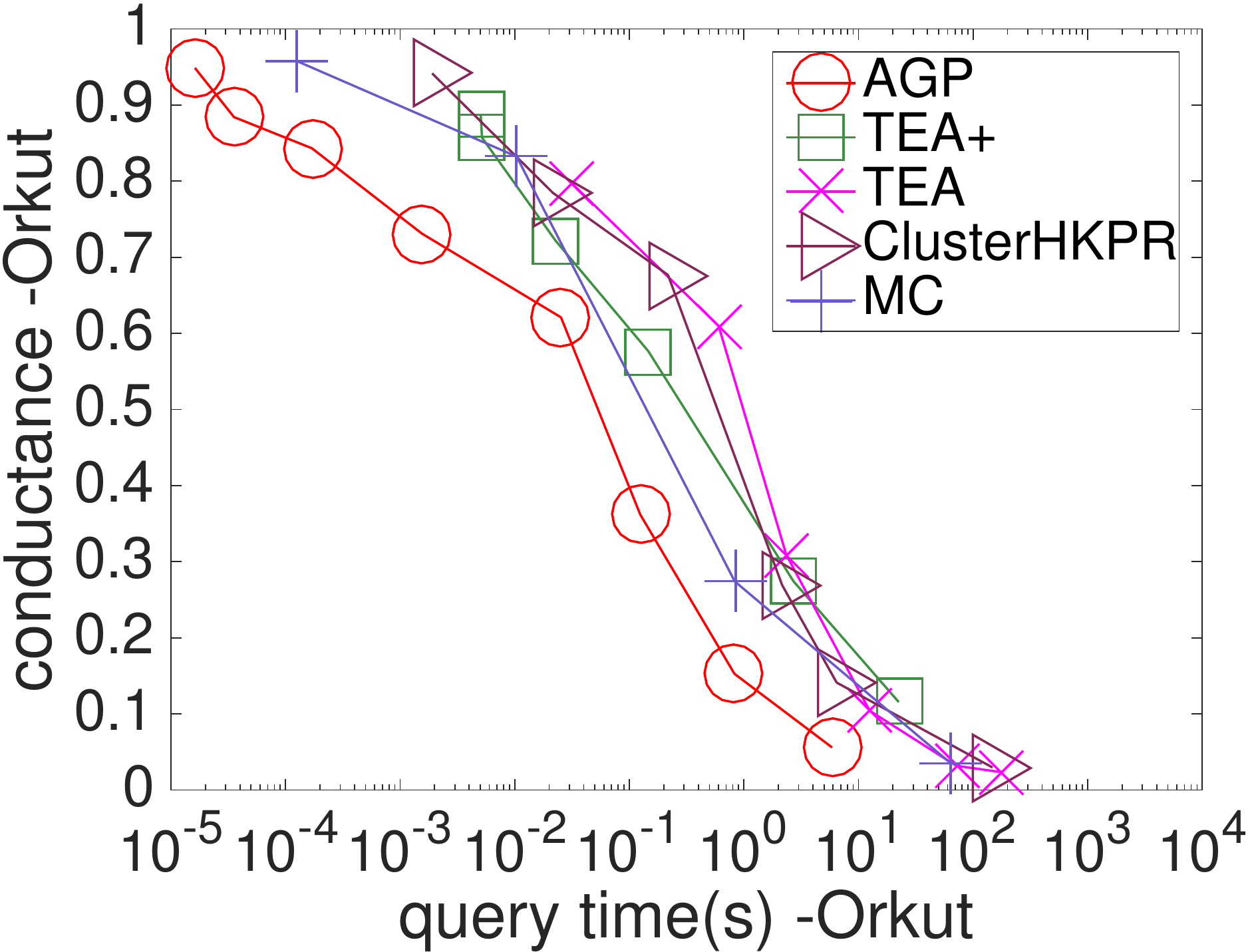} &
 			\hspace{-4mm} \includegraphics[height=32mm]{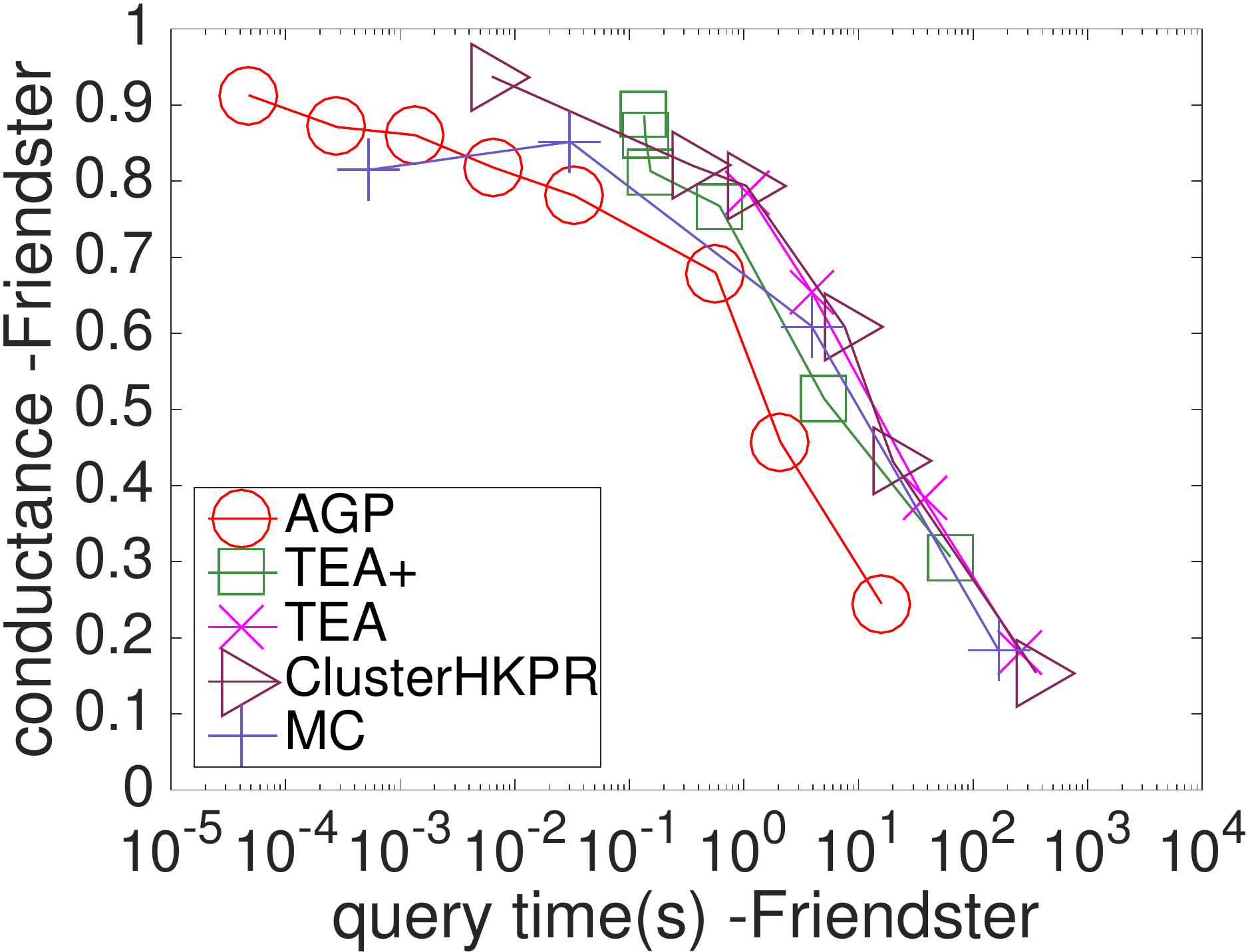} &
 			%\hspace{-4mm} \includegraphics[height=34mm]{./Figs/HKPR-conductance-query-TW.eps} 
 		\end{tabular}
 		\vspace{-5mm}
 		\caption{Tradeoffs between {\em conductance} and query time in local clustering.}
 		\label{fig:HKPR-conductance-query}
 		\vspace{-2mm}
 	\end{small}
 \end{figure}

\header{\bf Analysis.} We now present a series of lemmas that characterize the error guarantee and the running time of Algorithm~\ref{alg:AGP-RQ}. For readability, we only give some intuitions for each lemma and defer the detailed proofs to the technical report~\cite{TechnicalReport}. We first present a lemma to show Algorithm~\ref{alg:AGP-RQ} can return unbiased estimators for the residue and reserve vectors at each level. 
%\vspace{-1mm}
\begin{lemma}
\vspace{-2mm}
	\label{lem:unbiasedness}
	For each node $v\in V$, Algorithm~\ref{alg:AGP-RQ} computes estimators $\vec{\hat{r}}^{(\ell)}(v)$ and $\vec{\hat{q}}^{(\ell)}(v)$ such that $\E\left[ \vec{\hat{q}}^{(\ell)}(v)\right]=\vec{q}^{(\ell)}(v)$ and $\E\left[ \vec{\hat{r}}^{(\ell)}(v)\right]=\vec{r}^{(\ell)}(v)$ holds for $\forall \ell \in \{0,1, 2, ... , L\}$.
\end{lemma}
To give some intuitions on the correctness of Lemma~\ref{lem:unbiasedness}, recall that in lines 4-5 of Algorithm~\ref{alg:AGP-deter}, we add $\frac{Y_{i+1}}{Y_i} \cdot \frac{\hat{\vec{r}}^{(i)}(u)}{d_v^a\cdot d_u^b}$ to each residue $\hat{\vec{r}}^{(i+1)}(v)$ for $\forall v\in N_u$. 
We perform the same operation in Algorithm~\ref{alg:AGP-RQ} for each neighbor $v \in N_u$ with large degree $d_v$. For each remaining neighbor $v\in V$, we add a residue of $ \varepsilon$ to  $\hat{\vec{r}}^{(i+1)}(v)$ with probability $\frac{1}{\varepsilon}\frac{Y_{i+1}}{Y_i} \cdot \frac{\hat{\vec{r}}^{(i)}(u)}{d_v^a\cdot d_u^b}$, leading to an expected increment of $\frac{Y_{i+1}}{Y_i} \cdot \frac{\hat{\vec{r}}^{(i)}(u)}{d_v^a\cdot d_u^b}$. Therefore, Algorithm~\ref{alg:AGP-RQ} computes an unbiased estimator for each residue vector $\vec{r}^{(i)}$, and consequently an unbiased estimator for each reserve vector $\vec{q}^{(i)}$.

In the next Lemma, we bound the variance of the approximate graph propagation vector $\hat{\pi}$, which takes a surprisingly simple form. 
%\vspace{-3mm}
\begin{lemma}
%\vspace{-1mm}
	\label{lem:variance}
    For any node $v \in V$, the variance of $\hat{\vec{\pi}}(v)$ obtained by Algorithm~\ref{alg:AGP-RQ} satisfies 
    %$	\Var \left[ \hat{\vec{\pi}}(v) \right]\le \frac{L(L+1)}{2} \cdot \varepsilon \cdot \vec{\pi}(v). $
    $	\Var \left[ \hat{\vec{\pi}}(v) \right]\le \frac{L(L+1)\varepsilon}{2} \cdot \vec{\pi}(v). $
    % \begin{equation}\nonumber
    % 	\begin{aligned}
    % 		\Var \left[ \hat{\vec{\pi}}(v) \right]\le \frac{L(L+1)}{2} \cdot \varepsilon \cdot \vec{\pi}(v). 
    % 	\end{aligned}
    % \end{equation}
\end{lemma}
%\vspace{-1mm}
Recall that we can set $L=O(\log 1/\varepsilon)$ to obtain a relative error threshold of $\varepsilon$. Lemma~\ref{lem:variance} essentially states that the variance decreases linearly with the error parameter $\varepsilon$. Such property is desirable for bounding the relative error. In particular, for any node $v$ with $\vec{\pi}(v) > 100\cdot \frac{L(L+1)\varepsilon}{2}$, the standard deviation of $\hat{\vec{\pi}}(v)$ is bounded by ${1\over 10}\vec{\pi}(v)$. Therefore, we can set $\varepsilon = \frac{\delta }{50L(L+1)} = \Tilde{O}(\delta)$ to obtain the relative error guarantee in Definition~\ref{def:pro-relative}.
%implies that $\hat{\vec{\pi}}(v) - \vec{\pi}(v) \le  {1\over 10}\vec{\pi}(v)$ 
In particular, we have the following theorem that bounds the expected cost of Algorithm~\ref{alg:AGP-RQ} under the relative error guarantee.

%\vspace{-1mm}
\begin{theorem}\label{thm:RP-error}
Algorithm~\ref{alg:AGP-RQ} achieves an approximate propagation with relative error $\delta$, that is, for any node $v$ with $\vec{\pi}(v)>\delta$, $\left|\vec{\pi}(v)-\hat{\vec{\pi}}(v)\right|\hspace{-1mm} \leq \hspace{-1mm}\frac{1}{10} \cdot \vec{\pi}(v)$. The expected time cost can be bounded by
%\vspace{-2mm}
\begin{equation*}
%\vspace{-1mm}
E[Cost] =	O\left(\frac{L^2}{\delta}\cdot \sum_{i=1}^{L} \left\| Y_i \cdot \left(\mathbf{D}^{-a}\mathbf{A}\mathbf{D}^{-b} \right)^i \cdot \vec{x} \right\|_1\right).
\end{equation*}
\end{theorem}

%To understand the time complexity in Theorem~\ref{thm:RP-error}, note that $ \left\|Y_i \cdot \left(\mathbf{D}^{-a}\mathbf{A}\mathbf{D}^{-b} \right)^i \cdot \vec{x} \right\|_1$ is the summation of the true residues at level $i$.  By the Pigeonhole principle, $ \frac{1}{\delta}\left\| Y_i \cdot \left(\mathbf{D}^{-a}\mathbf{A}\mathbf{D}^{-b} \right)^i \cdot \vec{x} \right\|_1$ is the upper bound of the residues at level $i$ that are larger than $\delta$. This bound is the output size of the propagation at level $i$, which means Algorithm~\ref{alg:AGP-RQ} achieves near optimal time complexity. We defer detailed explanation to the technical report~\cite{TechnicalReport}. 

To understand the time complexity in Theorem~\ref{thm:RP-error}, note that by the Pigeonhole principle, $ \frac{1}{\delta}\sum_{i=1}^L \left\| w_i \cdot \left(\mathbf{D}^{-a}\mathbf{A}\mathbf{D}^{-b} \right)^i \cdot \vec{x} \right\|_1$ is the upper bound of the number of $v\in V$ satisfying $\vec{\pi}(v)\ge \delta$. In the proximity models such as PPR, HKPR and transition probability, this bound is the output size of the propagation, which means Algorithm~\ref{alg:AGP-RQ} achieves near optimal time complexity. We defer the detailed explanation to the technical report~\cite{TechnicalReport}.

Furthermore, we can compare the time complexity of Algorithm~\ref{alg:AGP-RQ} with other state-of-the-art algorithms in specific applications. For example, in the setting of heat kernel PageRank, the goal is to estimate $\vec{\pi}=\sum_{i=0}^\infty e^{-t} \cdot \frac{t^i}{i!}\cdot \left(\mathbf{A}\mathbf{D}^{-1} \right)^i \cdot \vec{e}_s $ for a given node $s$. The state-of-the-art algorithm TEA~\cite{yang2019TEA} computes an approximate HKPR vector $\hat{\vec{\pi}}$ such that for any $\pi(v)>\delta$, $|\vec{\pi}(v)-\hat{\vec{\pi}}(v)| \leq \frac{1}{10} \cdot \vec{\pi}(v)$ holds for high probability. By the fact that $t$ is the a constant and $\Tilde{O}$ is the Big-Oh notation ignoring log factors, the total cost of TEA is bounded by $O\left(t\log n \over \delta \right) = \Tilde{O}\left(1 \over \delta \right)$. On the other hand, in the setting of HKPR, the time complexity of Algorithm~\ref{alg:AGP-RQ} is bounded by 
%\vspace{-2mm}
\begin{equation*}
%\vspace{-2mm}
O\left(\frac{L^2}{\delta}\cdot \sum_{i=1}^{L} \left\| Y_i \cdot \left(\cdot \mathbf{A} \mathbf{D}^{-1} \right)^i \cdot \vec{e}_s \right\|_1 \right) = \frac{L^2}{\delta}\cdot \sum_{i=1}^{L} Y_i \le \frac{L^3}{\delta} = \Tilde{O}\left(1 \over \delta \right).
\end{equation*}
Here we use the facts that $ \left\| \left( \mathbf{A} \mathbf{D}^{-1} \right)^i \cdot \vec{e}_s \right\|_1 =1$ and $Y_i \le 1$. This implies that under the specific application of estimating HKPR, the time complexity of the more generalized  Algorithm~\ref{alg:AGP-RQ} is asymptotically the same as the complexity of TEA. Similar bounds also holds for Personalized PageRank and transition probabilities.

\header{\bf Propagation on directed graph.}
Our generalized propagation structure also can be extended to directed graph by $\pi=\sum_{i=0}^\infty w_i \cdot \left(\D^{-a}\tilde{\A}\D^{-b} \right)^i \cdot \vec{x},$
% \begin{equation}\label{eqn:pi_gen_directed}
% \vspace{-2mm}
% 	\begin{aligned}
% 		\pi=\sum_{i=0}^\infty w_i \cdot \left(\bm{D}^{-a}\cdot \tilde{\bm{A}} \cdot \bm{D}^{-b} \right)^i \cdot \vec{x}, 
% 	\end{aligned}
% \end{equation}
where $\D$ denotes the diagonal out-degree matrix, and $\tilde{\A}$ represents the adjacency matrix or its transition according to specific applications. For PageRank, single-source PPR, HKPR, Katz we set $\tilde{\A}=\A^\top$ with the following recursive equation: 
%\vspace{-2mm}
\begin{equation}\nonumber
%\vspace{-2mm}
	\begin{aligned}
		\vec{r}^{(i+1)}(v)=\sum_{u\in N_{in}(v)} \left(\frac{Y_{i+1}}{Y_i} \right) \cdot \frac{\vec{r}^{(i)}(u)}{d_{out}^{a}(v) \cdot d_{out}^b(u)}. 
	\end{aligned}
\end{equation}
where $N_{in}(v)$ denotes the in-neighbor set of $v$ and $d_{out}(u)$ denotes the out-degree of $u$. For single-target PPR, we set $\tilde{\A}=\A$. 
\section{Experiments} \label{sec:exp}

\begin{table}[t]
	%\vspace{-2mm}
	\centering
	\tblcapup
	\caption{Datasets for local clustering.}
	\vspace{-4mm}
	\tblcapdown
	\begin{small}
		\begin{tabular}{|l|l|r|r|} %p{1.3in}|}
			\hline
			{\bf Data Set} & {\bf Type} & {\bf $\boldsymbol{n}$} & {\bf $\boldsymbol{m}$}	 \\ \hline
			%ca-GrQc (GQ) & undirected & 5,242 & 28,968\\
			%AS-2000(AS) & undirected & 6,474 & 25,144\\
			%CA-HepTh(HT) & undirected & 9,877 & 51,946\\
			%Wikivote (WV) & directed & 7,115 & 103,689\\
			%CA-HepPh (HP) & undirected & 12008 & 236978\\
			% Wiki-Vote(WV)	& 	directed &	7,155	&	103,689 \\
			% {HepTh(HT)}	    & 	undirected &	9,877	&	25,998		\\
			% {AS-Caida(AC)}	    &	directed &	26,475	&	106,762		\\
			% {HepPh(HP)}	        &	directed &	34,546	&	421,578 \\
			% Cnr-2000 (CN) & directed & 325,557 &	3,216,152 \\
			% {Web-Google(WG)} & directed & 875,713 &	5,105,039 \\)
			%As-Skitter (AS) & undirected & 1,696,415	& 11,095,298 \\
			%      {\color{red} In-2004(IN) } & directed & 1,382,908	& 16,917,053
			% \\
			YouTube  & undirected & 1,138,499 & 5,980,886 \\
			%DBLP-Author  & undirected & 5,425,963 & 17,298,032 \\
			% {Com-LiveJournal(CL)} & undirected & 3,997,962	& 34,681,189 \\
			%LiveJournal (LJ) &	directed & 4,847,571 & 68,475,391\\
			%IndoChina (IC)	& 	directed &	7,414,768  &	191,606,827		\\
			Orkut-Links  & undirected & 3,072,441 & 234,369,798 \\		
			%DBpediaLink (DL) & directed & 18,268,992 & 172,183,984 \\
			%WikiLink(WL) & directed & 12,150,976 & 378,142,420 \\
			% { Web-Base}   & { directed} & {
			% 118,142,155} & { 1,019,903,190} \\
			%Web-Base (WB) & directed & 118,142,155& 1,019,903,190 \\
			%It-2004 (IT)	&	directed & 41,290,682 & 1,135,718,909\\	
			Twitter  & directed & 41,652,230& 1,468,364,884 \\
			%SK-2005 (SK) & directed & 50,636,154	& 1,949,412,601 \\
			Friendster   & undirected & 68,349,466 & 3,623,698,684 \\
			%UK-Union (UK) & directed & 133,633,040 & 5,507,679,822 \\
			\hline
		\end{tabular}
	\end{small}
	\label{tbl:datasets}
	%\tbldown
	\vspace{-2mm}
\end{table}

This section experimentally evaluates AGP's performance in two concrete applications: local clustering with heat kernel PageRank and node classification with GNN. Specifically, Section~\ref{subsec:clustering} presents the experimental results of AGP in local clustering. Section~\ref{subsec:GNN} evaluates the  effectiveness of AGP on existing GNN models.

% \begin{figure*}[!t]
% 	\begin{small}
% 		\centering
% 		\vspace{-1mm}
% 		%    \begin{footnotesize}
% 		\begin{tabular}{cccc}
% 			%\multicolumn{4}{c}{\hspace{-4mm} \includegraphics[height=5mm]{./Figs/legend_large.eps}} \vspace{-1mm} \\
% 			\hspace{-2mm} \includegraphics[height=34mm]{./Figs/HKPR-conductance-query-YT.eps} &
% 			%\hspace{-3mm} \includegraphics[height=25mm]{./Figs/HKPR-conductance-query-DB.eps} &
% 			\hspace{-4mm} \includegraphics[height=34mm]{./Figs/HKPR-conductance-query-OL.eps} &
% 			\hspace{-4mm} \includegraphics[height=34mm]{./Figs/HKPR-conductance-query-FR.eps} &
% 			%\hspace{-4mm} \includegraphics[height=34mm]{./Figs/HKPR-conductance-query-TW.eps} 
% 		\end{tabular}
% 		\vspace{-3mm}
% 		\caption{Tradeoffs between {\em conductance} and query time in local clustering.}
% 		\label{fig:HKPR-conductance-query}
% 		\vspace{-1mm}
% 	\end{small}
% \end{figure*}

\begin{figure*}[t]
	\begin{small}
		\centering
		%\vspace{-5mm}
		%    \begin{footnotesize}
		\begin{tabular}{cccc}
			%\multicolumn{4}{c}{\hspace{-4mm} \includegraphics[height=5mm]{./Figs/legend_large.eps}} \vspace{-1mm} \\
			%\hspace{-3mm} \includegraphics[height=25mm]{./Figs/HKPR-conductance-query-DB.eps} &
			\hspace{-4mm} \includegraphics[height=34mm]{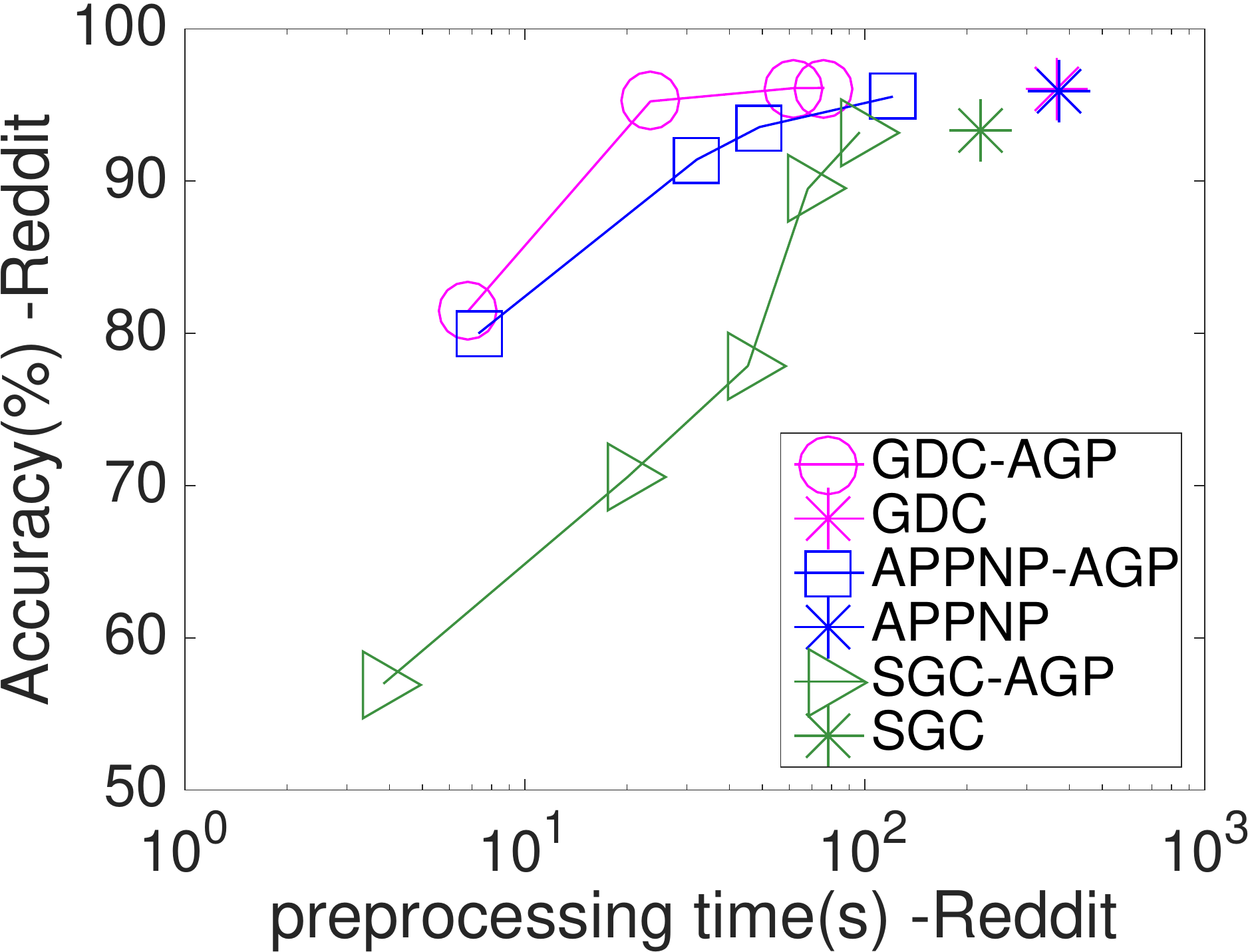} &
			\hspace{-4mm} \includegraphics[height=34mm]{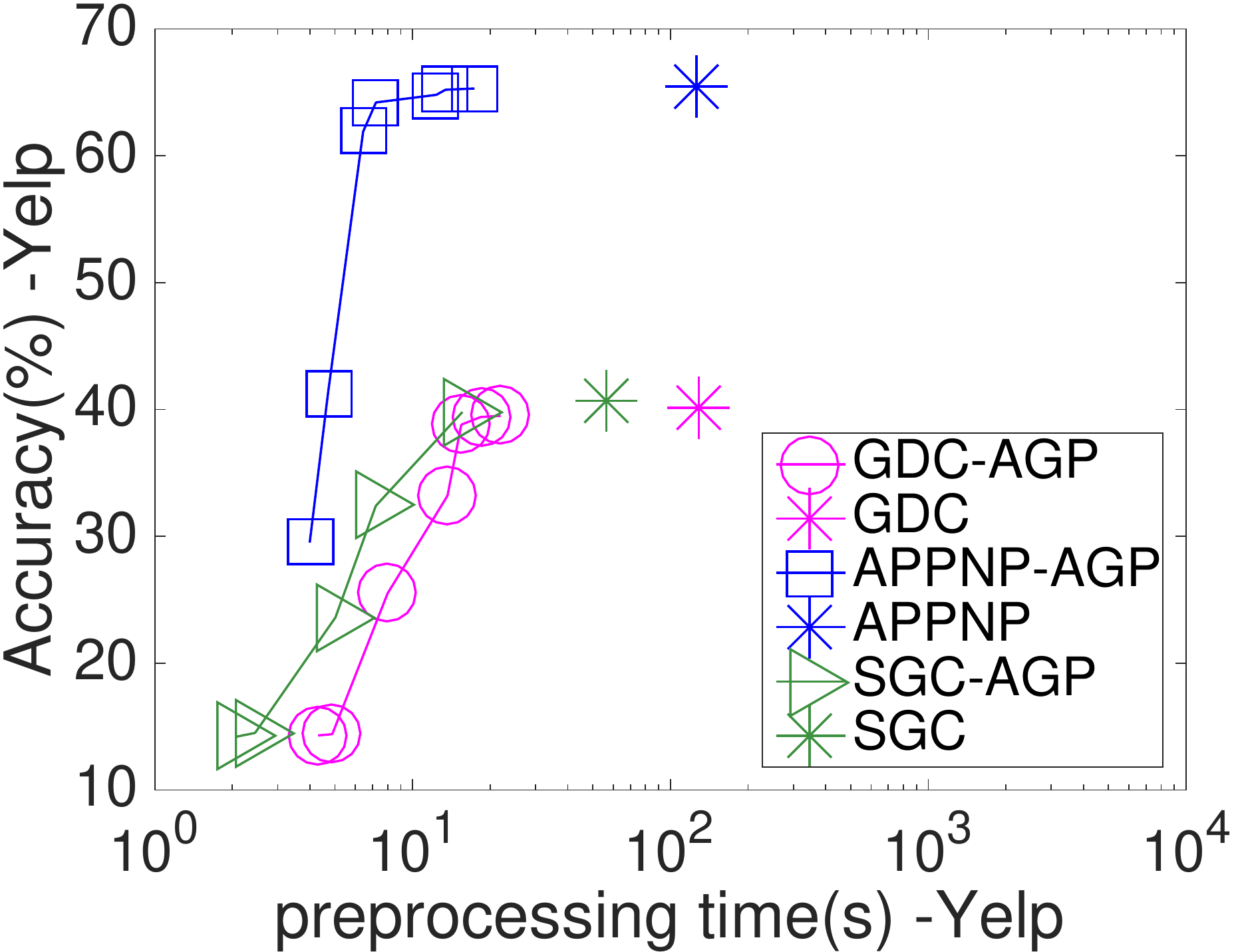} &
			\hspace{-2mm} \includegraphics[height=34mm]{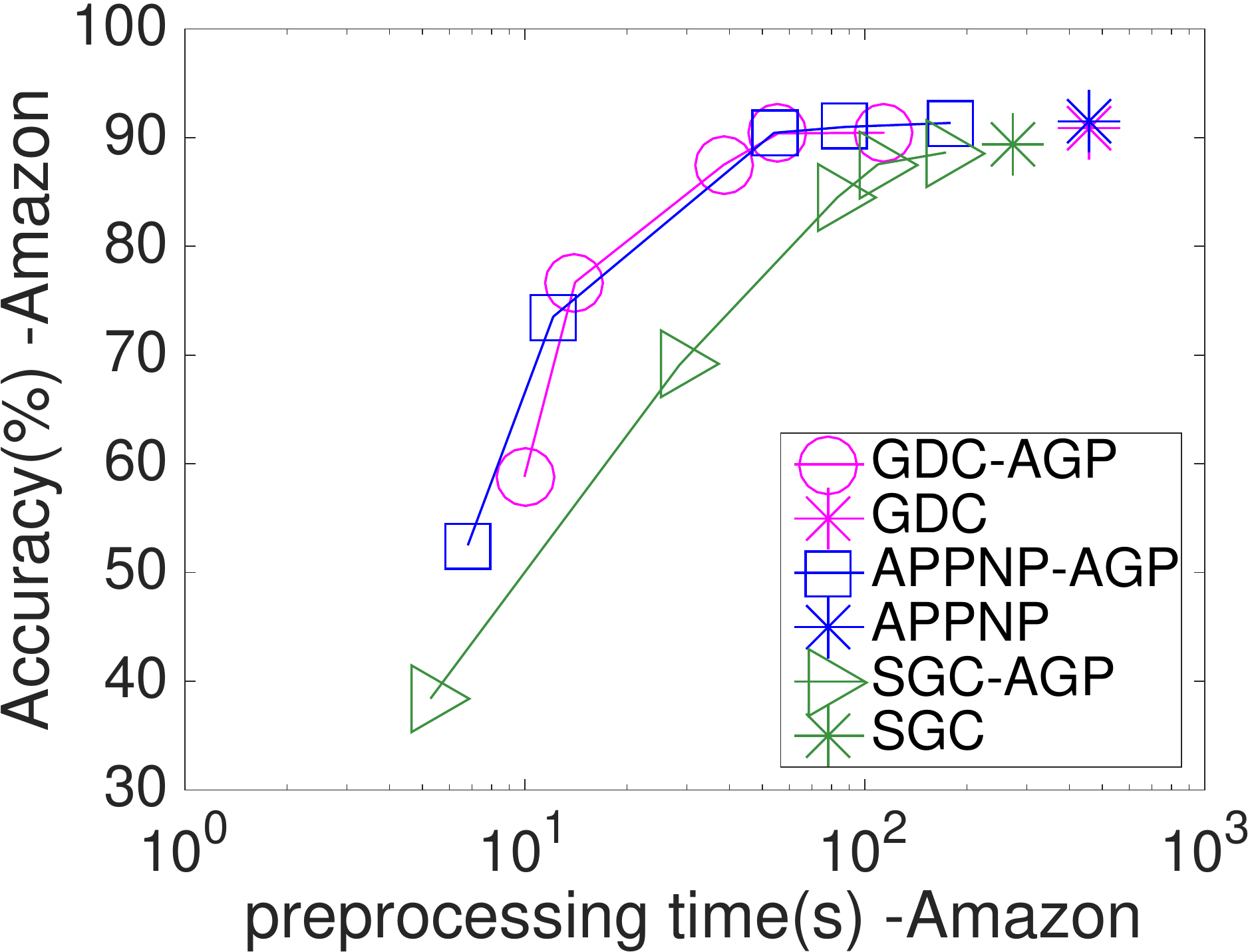} &
			\hspace{-4mm} \includegraphics[height=34mm]{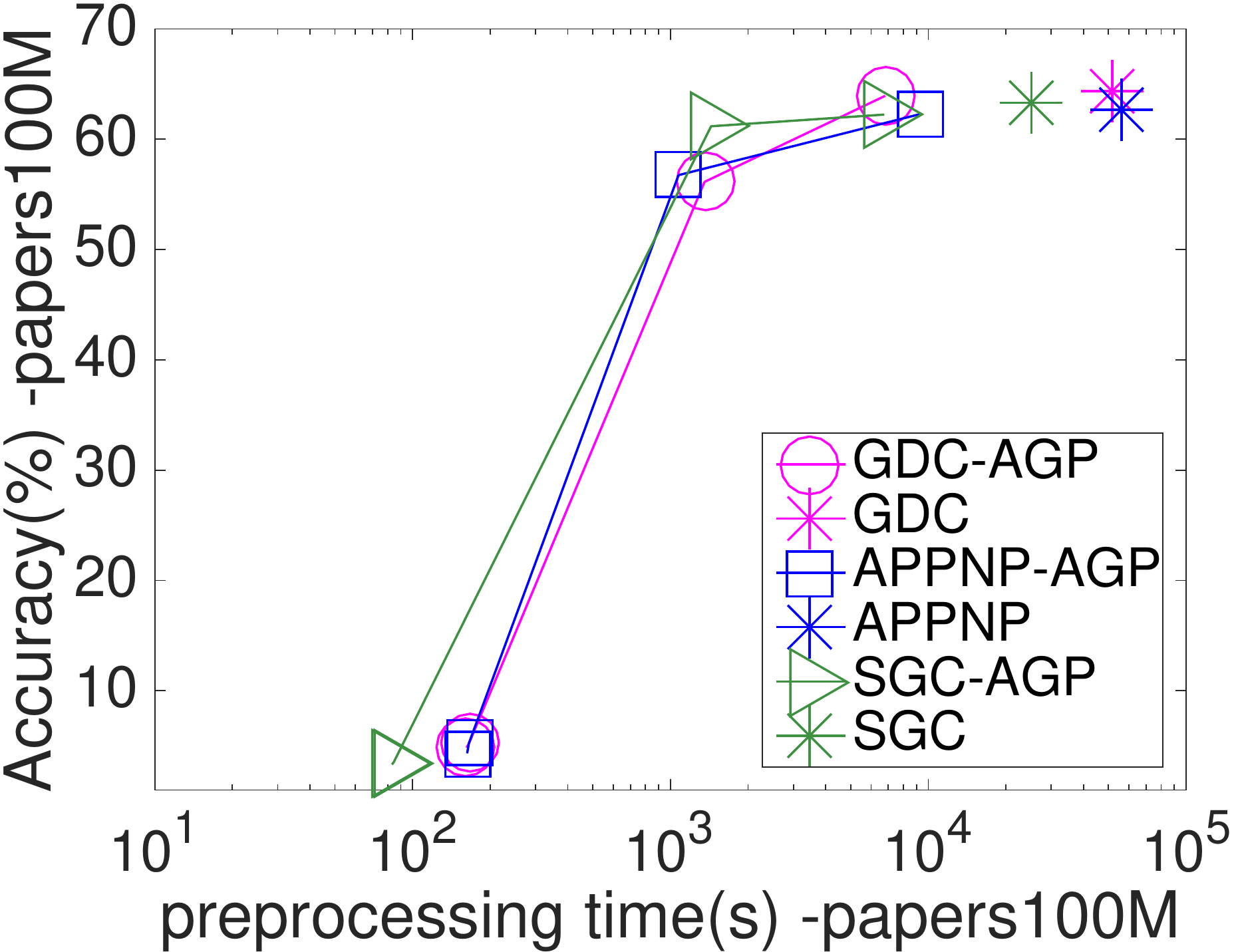} 
		\end{tabular}
		\vspace{-5mm}
		\caption{Tradeoffs between {\em Accuracy(\%)} and preprocessing time in node classification (Best viewed in color).}
		\label{fig:GNN-accuracy-query}
		\vspace{-2mm}
	\end{small}
\end{figure*}

\begin{figure}[t]
	\begin{small}
		\centering
		%\vspace{-2mm}
		%    \begin{footnotesize}
		\begin{tabular}{cccc}
			%\multicolumn{4}{c}{\hspace{-4mm} \includegraphics[height=5mm]{./Figs/legend_large.eps}} \vspace{-1mm} \\
			%\hspace{-3mm} \includegraphics[height=25mm]{./Figs/HKPR-conductance-query-DB.eps} &
		%	\hspace{-4mm}  \includegraphics[height=34mm]{./Figs/GNN-accuracy-memory-Amazon.eps} &
			\hspace{-4mm}  \includegraphics[height=32mm]{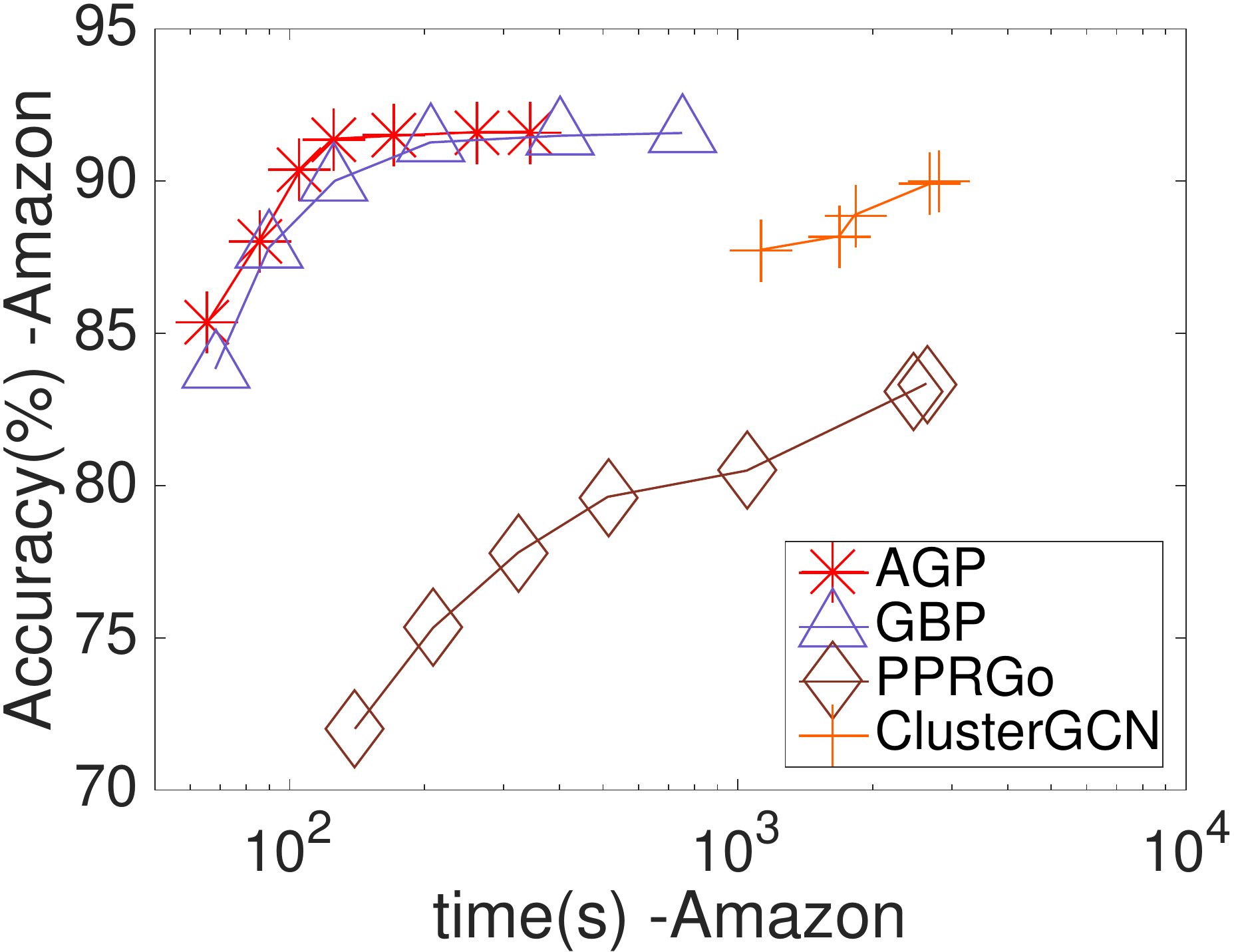} &
			\includegraphics[height=32mm]{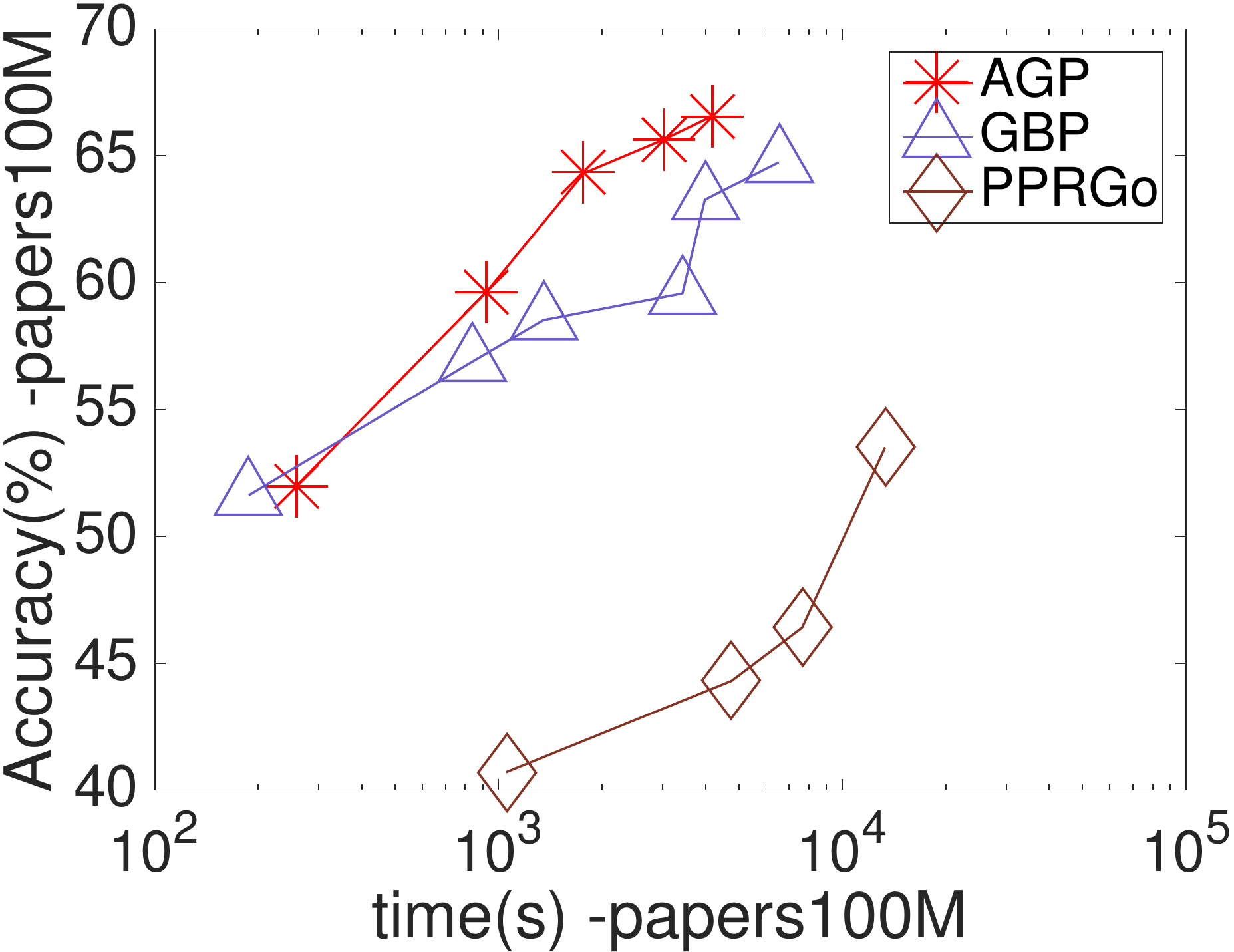}
		\end{tabular}
		\vspace{-5mm}
		\caption{Comparison with GBP, PPRGo and ClusterGCN.}
		\label{fig:GBP}
		\vspace{-4mm}
	\end{small}
\end{figure}

%\vspace{-2mm}
\subsection{Local clustering with HKPR}\label{subsec:clustering}
In this subsection, we conduct experiments to evaluate the performance of AGP in local clustering problem. We select HKPR among various node proximity measures as it achieves the state-of-the-art result for local clustering~\cite{chung2018computing,kloster2014heat,yang2019TEA}. 
%We will evaluate the trade-off curve between the query time and the approximation quality, as well as the trade-off curve between the query time and the clustering quality.

% first evaluate the query efficiency of heat kernel PageRank computed by AGP against other competitors. Then we compare the clustering quality of these methods using the derived HKPR results. Besides, we also exploit the influence of heat kernel parameter $t$ to the query effects of each algorithm. 

% choose one node proximity metric $\vec{\vec{\pi}}_s(t)$ first and approximate corresponding proximity results for each node $t \in V$ on the graph. 
% Next, sort the nodes on the graph in the descending order of $\frac{\vec{\pi}_s(t)}{d_t}$. 
% Finally, conduct the $\bf sweeping$ operation through the sorted nodes to find the subset with minimum conductance as the approximate clustering result. 
% In our experiments, we also obey this structure and apply heat kernel PageRank $\vec{\vec{\pi}}_s=\sum_{i=0}^\infty \frac{e^{-t}t^i}{i!}\cdot \left(\mathbf{A}\mathbf{D}^{-1}\right)^i\cdot \vec{e}_s$ as the node proximity following the state-of-the-art work~\cite{chung2018computing, kloster2014heat, yang2019TEA}. 

%normalize each proximity by
%Then they where heat kernel PageRank(HKPR) is the most commonly used proximity. 

%So in our experiments, we also compute the heat kernel PageRank first for every node on the graph with a given node as the seed. 

\header{\bf Datasets and Environment.} 
%Table~\ref{tbl:datasets} presents the detailed information of the datasets we used for local clustering, which can be obtained from~\cite{snapnets,LWA}. We use three undirected graphs: YouTube, Orkut, and Friendster in our experiments, as most of the local clustering methods can only support undirected graphs. We also use a large directed graph Twitter to demonstrate AGP's effectiveness on directed graphs. All local clustering experiments are conducted on a machine with an Intel(R) Xeon(R) Gold 6126@2.60GHz CPU and 500GB memory.
We use three undirected graphs: YouTube, Orkut, and Friendster in our experiments, as most of the local clustering methods can only support undirected graphs. We also use a large directed graph Twitter to demonstrate AGP's effectiveness on directed graphs. The four datasets can be obtained from~\cite{snapnets,LWA}. We summarize the detailed information of the four datasets in Table~\ref{tbl:datasets}.

%\header{\bf Methods and Parameters.} 
\header{\bf Methods.} 
We compare AGP with four local clustering methods: TEA~\cite{yang2019TEA} and its optimized version TEA+, ClusterHKPR~\cite{chung2018computing}, and the Monte-Carlo method (MC). We use the results derived by the Basic Propagation Algorithm~\ref{alg:AGP-deter} with $L = 50$ as the ground truths for evaluating the trade-off curves of the approximate algorithms. Detailed descriptions on parameter settings are deferred to the appendix due to the space limits.

\header{\bf Metrics.} 
%We use {\em MaxError} and {\em Precision@k} as our metrics to measure the approximation quality of each method. {\em MaxError} is defined as $MaxError =\max_{v \in V}\left|\frac{\vec{\pi}(v)}{d_v}-\frac{\vec{\epi}(v)}{d_v}\right|$, which measures the maximum error between the true normalized HKPR and the estimated value. Note that we use the basic propagation algorithm~\ref{alg:AGP-deter} with $L=50$ to obtain the ground truths of the normalized HKPR. 
We use {\em MaxError} as our metric to measure the approximation quality of each method. {\em MaxError} is defined as $MaxError =\max_{v \in V}\left|\frac{\vec{\pi}(v)}{d_v}-\frac{\vec{\epi}(v)}{d_v}\right|$, which measures the maximum error between the true normalized HKPR and the estimated value. We refer to $\frac{\vec{\pi}(v)}{d_{v}} $ as the {\em normalized HKPR} value from $s$ to $v$. %Note that we use the basic propagation algorithm~\ref{alg:AGP-deter} with $L=50$ to obtain the ground truths of the normalized HKPR.
On directed graph, $d_v$ is substituted by the out-degree $d_{out}(v)$. 
%We use {\em Precision@k} as our metric to measure the approximation quality of each method. Let $V_k$ denote the set of $k$ nodes with highest normalized HKPR values, and $\hat{V}_k$ denote the estimated top-$k$ node set returned by an approximate mthod. {\em Precision@k} is defined as the percentage of nodes in $\hat{V}_k$ that coincides with the actual top-$k$ results $V_k$. We use {\em Precision@k} to evaluate the accuracy of the relative node order of each method. We use the basic propagation algorithm~\ref{alg:AGP-deter} with $L=50$ to obtain the ground truths of the normalized HKPR. On directed graph, $d_t$ is substituted by the out-degree $d_{out}(v)$. 
 
%We also consider the quality of the cluster. More precisely, after deriving the (approximate )HKPR vector, we sort the nodes in descending order of the normalized HKPR values  and perform a sweeping operation to find the cluster with minimum conductance,  $\Phi(S)=\frac{|cut(S)|}{\min\{vol(S),2m-vol(S)\}}$ to measure the clustering quality, where $vol(S)=\sum_{v \in S}d(v)$, and $cut(S)=\{(u,v)\in E \mid u \in S, v \in V-S \}$. We evaluate the trade-offs between the minimum conductance found by each method and its query time. For each metric, we return the average of 50 randomly selected source nodes.  

%As mentioned in Section~\ref{subsec:cluserting-pre}, we perform local clustering methods with a sweeping algorithm. 
We also consider the quality of the cluster algorithms, which is measured by the {\em conductance}. Given a subset $S  \subseteq V$, the conductance is defined as $\Phi(S)\hspace{-1mm}=\hspace{-1mm}\frac{|cut(S)|}{\min\{vol(S),2m-vol(S)\}}$, where $vol(S)\hspace{-1mm}=\hspace{-1mm}\sum_{v \in S}d(v)$, and $cut(S)\hspace{-1mm}=\hspace{-1mm}\{(u,v)\hspace{-1mm}\in\hspace{-1mm} E \mid u \in S, v \in V-S \}$.
We perform a sweeping algorithm~\cite{Teng2004Nibble,FOCS06_FS,chung2007HKPR,chung2018computing,yang2019TEA} to find a subset $S$ with small conductance. More precisely, after deriving the (approximate) HKPR vector from a source node $s$, we sort the nodes $\{v_1, \ldots, v_n\}$ in descending order of the normalized HKPR values that $\frac{\vec{\pi}(v_1)}{d_{v_1}}\hspace{-1mm}\ge\hspace{-1mm} \frac{\vec{\pi}(v_2)}{d_{v_2}} \hspace{-1mm}\ge \hspace{-1mm}\ldots \hspace{-1mm}\ge \hspace{-1mm}\frac{\vec{\pi}(v_n)}{d_{v_n}}$. %we sort the nodes in descending order of the normalized HKPR values. 
Then, we sweep through $\{v_1, \ldots, v_n\}$ and find the node set with the minimum {\em conductance} among partial sets $S_i\hspace{-1mm} =\hspace{-1mm} \{v_1, \ldots,v_i\}, i=1,\ldots,n-1 $.

%More precisely, given a node $s$, we compute the HKPR vector $\vec{\pi} = \sum_{i=0}^\infty \frac{e^{-t}t^i}{i!}\cdot \left(\mathbf{A}\mathbf{D}^{-1}\right)^i\cdot \vec{e}_s$ of $s$, and sort the nodes $\{v_1, \ldots, v_n\}$ in descending order of $\frac{\vec{\pi}(v_1)}{d_{v_1}}\ge \frac{\vec{\pi}(v_2)}{d_{v_2}} \ge \ldots \ge \frac{\vec{\pi}(v_n)}{d_{v_n}}$. We refer to $\frac{\vec{\pi}(v)}{d_{v}} $ as the {\em normalized HKPR} value from $s$ to $v$. Then, we sweep through $\{v_1, \ldots, v_n\}$ and find the node set with the minimum conductance among partial sets $S_i = \{v_1, \ldots,v_i\}, i=1,\ldots,n-1 $. 
%After deriving the (approximate )HKPR vector, we perform a sweep algorithm to find the clusters with minimum conductance,  $\Phi(S)=\frac{|cut(S)|}{\min\{vol(S),2m-vol(S)\}}$ to measure the clustering quality, where $vol(S)=\sum_{v \in S}d(v)$, and $cut(S)=\{(u,v)\in E \mid u \in S, v \in V-S \}$. We evaluate the trade-offs between the minimum conductance found by each method and its query time. For each metric, we return the average of 50 randomly selected source nodes. 

\header{\bf Experimental Results.}
%Figure~\ref{fig:HKPR-precision-query} plots the trade-off curve between {\em Precision@50} and the query time for each method. We omit TEA and TEA+ on Twitter as they cannot handle directed graphs. We observe that AGP achieves the highest precision among the five approximate algorithms on all four datasets under the same query time. 
Figure~\ref{fig:HKPR-maxerror-query} plots the trade-off curve between the {\em MaxError} and query time. The time of reading graph is not counted in the query time. We observe that AGP  achieves the lowest curve among the five algorithms on all four datasets, which means AGP incurs the least error under the same query time. As a generalized algorithm for the graph propagation problem, these results suggest that AGP outperforms the state-of-the-art HKPR algorithms in terms of the approximation quality. 

% $1$ using the least time, which reflects the query efficiency of AGP. Besides, we notice that TEA+ and Monte-Carlo show a better performance than TEA and ClusterHKPR, which concurs with  analysis. In Figure~\ref{fig:HKPR-maxerror-query}, we plot the trade-offs between {\em MaxError} and query time. AGP is also the fastest method to reach the same additive error. Note that AGP can always achieve a $10x$ faster than TEA+ and $20\times-30\times$ faster than Monte-Carlo based methods. 
% Because the algorithm structure of TEA and TEA+ are only for undirected graphs. So we omit the lines of these two methods on the directed graph {\em Twitter}. On {\em Twitter}, AGP still has a better performance than the other Monte-Carlo based methods, which demonstrates the effectiveness of AGP on directed graphs. 

%To evaluate the quality of the clusters found by each method, Figure~\ref{fig:HKPR-conductance-query} shows the trade-off curve between conductance and the query time for each algorithm. We omit Twitter as the conductance metric which is only defined on undirected graphs.  We observe that AGP can achieve the lowest conductance-query time curve among the five approximate methods, which concurs that AGP provides estimators that are closest to the actual normalized HKPR. %ORIGIN!!
To evaluate the quality of the clusters found by each method, Figure~\ref{fig:HKPR-conductance-query} shows the trade-off curve between conductance and the query time on two large undirected graphs Orkut and Friendster. 
We observe that AGP can achieve the lowest conductance-query time curve among the five approximate methods on both of the two datasets, which concurs with the fact that AGP provides estimators that are closest to the actual normalized HKPR.

\subsection{Node classification with GNN}\label{subsec:GNN}
In this section, we evaluate AGP's ability to scale  existing Graph Neural Network models on large graphs.

\header{\bf Datasets.}
%We use four publicly available graph datasets with different size: a socal network Reddit~\cite{hamilton2017graphSAGE}, a customer interaction network Yelp~\cite{zeng2019graphsaint}, a co-purchasing network Amazon~\cite{chiang2019clusterGCN} and a large citation network Papers100M~\cite{hu2020ogb}. Table~\ref{tbl:datasets_gnn} summarizes the statistics of the datasets. Note that $d$ is the dimension of the node feature, and the label rate is the percentage of labeled nodes in the graph. Following~\cite{zeng2019graphsaint,zou2019layer}, we perform inductive node classification on Yelp, Amazon and Reddit, and semi-supervised transductive node classification on Papers100M. More specifically, for inductive node classification tasks, we train the model on a graph with labeled nodes and predict nodes' labels on a testing graph. For semi-supervised transductive node classification tasks, we train the model with a small subset of labeled nodes and predict other nodes' labels in the same graph. We follow the same training/validation/testing split as previous works in GNN~\cite{zeng2019graphsaint,hu2020ogb}. A detailed discussion on the setting of the experiments can be found in the appendix.
We use four publicly available graph datasets with different size: a socal network Reddit~\cite{hamilton2017graphSAGE}, a customer interaction network Yelp~\cite{zeng2019graphsaint}, a co-purchasing network Amazon~\cite{chiang2019clusterGCN} and a large citation network Papers100M~\cite{hu2020ogb}. Table~\ref{tbl:datasets_gnn} summarizes the statistics of the datasets, where $d$ is the dimension of the node feature, and the label rate is the percentage of labeled nodes in the graph. A detailed discussion on datasets is deferred to the appendix. 

% We first evaluate GBP's performance for transductive semi-supervised learning on the three popular citation networks (Cora, Citeseer, and Pubmed). Then we compare GBP with scalable GNN methods three medium to large graphs PPI, Yelp, Amazon in terms of inductive learning ability. Finally, we present the first empirical study of transductive semi-supervised on billion-scale network Friendster. 

\begin{table}[t]
	%\vspace{-2mm}
	\centering
	\tblcapup
	\caption{Datasets for node classification.}
	\vspace{-5mm}
	\tblcapdown
	\begin{small}
		\begin{tabular}{|l|r|r|r|r|r|} %p{1.3in}|}
			\hline
			{\bf Data Set} & {\bf $\boldsymbol{n}$}& {\bf $\boldsymbol{m}$} & {\bf $\boldsymbol{d}$}& {\bf Classes} 	& {\bf Label $\%$}  \\ \hline
    Reddit     & 232,965 &   114,615,892 &602 & 41  & 0.0035           \\
    Yelp        & 716,847     & 6,977,410   & 300 & 100 & 0.7500         \\
    Amazon        & 2,449,029   & 61,859,140  & 100   & 47 & 0.7000      \\
    Papers100M     & 111,059,956 & 1,615,685,872  & 128  &  172 & 0.0109           \\
			\hline
		\end{tabular}
	\end{small}
	\label{tbl:datasets_gnn}
	%\tbldown
	%\vspace{-2mm}
\end{table}

%\header{\bf GNN models and detailed setup.} 
\header{\bf GNN models.} 
%%We consider three proximity-based GNN models: APPNP~\cite{Klicpera2018APPNP}, SGC~\cite{wu2019SGC}, and GDC~\cite{klicpera2019GDC}. We augment the three models with the AGP algorithm~\ref{alg:AGP-RQ} to obtain three variants: APPNP-AGP, SGC-AGP and GDC-AGP. Take SGC-AGP as an example. Recall that SGC uses $\mathbf{Z}= \left(\mathbf{D}^{-\frac{1}{2}} \mathbf{A}\mathbf{D}^{-\frac{1}{2}} \right)^L \cdot \X$ to perform feature aggregation, where $\vec{X}$ is the $n\times d$ feature matrix. SGC-AGP treats each column of $\vec{X}$ as a graph signal $\vec{x}$ and perform randomized propagation algorithm (Algorithm~\ref{alg:AGP-RQ}) with predetermined error parameter $\delta$ to obtain the the final representation $\mathbf{Z}$. To achieve high parallelism, we perform propagation for multiple columns of $\mathbf{X}$ in parallel. 
We first consider three proximity-based GNN models: APPNP~\cite{Klicpera2018APPNP},SGC~\cite{wu2019SGC}, and GDC~\cite{klicpera2019GDC}. We augment the three models with the AGP Algorithm~\ref{alg:AGP-RQ} to obtain three variants: APPNP-AGP, SGC-AGP and GDC-AGP. Besides, we also compare AGP with three scalable methods: PPRGo~\cite{bojchevski2020scaling}, GBP~\cite{chen2020GBP}, and ClusterGCN~\cite{chiang2019clusterGCN}.

\header{\bf Experimental results.} 
For GDC, APPNP and SGC, we divide the computation time into two parts: the preprocessing time for computing $\mathbf{Z}$ and the training time for performing mini-batch. Accrording to~\cite{chen2020GBP}, the preprocessing time is the main bottleneck for achieving scalability. Hence, in Figure~\ref{fig:GNN-accuracy-query}, we show the trade-off between the preprocessing time and the classification accuracy for SGC, APPNP, GDC and the corresponding AGP models. For each dataset, the three snowflakes represent the exact methods SGC, APPNP, and GDC, which can be distinguished by colors.
We observe that compared to the exact models, the approximate models generally achieve a $10\times$ speedup in preprocessing time without sacrificing the classification accuracy. For example, on the billion-edge graph Papers100M,  SGC-AGP achieves an accuracy of $62\%$ in less than $2,000$ seconds, while the exact model SGC needs over $20,000$ seconds to finish. %We also observe that APPNP-AGP achieves a higher trade-off curve than PPRGo does, which means AGP is more efficient than PPRGo in terms of scaling PPR-based GNN on large graphs. %To eliminate the effect of parallelism, we also present each method's clock time in the appendix of the technical report~\cite{TechnicalReport} . The results concur with what we found in Figure~\ref{fig:GNN-accuracy-query}. 

Furthermore, we compare AGP against three recent works, PPRGo, GBP, ClusterGCN, to demonstrate the efficiency of AGP. Because ClusterGCN cannot be decomposed into propagation phase and traing phase, in Figure~\ref{fig:GBP}, we draw the trade-off plot between the computation time (i.e., preprocessing time plus training time) and the classsification accuracy on Amazon and Papers100M. %, the largest publicly available graphs for inductive and transductive node classification tasks, respectively.
%We draw the trade-off plots between the total computation time (i.e., preprocessing time plus training time) and classification accuracy for each model.
%Note that on Papers100M, we omit ClusterGCN and the Monte-Carlo phase in GBP because they all suffer from the out-of-memory problem. 
Note that on Papers100M, we omit ClusterGCN because of the out-of-memory problem. We tune $a, b, w_i$ in AGP for scalability. The detailed hyper-parameters of each method are summarized in the appendix. We observe AGP outperforms PPRGo and ClusterGCN on both Amazon and Papers100M in terms of accuracy and running time. In particular, given the same running time, AGP achieves a higher accuracy than GBP does on Papers100M. We attribute this quality to the randomization introduced in AGP.

%Figure~\ref{fig:GNN-accuracy-memory} shows the memory overhead of each approximate method. Recall that the AGP algorithm~\ref{alg:AGP-RQ} only maintains two $n$ dimension vectors: the residue vector $\vec{r}$ and the reserve vector $\vec{q}$. Consequently, AGP only takes a fixed memory, which can be ignored compared to the graph's size and the feature matrix.  Such property is ideal for scaling GNN models on massive graphs. On the other hand, PPRGo requires a large memory size to store the intermediate local push results from each training node.

 % Besides, it set the maximum length of walk as $K=t\cdot \frac{\log{1/\delta}}{\log{\log{1/\delta}}}$. When the poisson random number $k$ is large than $K$, it truncate the walk at its $K_{th}$ step. For ClusterHKPR, 
% To obtain a tradeoff curve between the approximation quality and query time, we vary $\delta$ from $0.5$ to $0.0001$ in our experiments. 
%and generate $\frac{16\log{n}}{\delta^2 }$ walks from the seed, instead of $$. 
%evaluates the performance of AGP against state-of-the-art methods.

%%% Local Variables:
%%% mode: latex
%%% TeX-master: "paper"
%%% End:

%\vspace{-2mm}
\section{Conclusion} \label{sec:conclusion}
\vspace{-1mm}
In this paper, we propose the concept of approximate graph propagation, which unifies various proximity measures, including transition probabilities, PageRank and Personalized PageRank, heat kernel PageRank, and Katz. We present a randomized graph propagation algorithm that achieves almost optimal computation time with a theoretical error guarantee. We conduct an extensive experimental study to demonstrate the effectiveness of AGP on real-world graphs. We show that AGP outperforms the state-of-the-art algorithms in the specific application of local clustering and node classification with GNNs. For future work, it is interesting to see if the AGP framework can inspire new proximity measures for graph learning and mining tasks.

%\vspace{-1mm}

%%% Local Variables:
%%% mode: latex
%%% TeX-master: "paper"
%%% End:

%\vspace{-1mm}
\section{ACKNOWLEDGEMENTS}
\vspace{-1mm}
Zhewei Wei was supported by National Natural Science Foundation of China (NSFC) No. 61972401 and No. 61932001, by the Fundamental Research Funds for the Central Universities and the Research Funds of Renmin University of China under Grant 18XNLG21, and by Alibaba Group through Alibaba Innovative Research Program. 
Hanzhi Wang was supported by the Outstanding Innovative Talents Cultivation Funded Programs 2020 of Renmin Univertity of China.
Sibo Wang was supported by Hong Kong RGC ECS No. 24203419, RGC CRF No. C4158-20G, and NSFC No. U1936205. 
Ye Yuan was supported by NSFC No. 61932004 and No. 61622202, and by FRFCU No. N181605012. 
Xiaoyong Du was supported by NSFC No. 62072458. 
Ji-Rong Wen was supported by NSFC  No. 61832017, and by Beijing Outstanding Young Scientist Program NO. BJJWZYJH012019100020098. 
This work was supported by Public Computing Cloud, Renmin University of China, and by China Unicom Innovation Ecological Cooperation Plan. 

%This research was supported in part by National Natural Science Foundation of China (No. 61832017, No. 61972401, No. 61932001, No. U1711261, No. 61932004 and No. 61622202), by FRFCU No. N181605012, by Beijing Outstanding Young Scientist Program NO. BJJ WZYJH012019100020098, and by the Fundamental Research Funds for the Central Universities and the Research Funds of Renmin University of China under Grant 18XNLG21.

%%% Local Variables:
%%% mode: latex
%%% TeX-master: "paper"
%%% End:


\balance
\bibliographystyle{plain}
\bibliography{paper}

\begin{thebibliography}{10}

\bibitem{TechnicalReport}
\url{https://arxiv.org/pdf/2106.03058.pdf}.

\bibitem{snapnets}
\url{http://snap.stanford.edu/data}.

\bibitem{LWA}
\url{http://law.di.unimi.it/datasets.php}.

\bibitem{andersen2008robust}
Reid Andersen, Christian Borgs, Jennifer Chayes, John Hopcroft, Kamal Jain,
  Vahab Mirrokni, and Shanghua Teng.
\newblock Robust pagerank and locally computable spam detection features.
\newblock In {\em Proceedings of the 4th international workshop on Adversarial
  information retrieval on the web}, pages 69--76, 2008.

\bibitem{FOCS06_FS}
Reid Andersen, Fan R.~K. Chung, and Kevin~J. Lang.
\newblock Local graph partitioning using pagerank vectors.
\newblock In {\em FOCS}, pages 475--486, 2006.

\bibitem{backstrom2011supervised}
Lars Backstrom and Jure Leskovec.
\newblock Supervised random walks: predicting and recommending links in social
  networks.
\newblock In {\em Proceedings of the fourth ACM international conference on Web
  search and data mining}, pages 635--644, 2011.

\bibitem{bojchevski2020scaling}
Aleksandar Bojchevski, Johannes Klicpera, Bryan Perozzi, Amol Kapoor, Martin
  Blais, Benedek R{\'o}zemberczki, Michal Lukasik, and Stephan G{\"u}nnemann.
\newblock Scaling graph neural networks with approximate pagerank.
\newblock In {\em KDD}, pages 2464--2473, 2020.

\bibitem{bressan2018sublinear}
Marco Bressan, Enoch Peserico, and Luca Pretto.
\newblock Sublinear algorithms for local graph centrality estimation.
\newblock In {\em FOCS}, pages 709--718, 2018.

\bibitem{bringmann2012efficient}
Karl Bringmann and Konstantinos Panagiotou.
\newblock Efficient sampling methods for discrete distributions.
\newblock In {\em ICALP}, pages 133--144, 2012.

\bibitem{cam1935stirling}
LL~Cam et~al.
\newblock The central limit theorem around.
\newblock {\em Statistical Science}, pages 78--91, 1935.

\bibitem{chen2020GBP}
Ming Chen, Zhewei Wei, Bolin Ding, Yaliang Li, Ye~Yuan, Xiaoyong Du, and
  Ji-Rong Wen.
\newblock Scalable graph neural networks via bidirectional propagation.
\newblock In {\em NeurIPS}, 2020.

\bibitem{chiang2019clusterGCN}
Wei-Lin Chiang, Xuanqing Liu, Si~Si, Yang Li, Samy Bengio, and Cho-Jui Hsieh.
\newblock Cluster-gcn: An efficient algorithm for training deep and large graph
  convolutional networks.
\newblock In {\em KDD}, pages 257--266, 2019.

\bibitem{chung2007HKPR}
Fan Chung.
\newblock The heat kernel as the pagerank of a graph.
\newblock {\em PNAS}, 104(50):19735--19740, 2007.

\bibitem{chung2018computing}
Fan Chung and Olivia Simpson.
\newblock Computing heat kernel pagerank and a local clustering algorithm.
\newblock {\em European Journal of Combinatorics}, 68:96--119, 2018.

\bibitem{coskun2016efficient}
Mustafa Coskun, Ananth Grama, and Mehmet Koyuturk.
\newblock Efficient processing of network proximity queries via chebyshev
  acceleration.
\newblock In {\em KDD}, pages 1515--1524, 2016.

\bibitem{Fogaras2005MC}
D{\'{a}}niel Fogaras, Bal{\'{a}}zs R{\'{a}}cz, K{\'{a}}roly Csalog{\'{a}}ny,
  and Tam{\'{a}}s Sarl{\'{o}}s.
\newblock Towards scaling fully personalized pagerank: Algorithms, lower
  bounds, and experiments.
\newblock {\em Internet Mathematics}, 2(3):333--358, 2005.

\bibitem{foster2001faster}
Kurt~C Foster, Stephen~Q Muth, John~J Potterat, and Richard~B Rothenberg.
\newblock A faster katz status score algorithm.
\newblock {\em Computational \& Mathematical Organization Theory},
  7(4):275--285, 2001.

\bibitem{gupta2013wtf}
Pankaj Gupta, Ashish Goel, Jimmy Lin, Aneesh Sharma, Dong Wang, and Reza Zadeh.
\newblock Wtf: The who to follow service at twitter.
\newblock In {\em {WWW}}, pages 505--514, 2013.

\bibitem{hamilton2017graphSAGE}
William~L. Hamilton, Rex Ying, and Jure Leskovec.
\newblock Inductive representation learning on large graphs.
\newblock In {\em {NeurIPS}}, 2017.

\bibitem{He2016ResNet}
Kaiming He, Xiangyu Zhang, Shaoqing Ren, and Jian Sun.
\newblock Deep residual learning for image recognition.
\newblock In {\em CVPR}, pages 770--778, 2016.

\bibitem{hu2020ogb}
Weihua Hu, Matthias Fey, Marinka Zitnik, Yuxiao Dong, Hongyu Ren, Bowen Liu,
  Michele Catasta, and Jure Leskovec.
\newblock Open graph benchmark: Datasets for machine learning on graphs.
\newblock In {\em NeurIPS}, 2020.

\bibitem{jeh2003scaling}
Glen Jeh and Jennifer Widom.
\newblock Scaling personalized web search.
\newblock In {\em Proceedings of the 12th international conference on World
  Wide Web}, pages 271--279, 2003.

\bibitem{jung2017bepi}
Jinhong Jung, Namyong Park, Sael Lee, and U~Kang.
\newblock Bepi: Fast and memory-efficient method for billion-scale random walk
  with restart.
\newblock In {\em SIGMOD}, pages 789--804, 2017.

\bibitem{katz1953Katz}
Leo Katz.
\newblock A new status index derived from sociometric analysis.
\newblock {\em Psychometrika}, 18(1):39--43, 1953.

\bibitem{kipf2016GCN}
Thomas~N Kipf and Max Welling.
\newblock Semi-supervised classification with graph convolutional networks.
\newblock In {\em ICLR}, 2017.

\bibitem{Klicpera2018APPNP}
Johannes Klicpera, Aleksandar Bojchevski, and Stephan G{\"u}nnemann.
\newblock Predict then propagate: Graph neural networks meet personalized
  pagerank.
\newblock In {\em ICLR}, 2019.

\bibitem{klicpera2019GDC}
Johannes Klicpera, Stefan Weißenberger, and Stephan Günnemann.
\newblock Diffusion improves graph learning.
\newblock In {\em NeurIPS}, pages 13354--13366, 2019.

\bibitem{kloster2014heat}
Kyle Kloster and David~F Gleich.
\newblock Heat kernel based community detection.
\newblock In {\em KDD}, pages 1386--1395, 2014.

\bibitem{Liben2003link}
David Liben{-}Nowell and Jon~M. Kleinberg.
\newblock The link prediction problem for social networks.
\newblock In {\em {CIKM}}, pages 556--559, 2003.

\bibitem{lin2020index}
Dandan Lin, Raymond Chi-Wing Wong, Min Xie, and Victor~Junqiu Wei.
\newblock Index-free approach with theoretical guarantee for efficient random
  walk with restart query.
\newblock In {\em ICDE}, pages 913--924. IEEE, 2020.

\bibitem{lofgren2013personalized}
Peter Lofgren and Ashish Goel.
\newblock Personalized pagerank to a target node.
\newblock {\em arXiv preprint arXiv:1304.4658}, 2013.

\bibitem{ou2016asymmetric}
Mingdong Ou, Peng Cui, Jian Pei, Ziwei Zhang, and Wenwu Zhu.
\newblock Asymmetric transitivity preserving graph embedding.
\newblock In {\em KDD}, pages 1105--1114, 2016.

\bibitem{page1999pagerank}
Lawrence Page, Sergey Brin, Rajeev Motwani, and Terry Winograd.
\newblock The pagerank citation ranking: bringing order to the web.
\newblock 1999.

\bibitem{Shin2015BEAR}
Kijung Shin, Jinhong Jung, Lee Sael, and U.~Kang.
\newblock {BEAR:} block elimination approach for random walk with restart on
  large graphs.
\newblock In {\em {SIGMOD}}, pages 1571--1585, 2015.

\bibitem{Teng2004Nibble}
Daniel~A Spielman and Shang-Hua Teng.
\newblock Nearly-linear time algorithms for graph partitioning, graph
  sparsification, and solving linear systems.
\newblock In {\em STOC}, pages 81--90, 2004.

\bibitem{wang2020RBS}
Hanzhi Wang, Zhewei Wei, Junhao Gan, Sibo Wang, and Zengfeng Huang.
\newblock Personalized pagerank to a target node, revisited.
\newblock In {\em Proceedings of the 26th ACM SIGKDD International Conference
  on Knowledge Discovery \& Data Mining}, pages 657--667, 2020.

\bibitem{Wang2016HubPPR}
Sibo Wang, Youze Tang, Xiaokui Xiao, Yin Yang, and Zengxiang Li.
\newblock Hubppr: Effective indexing for approximate personalized pagerank.
\newblock {\em {PVLDB}}, 10(3):205--216, 2016.

\bibitem{Wang2017FORA}
Sibo Wang, Renchi Yang, Xiaokui Xiao, Zhewei Wei, and Yin Yang.
\newblock {FORA:} simple and effective approximate single-source personalized
  pagerank.
\newblock In {\em KDD}, pages 505--514, 2017.

\bibitem{wei2018topppr}
Zhewei Wei, Xiaodong He, Xiaokui Xiao, Sibo Wang, Shuo Shang, and Ji-Rong Wen.
\newblock Topppr: top-k personalized pagerank queries with precision guarantees
  on large graphs.
\newblock In {\em SIGMOD}, pages 441--456. ACM, 2018.

\bibitem{wu2019SGC}
Felix Wu, Amauri Souza, Tianyi Zhang, Christopher Fifty, Tao Yu, and Kilian
  Weinberger.
\newblock Simplifying graph convolutional networks.
\newblock In {\em ICML}, pages 6861--6871. PMLR, 2019.

\bibitem{yang2019TEA}
Renchi Yang, Xiaokui Xiao, Zhewei Wei, Sourav~S Bhowmick, Jun Zhao, and
  Rong-Hua Li.
\newblock Efficient estimation of heat kernel pagerank for local clustering.
\newblock In {\em SIGMOD}, pages 1339--1356, 2019.

\bibitem{zeng2019graphsaint}
Hanqing Zeng, Hongkuan Zhou, Ajitesh Srivastava, Rajgopal Kannan, and Viktor
  Prasanna.
\newblock Graphsaint: Graph sampling based inductive learning method.
\newblock In {\em ICLR}, 2020.

\bibitem{zou2019layer}
Difan Zou, Ziniu Hu, Yewen Wang, Song Jiang, Yizhou Sun, and Quanquan Gu.
\newblock Layer-dependent importance sampling for training deep and large graph
  convolutional networks.
\newblock In {\em NeurIPS}, pages 11249--11259, 2019.

\end{thebibliography}


\appendix
%\clearpage
\begin{table}[t]
\begin{minipage}[t]{1\columnwidth}
%\begin{table}[t]
	%\vspace{-1mm}
	%\centering
	%\lefting
	\tblcapup
	\caption{Hyper-parameters of AGP.}
	\vspace{-5mm}
	\tblcapdown
	\begin{small}
	\begin{threeparttable}
		\begin{tabular}{|@{\hspace{+1mm}}c@{\hspace{1mm}}|@{\hspace{+1mm}}c@{\hspace{+1mm}}|@{\hspace{+1mm}}c@{\hspace{+1mm}}|@{\hspace{+0.5mm}}c@{\hspace{+0.5mm}}|@{\hspace{+0.5mm}}c@{\hspace{+0.5mm}}|@{\hspace{+3.5mm}}c@{\hspace{+3.5mm}}|@{\hspace{+1.1mm}}c@{\hspace{+1.1mm}}|@{\hspace{+1mm}}c@{\hspace{+1mm}}|}
			\hline
			\multirow{2}{*}{{\bf Dataset}} & {\bf Learning}& \multirow{2}{*}{{\bf Dropout}} & {\bf Hidden} &{\bf Batch} & \multirow{2}{*}{\bf $\boldsymbol{t}$}&\multirow{2}{*}{\bf $\boldsymbol{\alpha}$}&\multirow{2}{*}{\bf $\boldsymbol{L}$} \\ 
			~&{\bf rate}&&{\bf dimension} & {\bf size} &&& \\ \hline
Yelp & 0.01 & 0.1 & 2048 & $3\cdot 10^4$&  4 & 0.9& 10 \\ 
Amazon  & 0.01 & 0.1 & 1024 &$10^5$& 4 & 0.2& 10\\ 
Reddit  & 0.0001 & 0.3 &  2048  &$10^4$& 3 & 0.1& 10 \\ 
Papers100M & 0.0001 & 0.3 & 2048 &$10^4$& 4 & 0.2& 10 \\ \hline
	\end{tabular}
	\end{threeparttable}
	\end{small}
	\label{tbl:parameters}
	%\tbldown
	\vspace{+2mm}
%\end{table}
\end{minipage}

\begin{minipage}[t]{1\columnwidth}
%\begin{table}[t]
	%\vspace{+1mm}
	\centering
	%\lefting
	\tblcapup
	\caption{Hyper-parameters of GBP.}
	\vspace{-5mm}
	\tblcapdown
	\begin{small}
		\begin{tabular}{|@{\hspace{+1mm}}c@{\hspace{+1mm}}|@{\hspace{+1mm}}c@{\hspace{+1mm}}|@{\hspace{+1mm}}c@{\hspace{+1mm}}|@{\hspace{+0.5mm}}c@{\hspace{+0.5mm}}|@{\hspace{+0.5mm}}c@{\hspace{+0.5mm}}|@{\hspace{+1mm}}c@{\hspace{+1mm}}|@{\hspace{+1mm}}c@{\hspace{+1mm}}|@{\hspace{+0.5mm}}c@{\hspace{+0.5mm}}|}\hline
    \multirow{2}{*}{\bf Dataset} & {\bf Learning} & \multirow{2}{*}{\bf Dropout} & {\bf Hidden} &{\bf Batch} & \multirow{2}{*}{\bf $r_{max}$} & \multirow{2}{*}{\bf $r$} & \multirow{2}{*}{\bf $\alpha$} \\ 
    &{\bf rate}&& {\bf dimension} & {\bf size} & & & \\ \hline
    Amazon & 0.01 & 0.1 &1024& $10^5$ & $10^{-7}$ & 0.2 & 0.2  \\ 
    Papers100M & 0.0001 & 0.3 &2048& $10^4$ & $10^{-8}$ & 0.5 & 0.2  \\ \hline
	\end{tabular}
	\end{small}
	\label{tbl:para-GBP}
	%\tbldown
	\vspace{+2mm}
%\end{table}
\end{minipage}

\begin{minipage}[t]{1\columnwidth}
%\begin{table}[t]
	%\vspace{+1mm}
	%\centering
	%\lefting
	\tblcapup
	\caption{Hyper-parameters of PPRGo.}
	\vspace{-5mm}
	\tblcapdown
	\begin{small}
		\begin{tabular}{|@{\hspace{+1mm}}c@{\hspace{+1mm}}|@{\hspace{+1mm}}c@{\hspace{+1mm}}|@{\hspace{+1mm}}c@{\hspace{+1mm}}|@{\hspace{+0.5mm}}c@{\hspace{+0.5mm}}|@{\hspace{+0.5mm}}c@{\hspace{+0.5mm}}|@{\hspace{+0.3mm}}c@{\hspace{+0.3mm}}|@{\hspace{+1mm}}c@{\hspace{+1mm}}|@{\hspace{+1.3mm}}c@{\hspace{+1.3mm}}|}\hline
    \multirow{2}{*}{\bf Dataset} & {\bf Learning} & \multirow{2}{*}{\bf Dropout} & {\bf Hidden} &{\bf Batch} & \multirow{2}{*}{\bf $r_{max}$} & \multirow{2}{*}{\bf $k$} & \multirow{2}{*}{\bf $L$} \\ 
    &{\bf rate}&& {\bf dimension} &{\bf size} & & & \\ \hline
    Amazon & 0.01 & 0.1 & 64 &$10^5$ &$5\cdot 10^{-5}$ & 64 & 8   \\ 
    Papers100M & 0.01 & 0.1 & 64 &$10^4$ &$10^{-4}$ & 32 & 8  \\ \hline
	\end{tabular}
	\end{small}
	\label{tbl:para-PPRGo}
	%\tbldown
	\vspace{+2mm}
%\end{table}
\end{minipage}

\begin{minipage}[t]{1\columnwidth}
%\begin{table}[t]
	%\vspace{+1mm}
	%\centering
	%\lefting
	\tblcapup
	\caption{Hyper-parameters of ClusterGCN.}
	\vspace{-5mm}
	\tblcapdown
	\begin{small}
		\begin{tabular}{|@{\hspace{+3.2mm}}c@{\hspace{+3.2mm}}|@{\hspace{+1mm}}c@{\hspace{+1mm}}|@{\hspace{+1mm}}c@{\hspace{+1mm}}|@{\hspace{+0.5mm}}c@{\hspace{+0.5mm}}|@{\hspace{+1mm}}c@{\hspace{+1mm}}|@{\hspace{+2.5mm}}c@{\hspace{+2.5mm}}|}\hline
    \multirow{2}{*}{\bf Dataset} & {\bf Learning} & \multirow{2}{*}{\bf Dropout} & {\bf Hidden} & \multirow{2}{*}{\bf layer} & \multirow{2}{*}{\bf partitions} \\ 
    &{\bf rate}&& {\bf dimension} & & \\ \hline
    Amazon & 0.01 & 0.2 & 400 & 4 & 15000  \\ \hline
	\end{tabular}
	\end{small}
	\label{tbl:para-ClusterGCN}
    \tbldown
	\vspace{+2mm}
%\end{table} 
\end{minipage}

\begin{minipage}[t]{1\columnwidth}
%\begin{table}[h]
	%\vspace{+1mm}
	%\centering
	%\lefting
	\tblcapup
	\caption{URLs of baseline codes.}
	\vspace{-5mm}
	\tblcapdown
	\begin{small}
		%\begin{tabular}{|c|c|c|} %p{1.3in}|}
			
		\begin{tabular}{|c|c|}\hline
	%{\bf Methods} & {\bf URL} & \hspace{-2mm}{\bf Commit}\hspace{-5mm}\\ \hline
    %GDC & https://github.com/klicperajo/gdc & \hspace{-2mm}14333fd\hspace{-2mm}\\
    %APPNP & https://github.com/rusty1s/pytorch$\_$geometric & \hspace{-2mm}f560655\hspace{-2mm}\\
    %SGC & https://github.com/Tiiiger/SGC &\hspace{-2mm}795ec93 \hspace{-2mm} \\ 
    %PPRGo & https://github.com/TUM-DAML/pprgo$\_$pytorch &\hspace{-2mm} d9f991e\hspace{-2mm}\\ 
    %GBP & https://github.com/chennnM/GBP &\hspace{-2mm}f811fc2 \\
    %\hspace{-2mm}ClusterGCN\hspace{-2mm} &\hspace{-2mm}https://github.com/benedekrozemberczki/ClusterGCN \hspace{-2mm} &\hspace{-2mm}a6b40cc\hspace{-2mm} \\ \hline
    {\bf Methods} & {\bf URL} \\ \hline
    GDC & https://github.com/klicperajo/gdc \\
    APPNP & https://github.com/rusty1s/pytorch$\_$geometric \\
    SGC & https://github.com/Tiiiger/SGC \\ 
    PPRGo & https://github.com/TUM-DAML/pprgo$\_$pytorch \\ 
    GBP & https://github.com/chennnM/GBP \\ 
    ClusterGCN &https://github.com/benedekrozemberczki/ClusterGCN \\ \hline
	\end{tabular}
	\end{small}
	\label{tbl:url}
	%\tbldown
	\vspace{+2mm}
%\end{table}
\end{minipage}
\end{table}

\section{Experimental Details}\label{sec:appendix}

%\begin{comment}
\begin{figure*}[t]
\begin{minipage}[t]{1\textwidth}
	\begin{small}
		\centering
		%\vspace{-5mm}
		%    \begin{footnotesize}
		\begin{tabular}{cccc}
			%\multicolumn{4}{c}{\hspace{-4mm} \includegraphics[height=5mm]{./Figs/legend_large.eps}} \vspace{-1mm} \\
			\hspace{-4mm} \includegraphics[height=32.7mm]{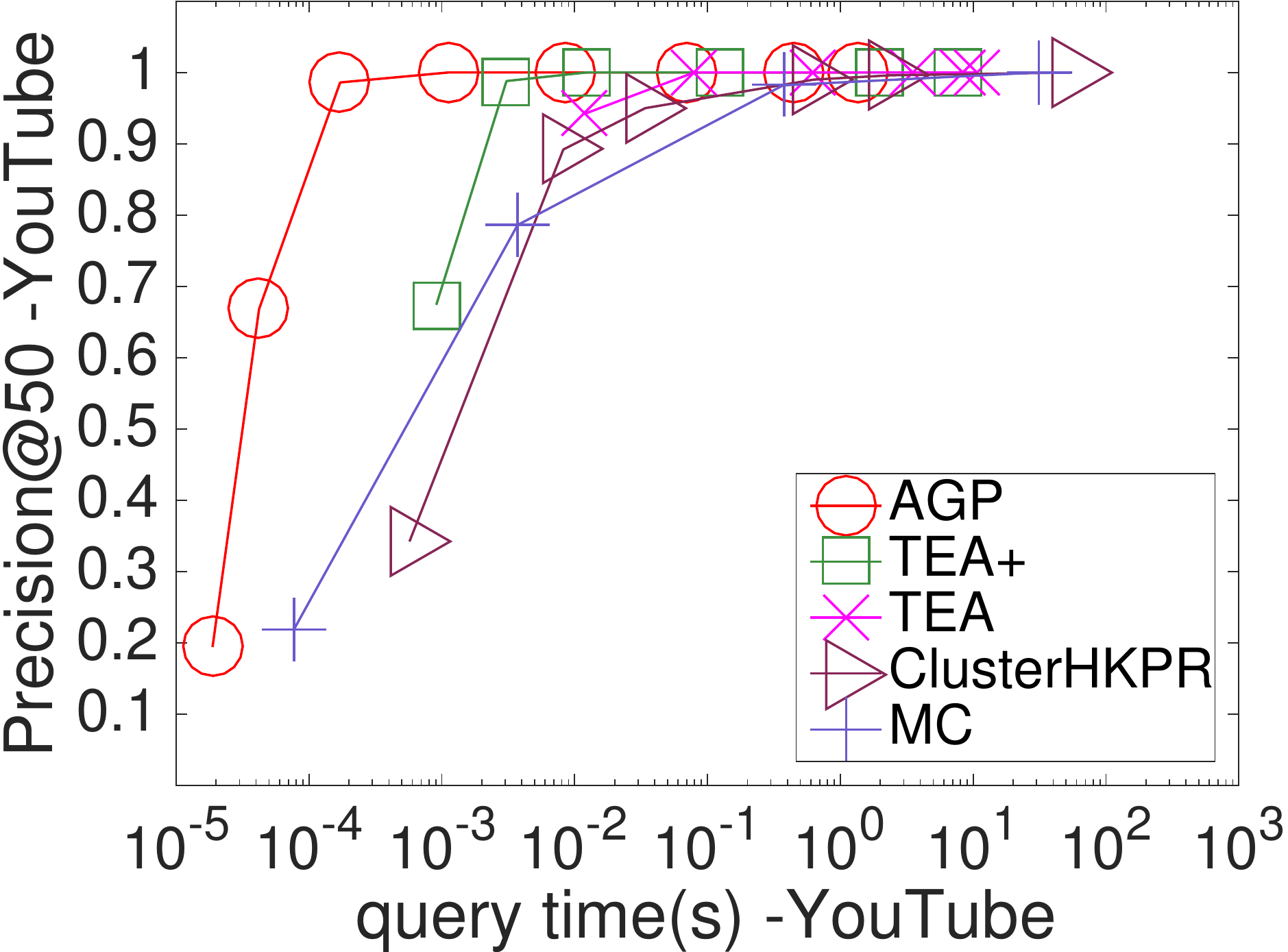} &
			%\hspace{-3mm} \includegraphics[height=25mm]{./Figs/HKPR-precision-query-DB.eps} &
			\hspace{-4mm} \includegraphics[height=32.7mm]{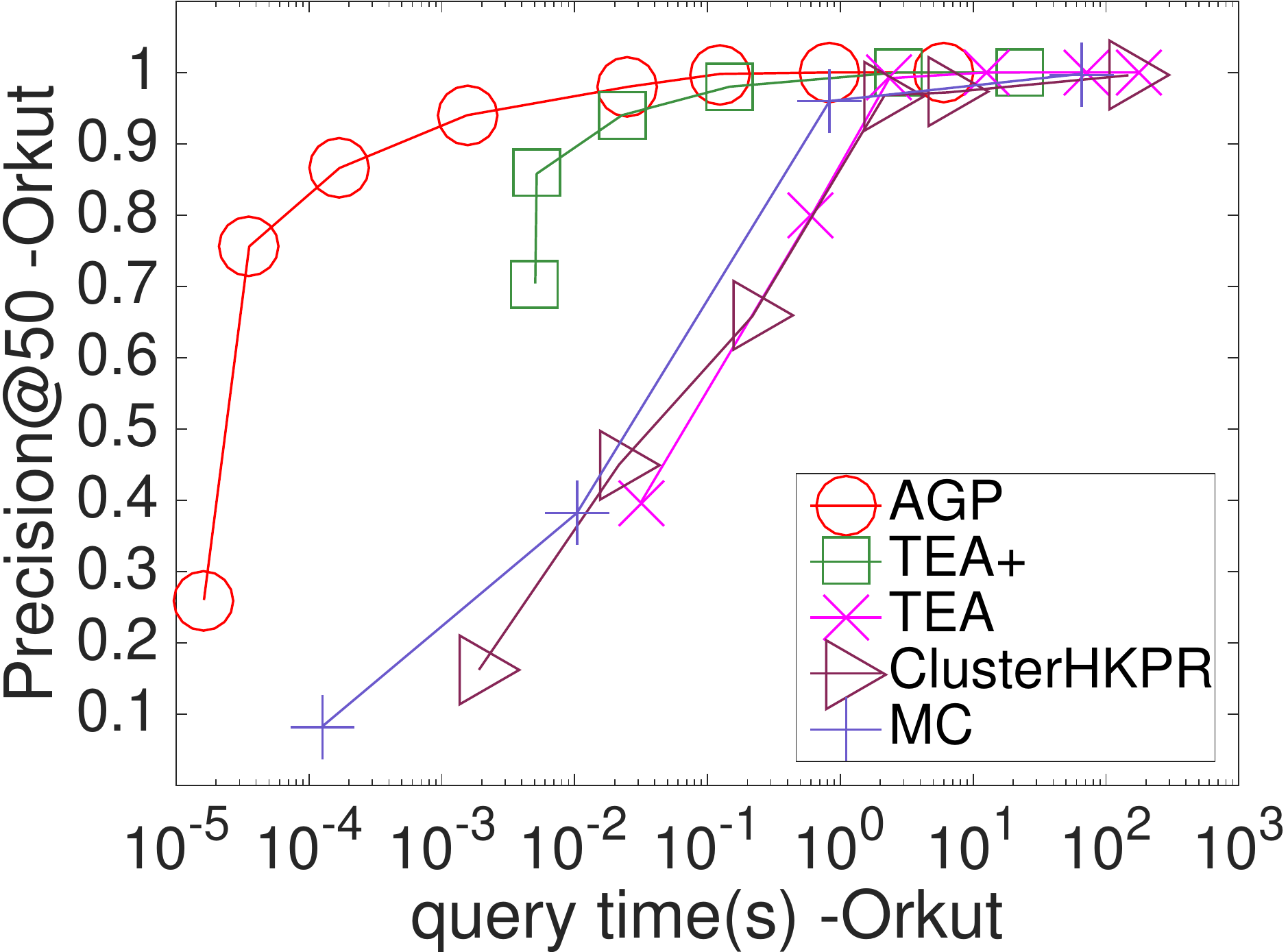} &
			\hspace{-4mm} \includegraphics[height=32.7mm]{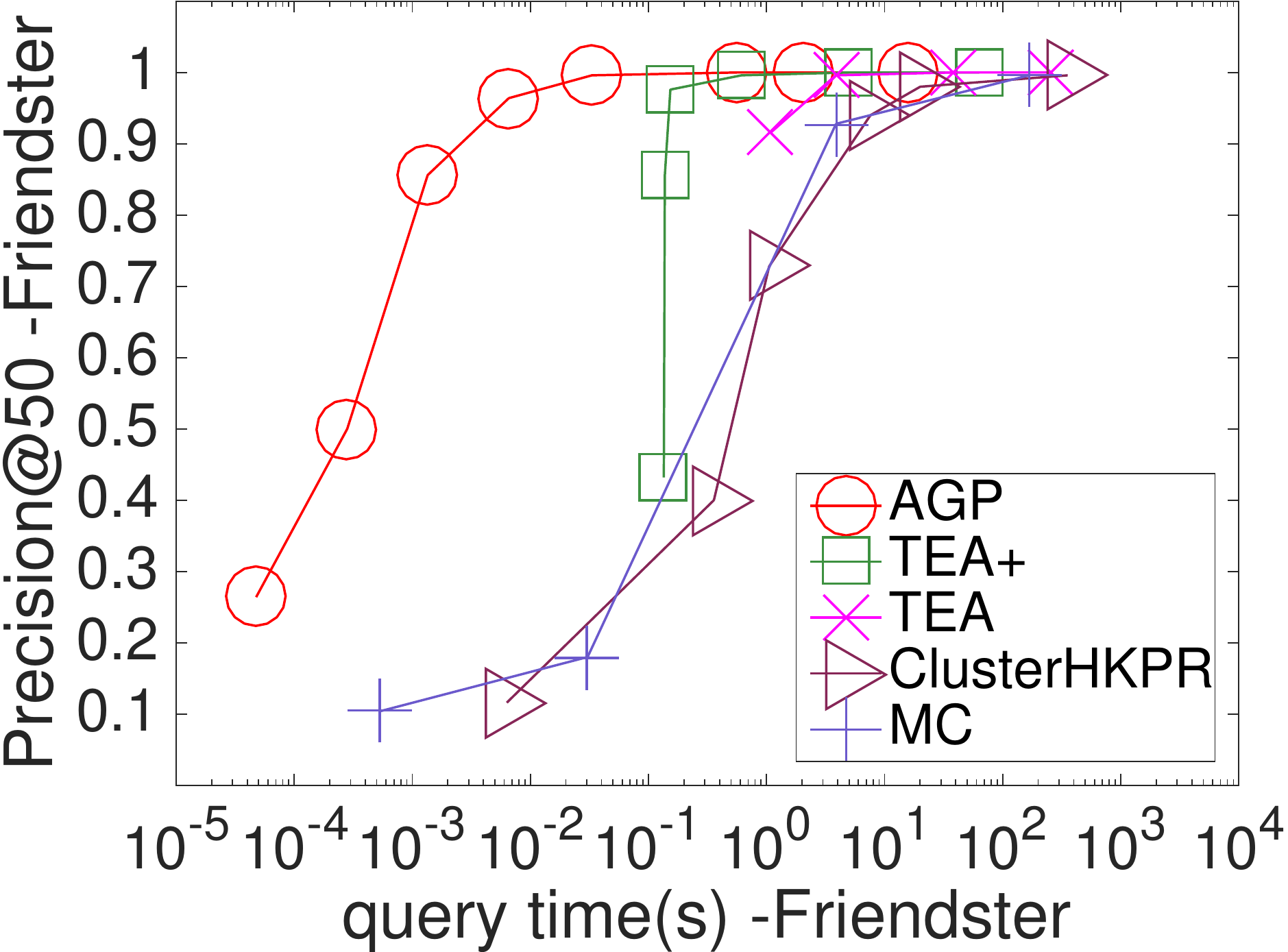} &
			\hspace{-4mm} \includegraphics[height=32.7mm]{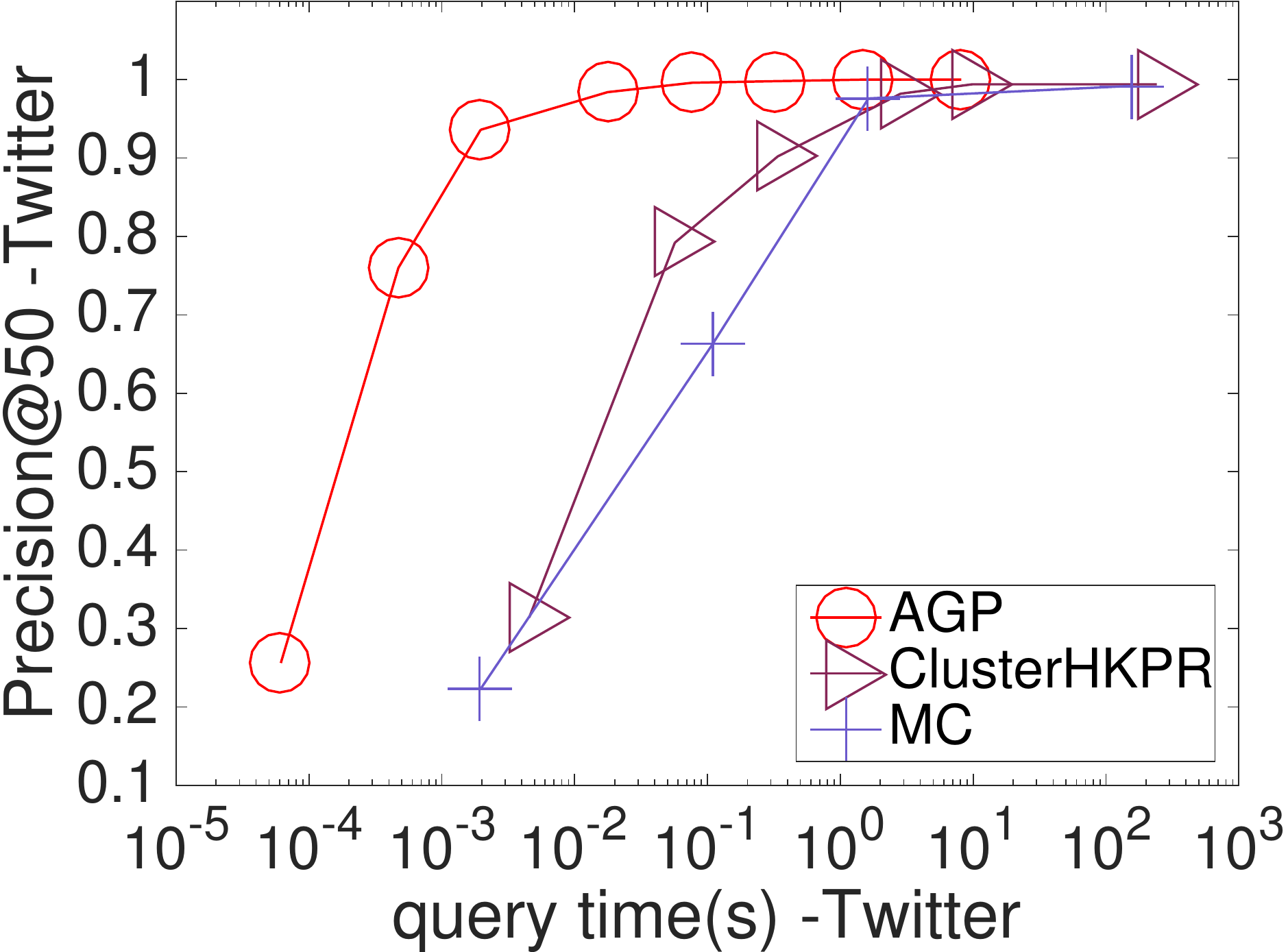} 
		\end{tabular}
		\vspace{-5mm}
		\caption{Tradeoffs between {\em Normalized Precision@50} and query time in local clustering.}
		\label{fig:HKPR-precision-query}
		\vspace{+1mm}
	\end{small}
%\end{figure*}
\end{minipage}

%\begin{figure*}[t]
\begin{minipage}[t]{1\textwidth}
	\begin{small}
		\centering
		%\vspace{-1mm}
		%    \begin{footnotesize}
		\begin{tabular}{cccc}
			%\multicolumn{4}{c}{\hspace{-4mm} \includegraphics[height=5mm]{./Figs/legend_large.eps}} \vspace{-1mm} \\
			\hspace{-4mm} \includegraphics[height=34mm]{./Figs/HKPR-conductance-query-OL-eps-converted-to.pdf} &
			\hspace{-4mm} \includegraphics[height=34mm]{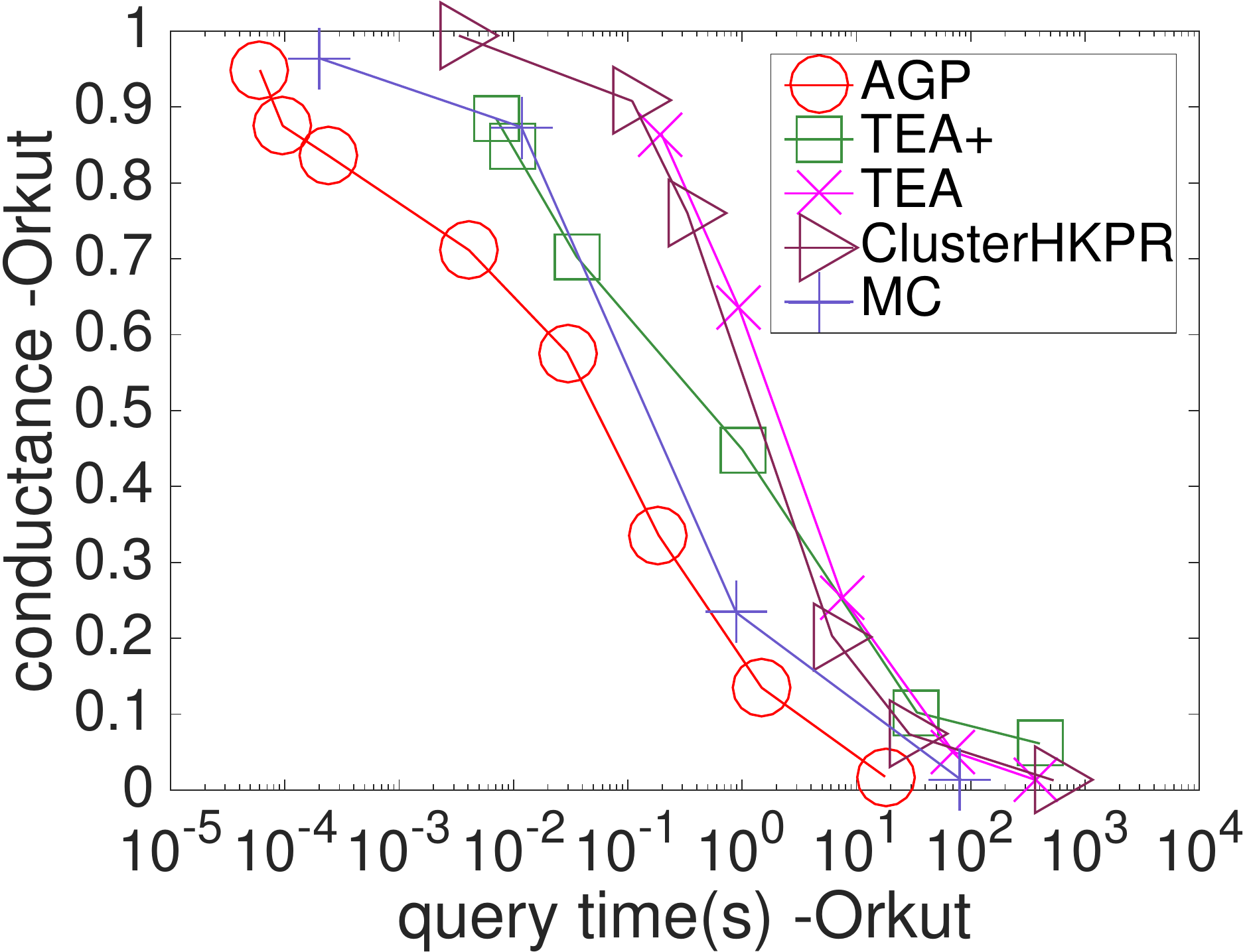} &
			\hspace{-4mm} \includegraphics[height=34mm]{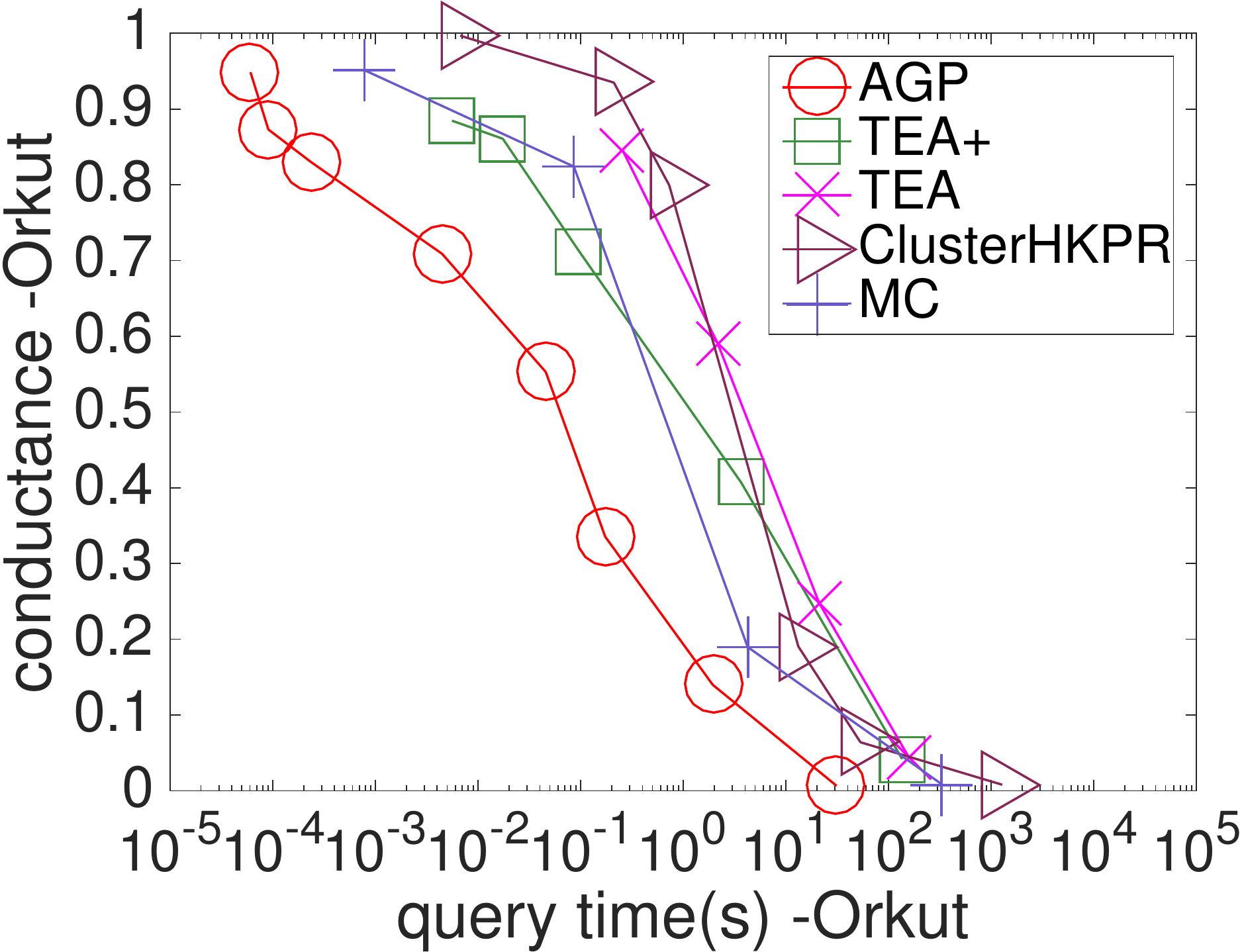} &
			\hspace{-4mm} \includegraphics[height=34mm]{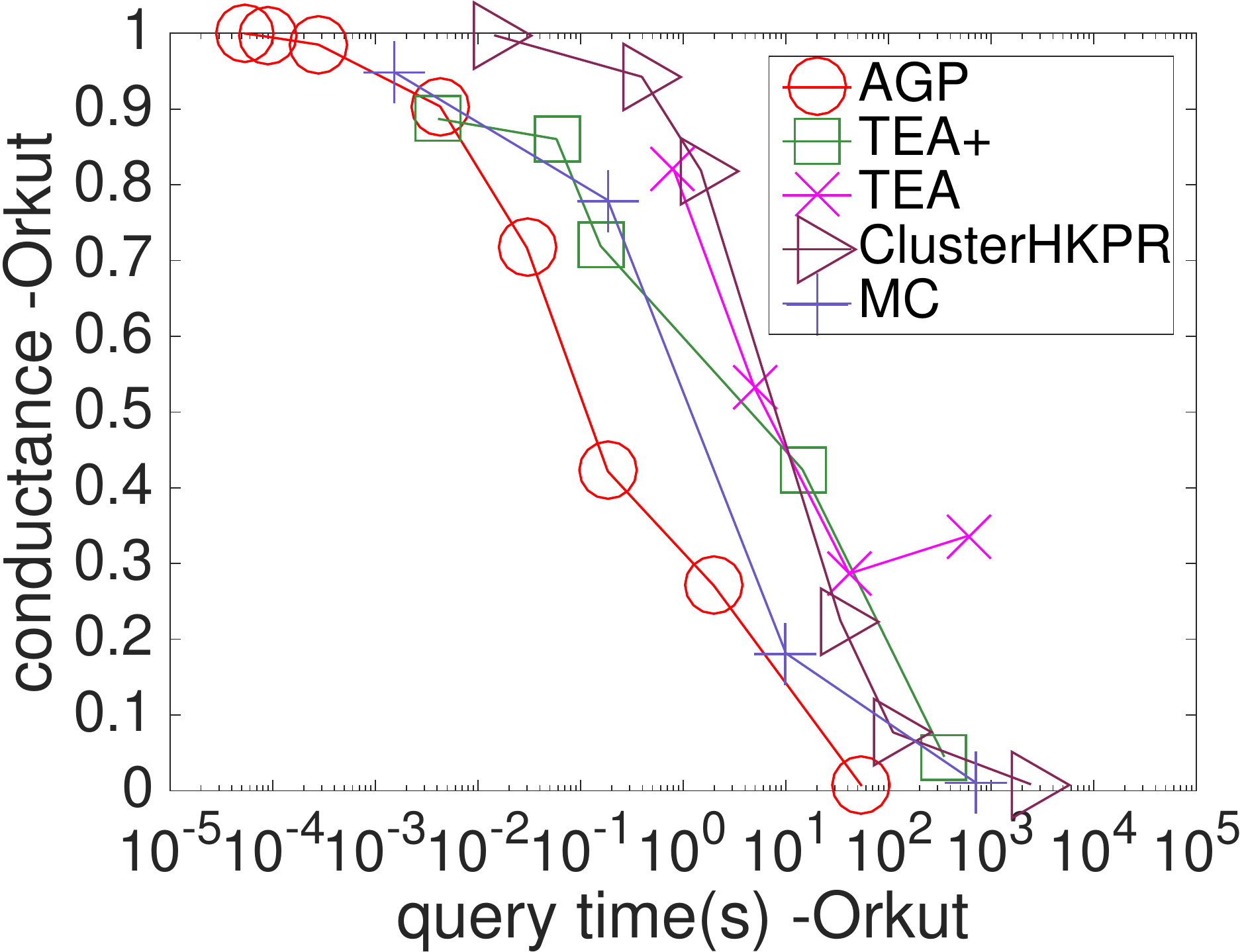} \\
			(a) $t=5$ & (b) $t=10$ & (c) $t=20$ & (d) $t=40$ 
		\end{tabular}
		\vspace{-5mm}
		\caption{Effect of heat constant t for {\em conductance} on {\em Orkut}.}
		\label{fig:conductance-query-OL}
	%\vspace{+1mm}
	\end{small}
%\end{figure*}
\end{minipage}
\end{figure*}

%\begin{comment}
%
%\end{comment}

\subsection{Local clustering with HKPR}

\header{\bf Methods and Parameters. }
For AGP, we set $a=0,b=1, w_i=\frac{e^{-t}t^i}{i!}$, and $\vec{x}=\vec{e}_s$ in Equation~\eqref{eqn:pi_gen} to simulate the HKPR equation $\vec{\pi}=\sum_{i=0}^\infty \frac{e^{-t}t^i}{i!}\cdot \left(\mathbf{A}\mathbf{D}^{-1}\right)^i\cdot \vec{e}_s$. We employ the randomized propagation algorithm~\ref{alg:AGP-RQ} with level number $L= O\left(\log {1/\delta}\right)$ and error parameter $\varepsilon = \frac{2\delta}{L(L+1)}$, where $\delta$ is the relative error threshold in Definition~\ref{def:pro-relative}. We use AGP to denote this method. We vary $\delta$ from $0.1$ to $10^{-8}$ with 0.1 decay step to obtain a trade-off curve between the approximation quality and the query time. 
%We also include the basic propagation algorithm~\ref{alg:AGP-deter} with $L = 50$ (denoted as Basic), which serves as the algorithm for computing ground and as a baseline for evaluating the trade-off curves of the approximate algorithms.
We use the results derived by the Basic Propagation Algorithm~\ref{alg:AGP-deter} with $L = 50$ as the ground truths for evaluating the trade-off curves of the approximate algorithms.

We compare AGP with four local clustering methods: TEA~\cite{yang2019TEA} and its optimized version TEA+, ClusterHKPR~\cite{chung2018computing}, and the Monte-Carlo method (MC). 
TEA~\cite{yang2019TEA}, as the state-of-the-art clustering method, combines a deterministic push process with the Monte-Carlo random walk. Given a graph $G=(V, E)$, and a seed node $s$, TEA conducts a local search algorithm to explore the graph around $s$ deterministically, and then generates random walks from nodes with residues exceeding a threshold parameter $r_{max}$.  One can manipulate $r_{max}$ to balance the two processes. It is shown in~\cite{yang2019TEA} that TEA can achieve $O\left(\frac{t\cdot \log{n}}{\delta}\right)$ time complexity, where $t$ is the constant heat kernel parameter.
ClusterHKPR~\cite{chung2018computing} is a Monte-Carlo based method that simulates adequate random walks from the given seed node and uses the percentage of random walks terminating at node $v$ as the estimation of $\vec{\pi}(v)$. 
%In the $k_{th}$ step, the walk stops at the node it currently walks at with the probability , or move to a randomly selected neighbor with the probability $\frac{t}{k}$. 
The length of walks $k$ follows the Poisson distribution $\frac{e^{-t}t^k}{k!}$. The number of random walks need to achieve a relative error of $\delta$ in Definition~\ref{def:pro-relative} is $O\left(\frac{t\cdot \log{n}}{\delta^3}\right)$. 
MC~\cite{yang2019TEA} is an optimized version of random walk process that sets identical length for each walk as $L=t\cdot \frac{\log{1/\delta}}{\log{\log{1/\delta}}}$. If a random walk visit node $v$ at the $k$-th step, we add $\frac{e^{-t}t^k}{n_r \cdot k!}$ to the propagation results $\vec{\epi}(v)$, where $n_r$ denotes the total number of random walks. The number of random walks to achieve a relative error of $\delta$ is also $O\left(\frac{t\cdot \log{n}}{\delta^3}\right)$. Similar to AGP, for each method, we vary $\delta$ from $0.1$ to $10^{-8}$ with 0.1 decay step to obtain a trade-off curve between the approximation quality and the query time. 
Unless specified otherwise, we set the heat kernel parameter $t$ as 5, following~\cite{kloster2014heat, yang2019TEA}. All local clustering experiments are conducted on a machine with an Intel(R) Xeon(R) Gold 6126@2.60GHz CPU and 500GB memory. 

\begin{figure*}[t]
\begin{minipage}[t]{1\textwidth}
	\begin{small}
		\centering
		\vspace{-4mm}
		%    \begin{footnotesize}
		\begin{tabular}{cccc}
			%\multicolumn{4}{c}{\hspace{-4mm} \includegraphics[height=5mm]{./Figs/legend_large.eps}} \vspace{-1mm} \\
			\hspace{-4.4mm} \includegraphics[height=34.5mm]{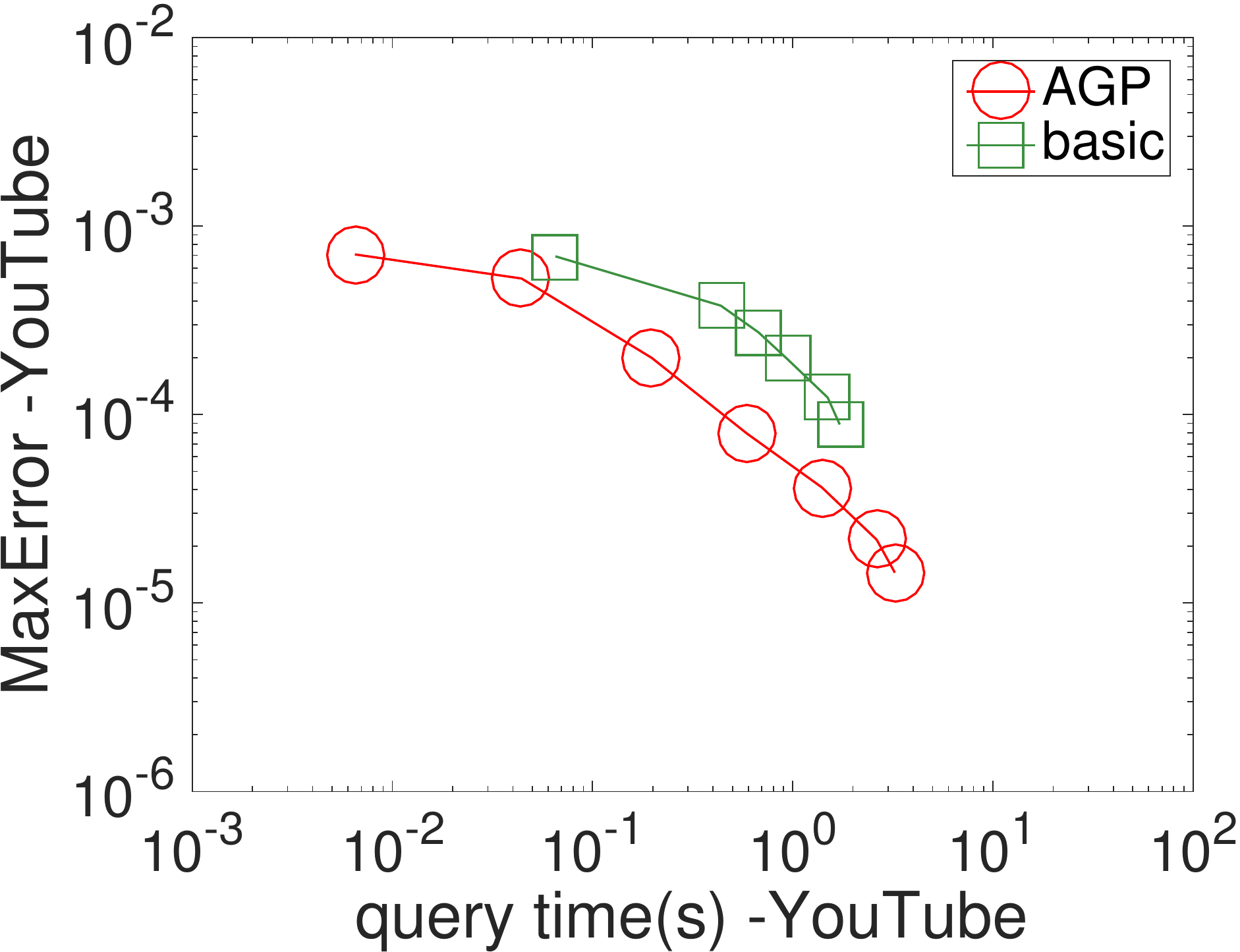} &
			\hspace{-4.4mm} \includegraphics[height=34.5mm]{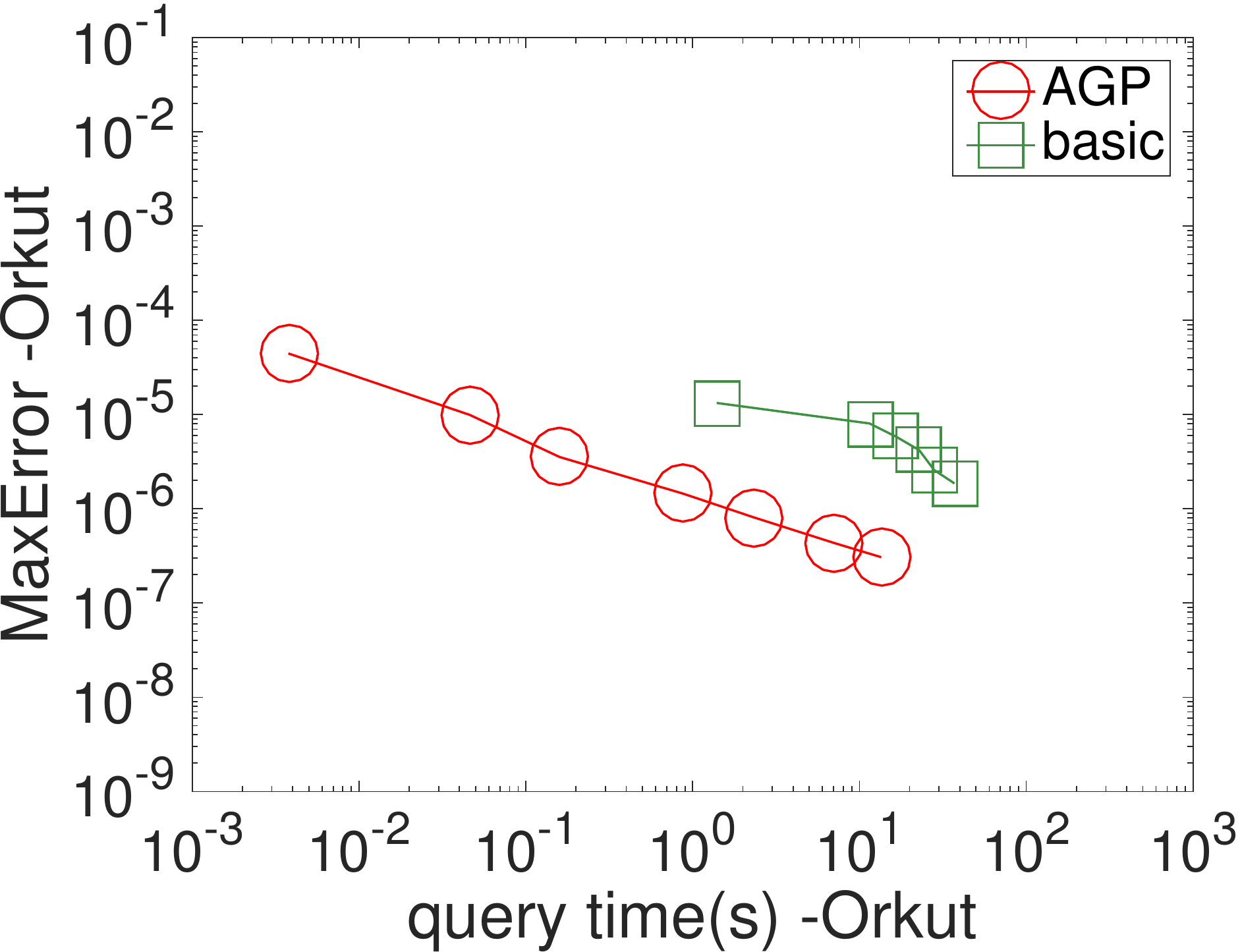} &
			\hspace{-4.4mm} \includegraphics[height=34.5mm]{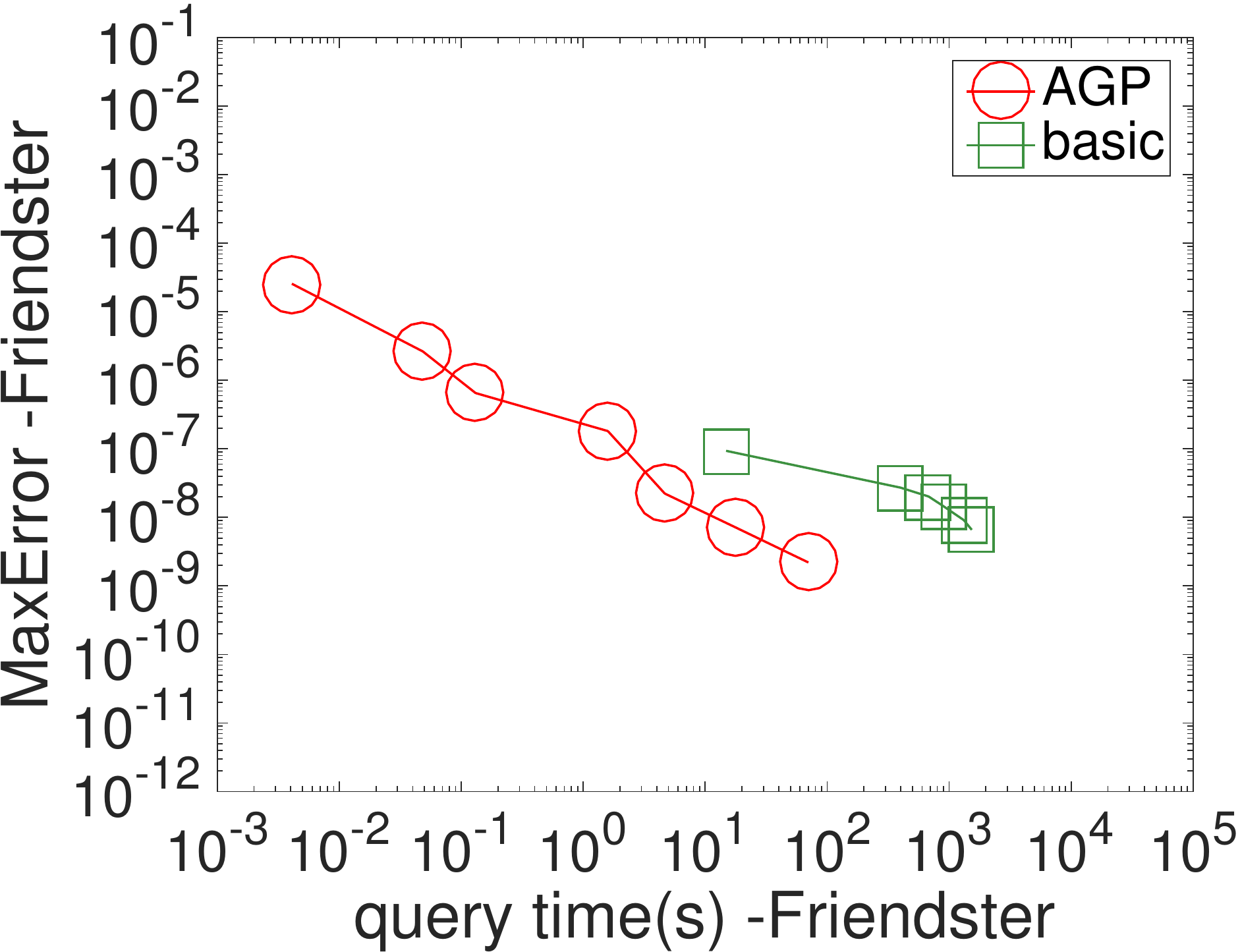} &
			\hspace{-4.4mm} \includegraphics[height=34.5mm]{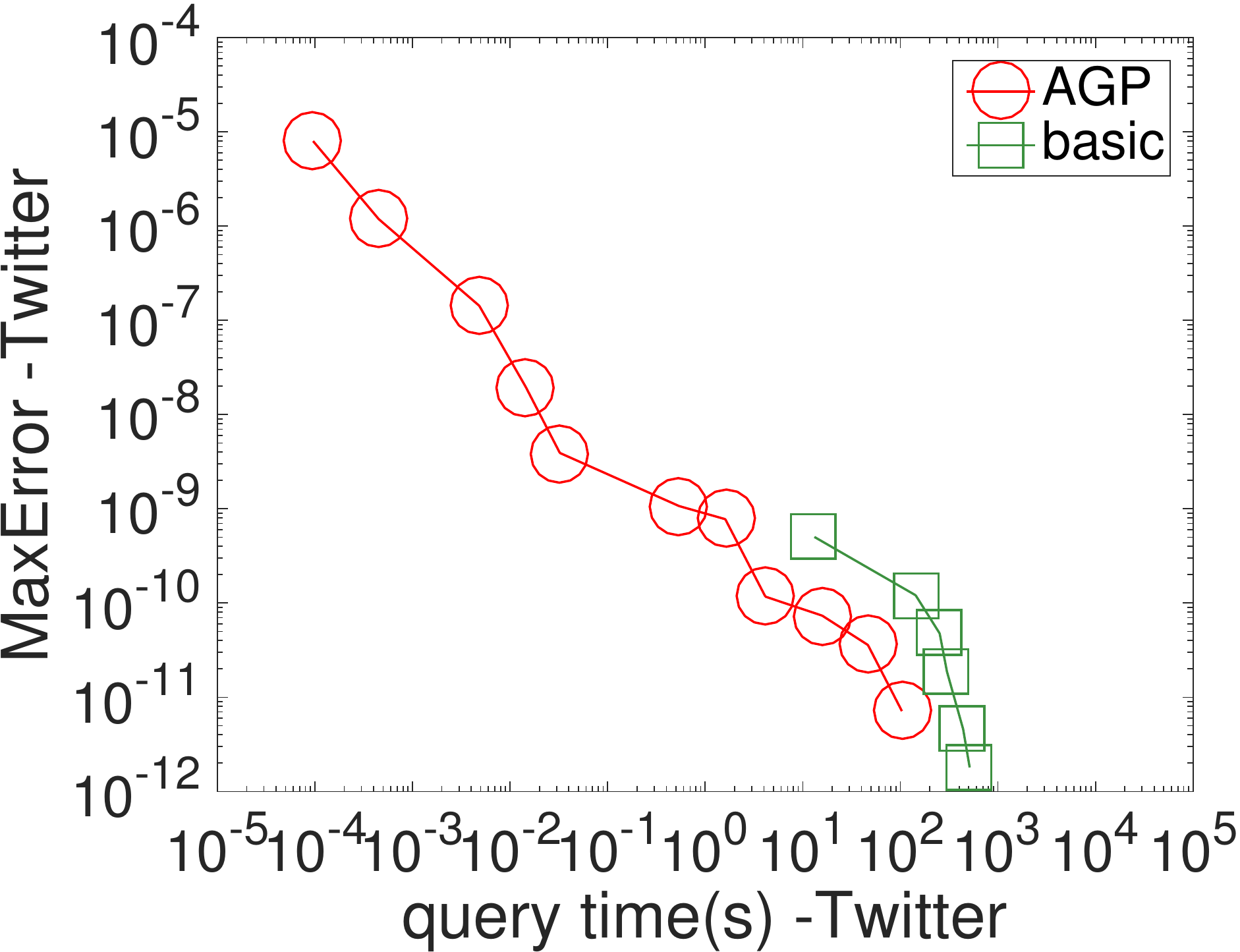} 
		\end{tabular}
		\vspace{-5mm}
		\caption{Tradeoffs between {\em MaxError} and query time of Katz.}
		\label{fig:Katz-MaxError-query}
		\vspace{+1mm}
	\end{small}
%\end{figure*}
\end{minipage}

%\begin{figure*}[t]
\begin{minipage}[t]{1\textwidth}
	\begin{small}
		\centering
		%\vspace{-2mm}
		%    \begin{footnotesize}
		\begin{tabular}{cccc}
			%\multicolumn{4}{c}{\hspace{-4mm} \includegraphics[height=5mm]{./Figs/legend_large.eps}} \vspace{-1mm} \\
			\hspace{-4mm} \includegraphics[height=33.8mm]{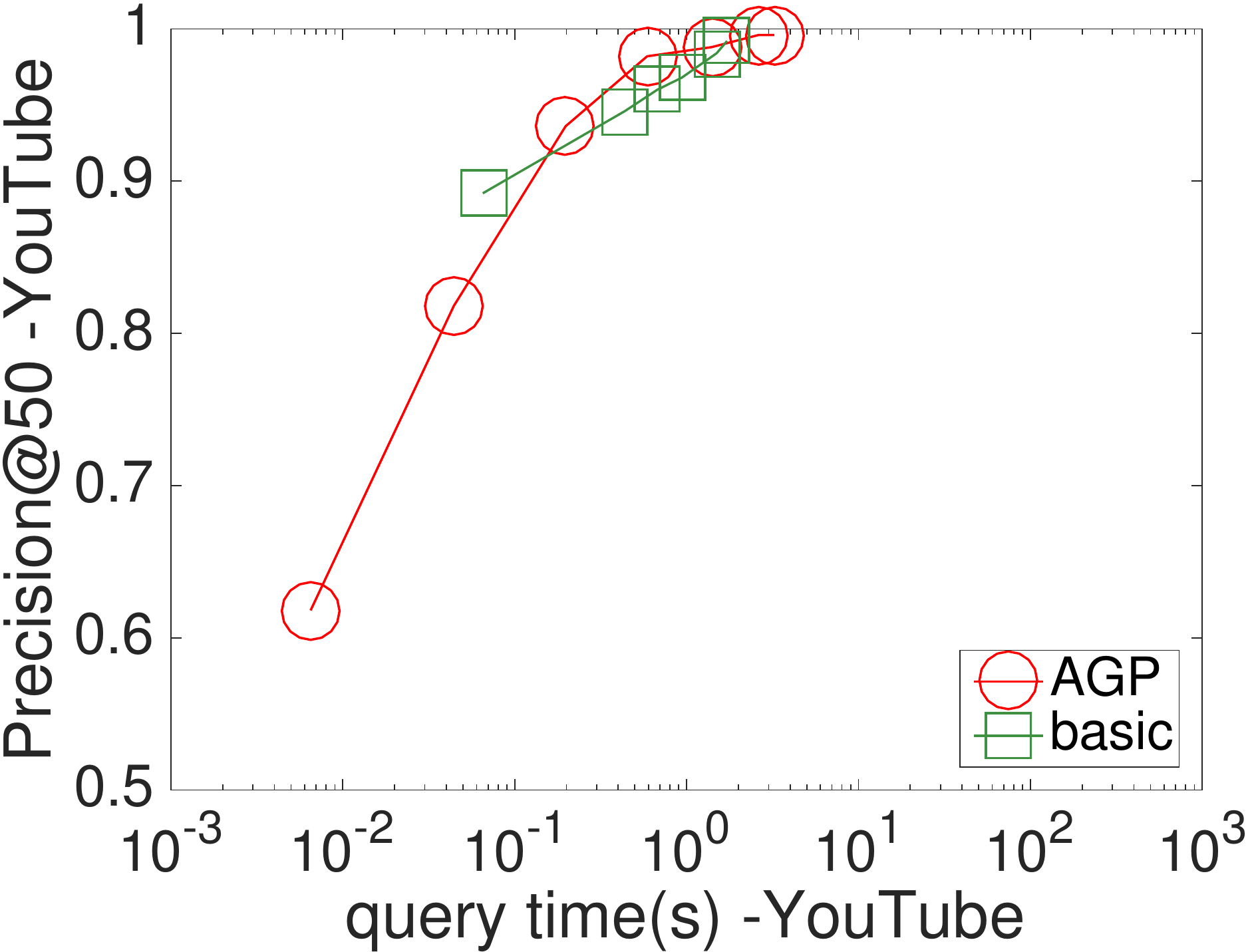} &
			\hspace{-4mm} \includegraphics[height=33.8mm]{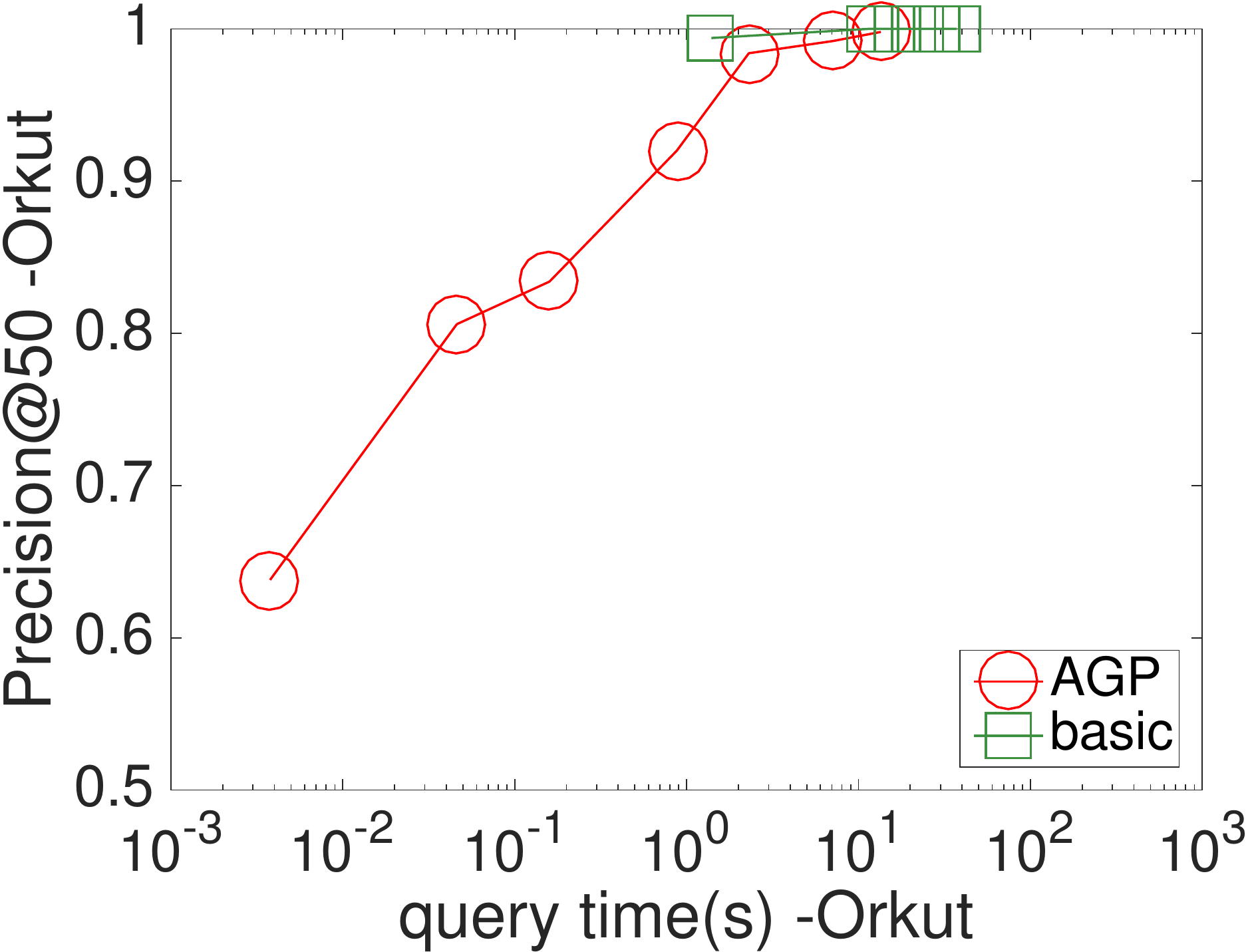} &
			\hspace{-4mm} \includegraphics[height=33.8mm]{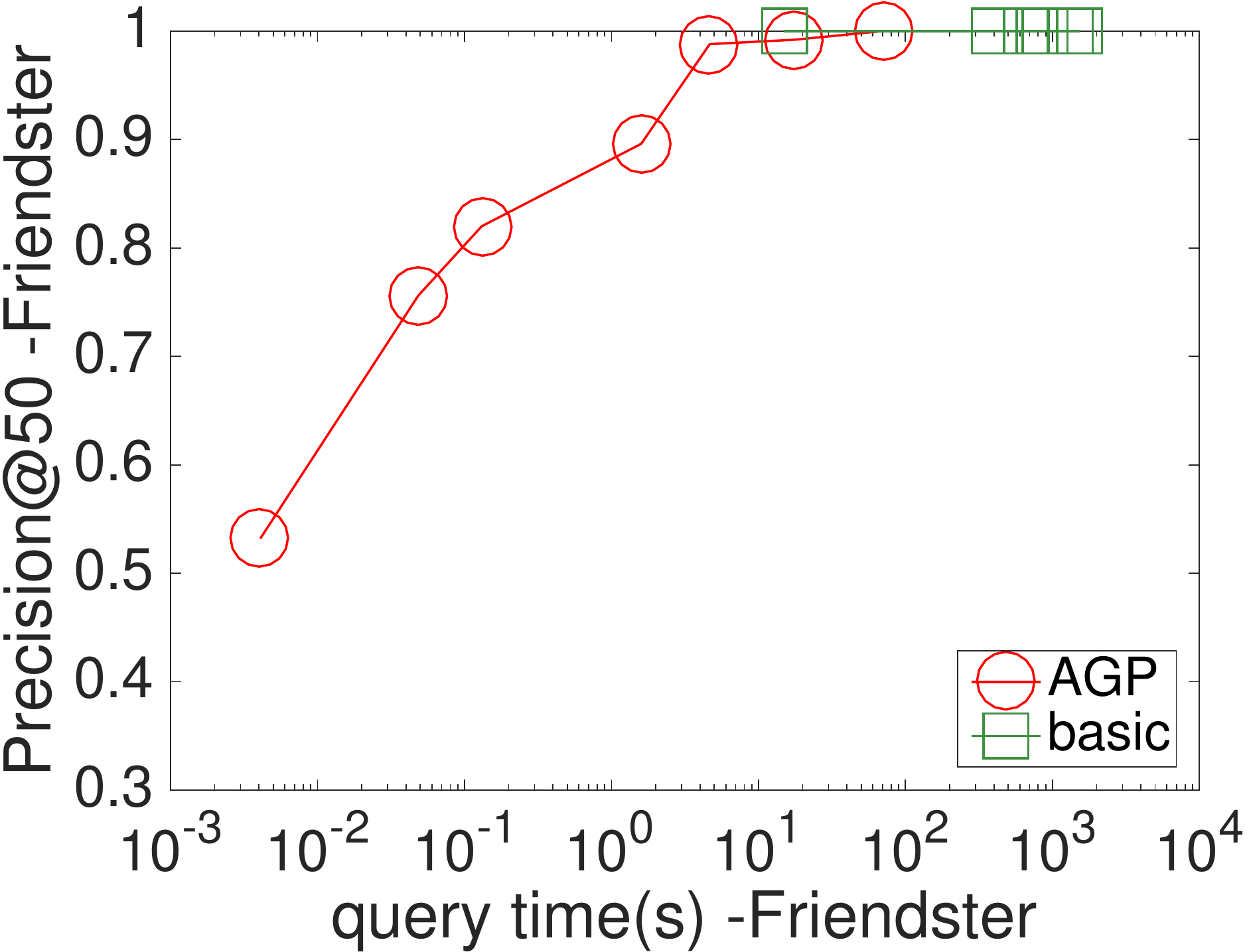} &
			\hspace{-4mm} \includegraphics[height=33.8mm]{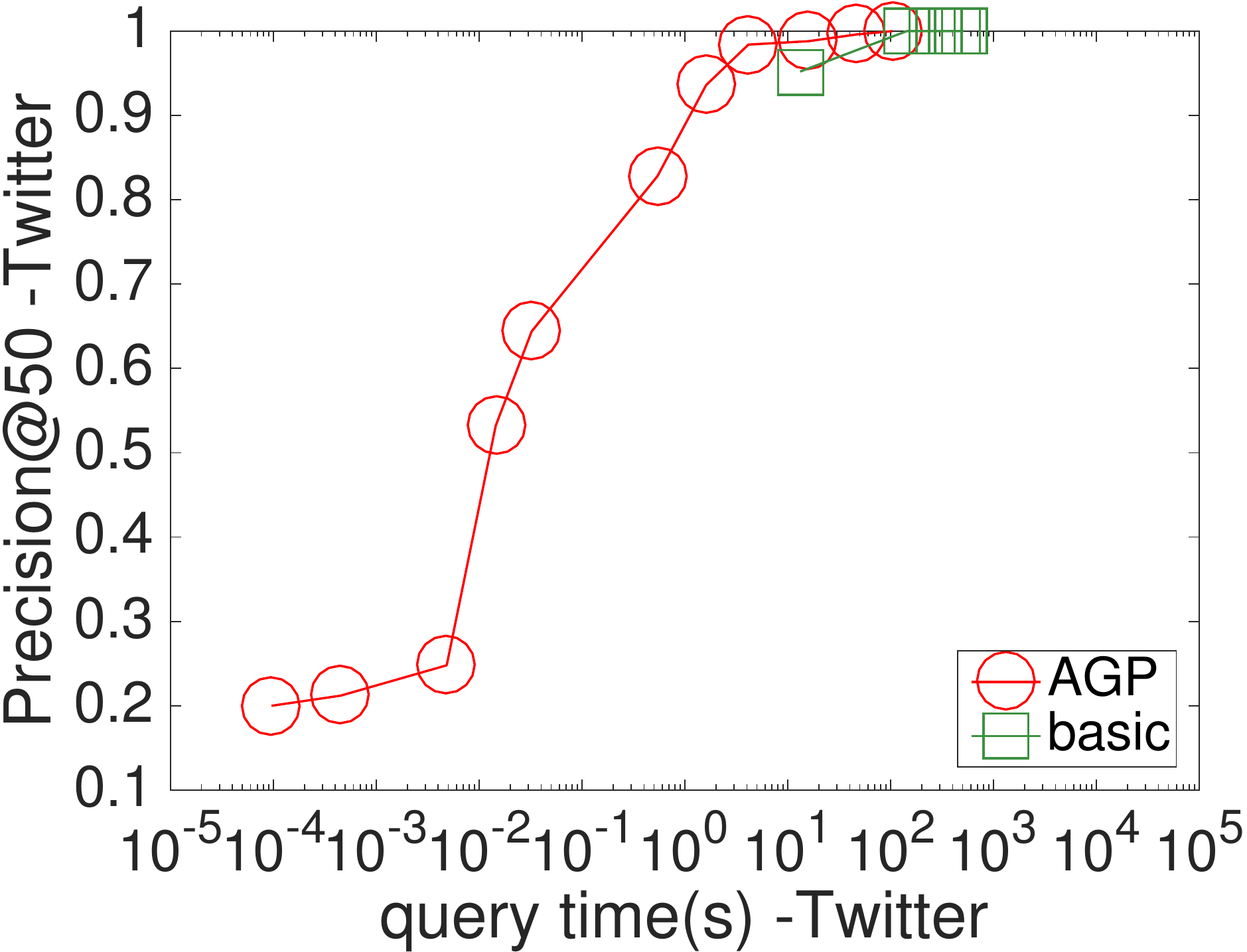} 
		\end{tabular}
		\vspace{-5mm}
		\caption{Tradeoffs between {\em Precision@50} and query time of Katz.}
		\label{fig:Katz-precision-query}
	    \vspace{-2mm}
	\end{small}
	\end{minipage}
\end{figure*}

\subsection{Node classification with GNN}

\header{\bf Datasets. }
Following~\cite{zeng2019graphsaint,zou2019layer}, we perform inductive node classification on Yelp, Amazon and Reddit, and semi-supervised transductive node classification on Papers100M. More specifically, for inductive node classification tasks, we train the model on a graph with labeled nodes and predict nodes' labels on a testing graph. For semi-supervised transductive node classification tasks, we train the model with a small subset of labeled nodes and predict other nodes' labels in the same graph. We follow the same training/validation/test-ing split as previous works in GNN~\cite{zeng2019graphsaint,hu2020ogb}.

%\end{comment}

\header{\bf GNN models.} 
We first consider three proximity-based GNN models: APPNP~\cite{Klicpera2018APPNP},SGC~\cite{wu2019SGC}, and GDC~\cite{klicpera2019GDC}. We augment the three models with the AGP Algorithm~\ref{alg:AGP-RQ} to obtain three variants: APPNP-AGP, SGC-AGP and GDC-AGP. Take SGC-AGP as an example. Recall that SGC uses $\mathbf{Z}=\left(\mathbf{D}^{-\frac{1}{2}} \mathbf{A}\mathbf{D}^{-\frac{1}{2}} \right)^L \hspace{-1mm}\cdot \X$ to perform feature aggregation, where $\X$ is the $n\times d$ feature matrix. SGC-AGP treats each column of $\X$ as a graph signal $\bm{x}$ and perform randomized propagation algorithm (Algorithm~\ref{alg:AGP-RQ}) with predetermined error parameter $\delta$ to obtain the the final representation $\mathbf{Z}$. To achieve high parallelism, we perform propagation for multiple columns of $\mathbf{X}$ in parallel. Since APPNP and GDC's original implementation cannot scale on billion-edge graph Papers100M, we implement APPNP and GDC in the AGP framework. In particular, we set $\varepsilon = 0$ in Algorithm~\ref{alg:AGP-RQ} to obtain the exact propagation matrix $\mathbf{Z}$, in which case the approximate models APPNP-AGP and GDC-AGP essentially become the exact models APPNP and GDC. We set $L\hspace{-1mm}=\hspace{-1mm}20$ for GDC-AGP and APPNP-AGP, and $L\hspace{-1mm}=\hspace{-1mm}10$ for SGC-AGP. Note that SGC suffers from the over-smoothing problem when the number of layers $L$ is large~\cite{wu2019SGC}. We vary the parameter $\varepsilon$ to obtain a trade-off curve between the classification accuracy and the computation time.

%\begin{comment}
\begin{figure*}[t]
	%\begin{minipage}[t]{1\textwidth}
	%	\begin{small}
	\centering
	%\vspace{-1mm}
	%    \begin{footnotesize}
	\begin{tabular}{cccc}
		%\multicolumn{4}{c}{\hspace{-4mm} \includegraphics[height=5mm]{./Figs/legend_large.eps}} \vspace{-1mm} \\
		%\hspace{-3mm} \includegraphics[height=25mm]{./Figs/HKPR-conductance-query-DB.eps} &
		\hspace{-4mm} \includegraphics[height=34mm]{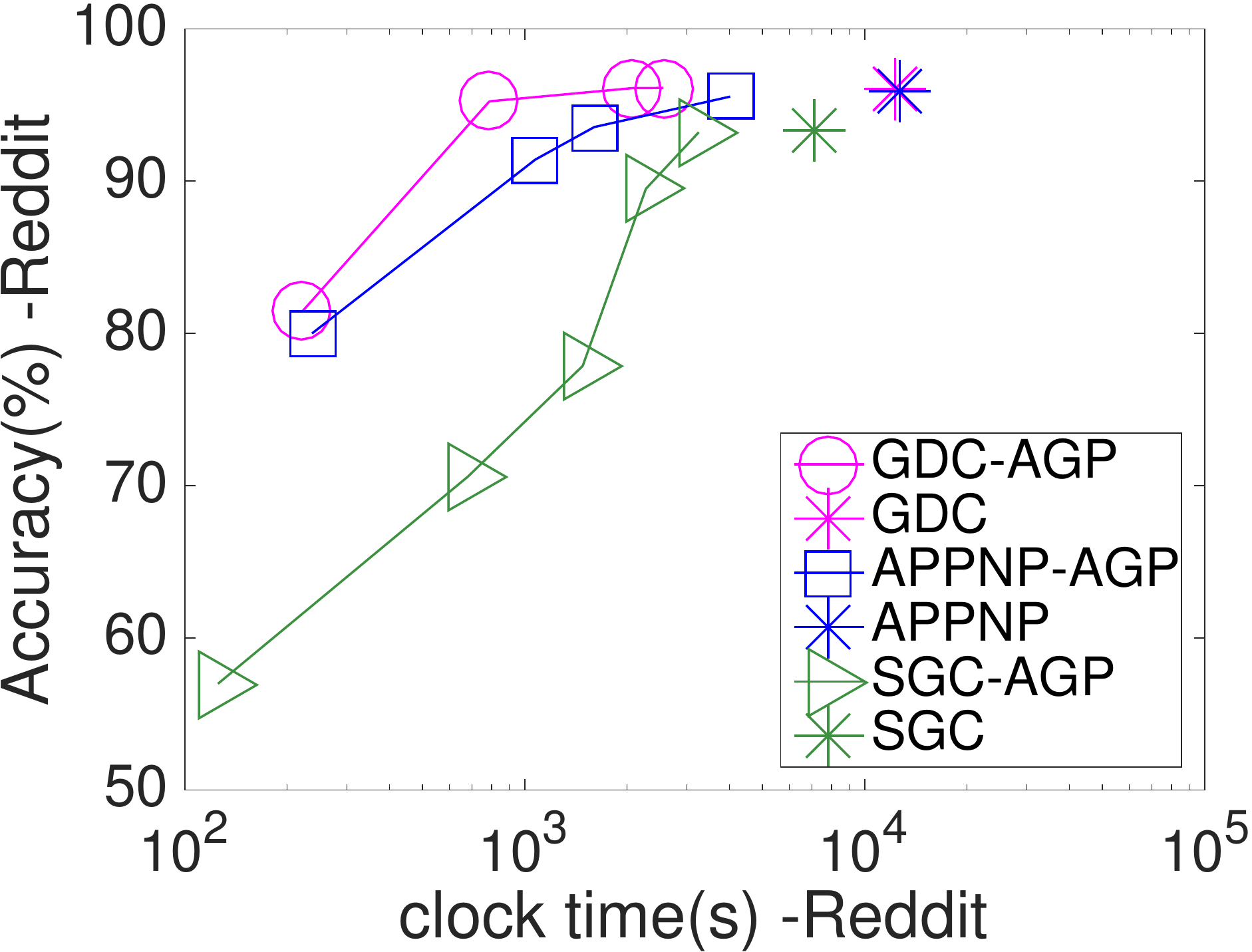} &
		\hspace{-2mm} \includegraphics[height=34mm]{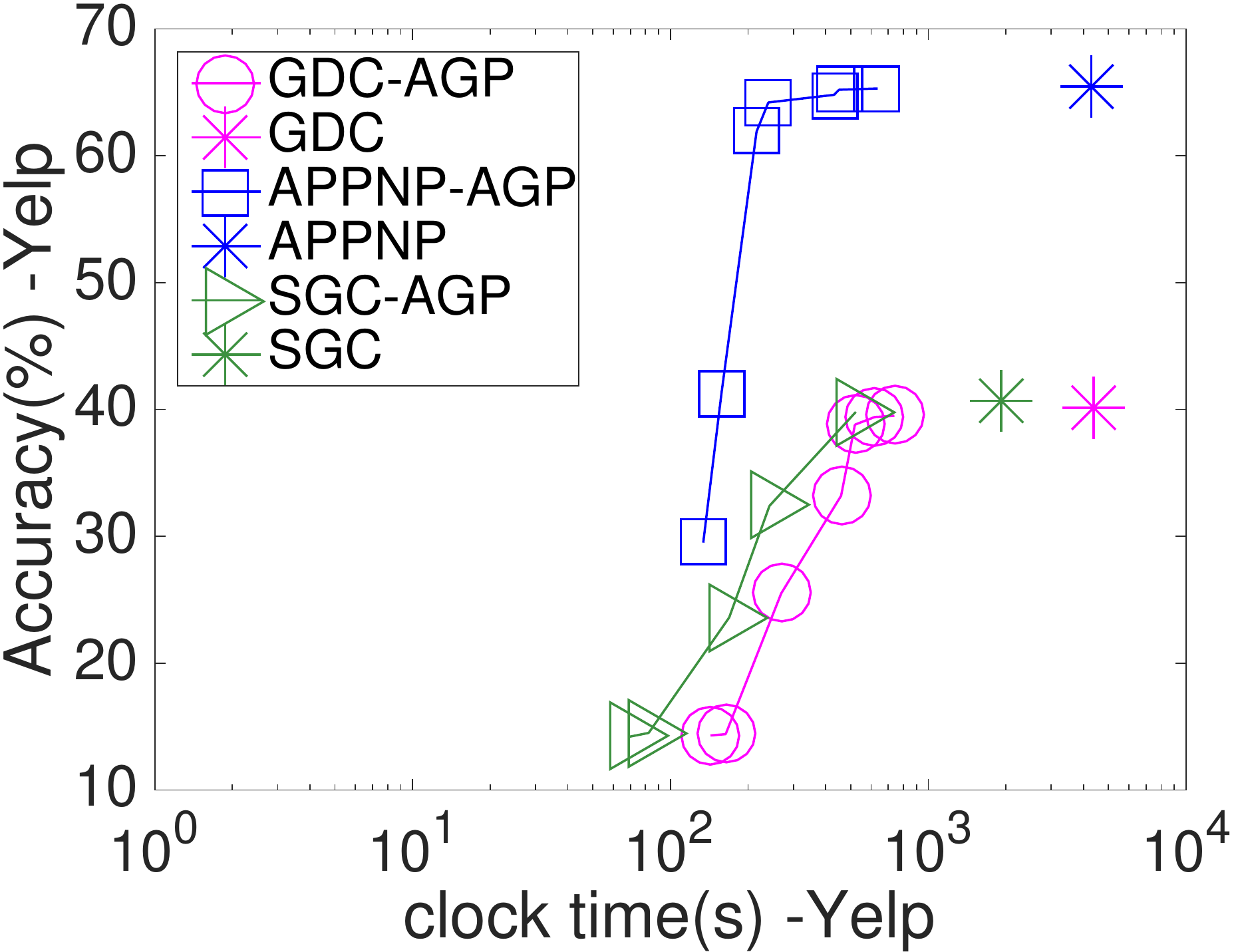} &
		\hspace{-4mm} \includegraphics[height=34mm]{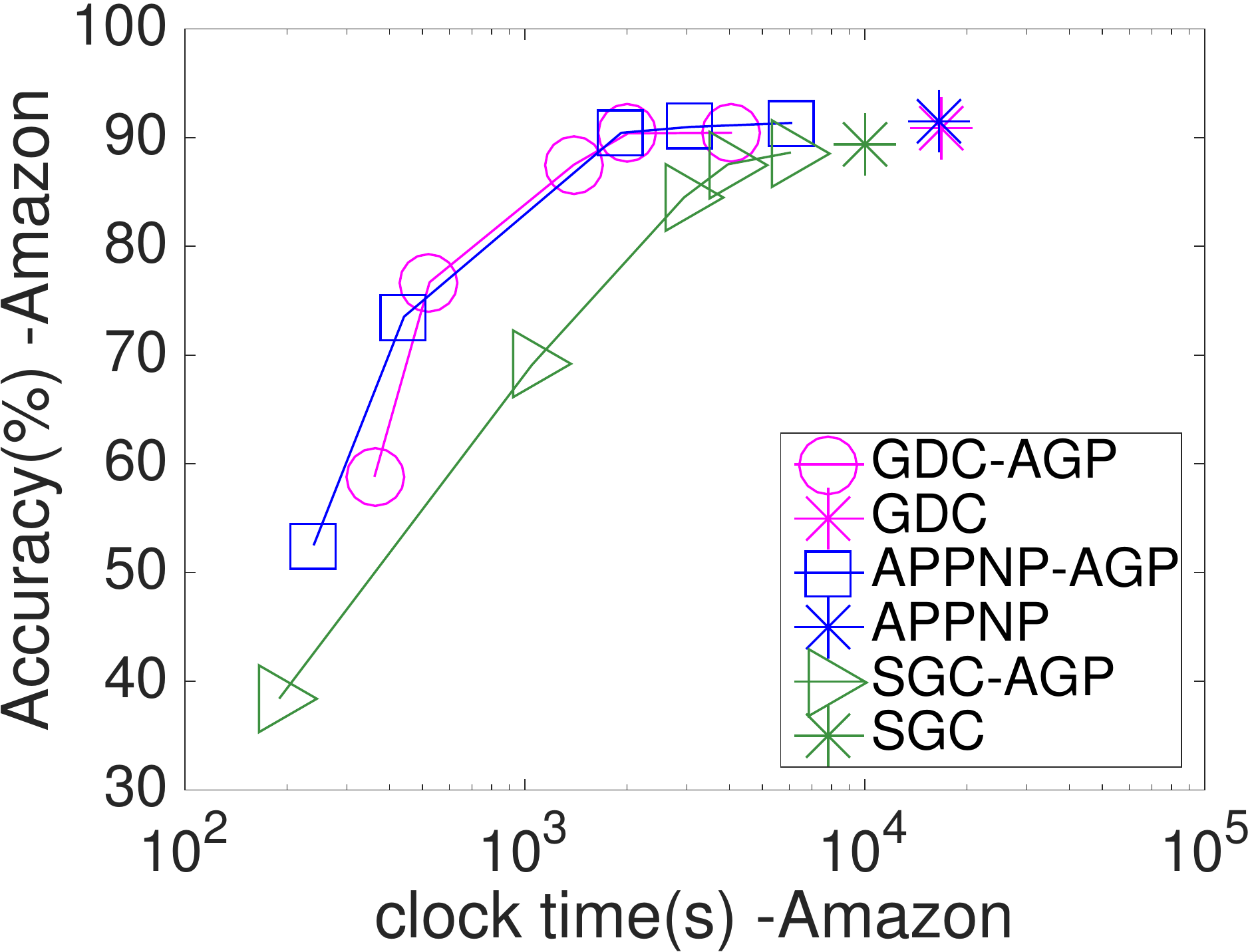} &
		\hspace{-4mm} \includegraphics[height=34mm]{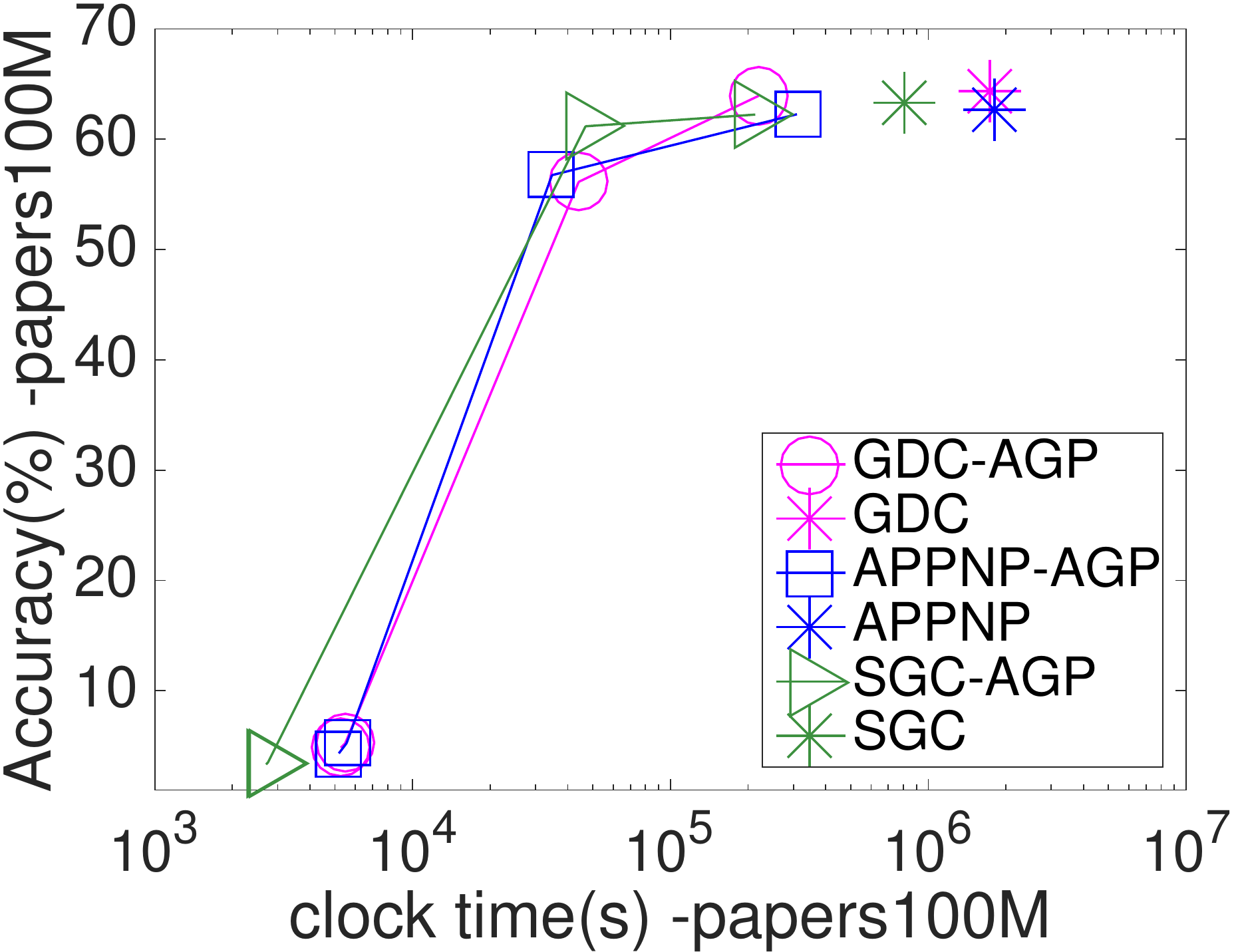} 
	\end{tabular}
	\vspace{-5mm}
	\caption{Tradeoffs between {\em Accuracy(\%)} and clock time in node classification.}
	\label{fig:GNN-accuracy-clock-time}
	\vspace{-2mm}
	%	\end{small}
	%\end{minipage}
\end{figure*}
%\end{comment}

\begin{figure}[t]
	%\begin{minipage}[t]{0.52\textwidth}
	%\begin{small}
	%\vspace{-2mm}
	%    \begin{footnotesize}
	\begin{tabular}{cc}
		%\centering
		%\multicolumn{4}{c}{\hspace{-4mm} \includegraphics[height=5mm]{./Figs/legend_large.eps}} \vspace{-1mm} \\
		%\hspace{-3mm} \includegraphics[height=25mm]{./Figs/HKPR-conductance-query-DB.eps} &
		%\hspace{-4mm} %\includegraphics[height=34mm]{./Figs/GNN-accuracy-memory-Yelp.eps} &
		\hspace{-3mm} \includegraphics[height=32mm]{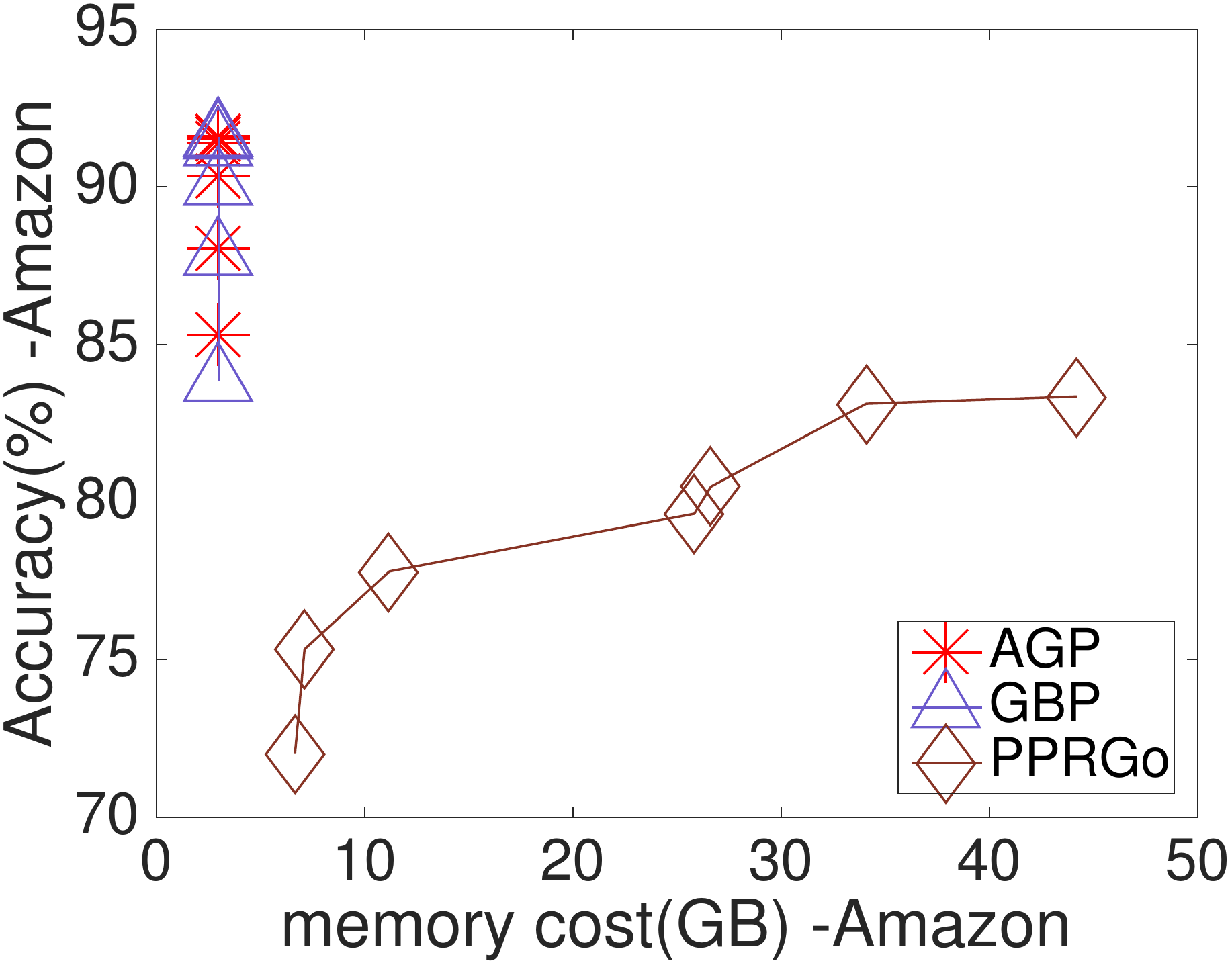} &
		%\hspace{-4mm} %\includegraphics[height=34mm]{./Figs/GNN-accuracy-memory-Reddit.eps} &
		\hspace{-1mm} \includegraphics[height=32mm]{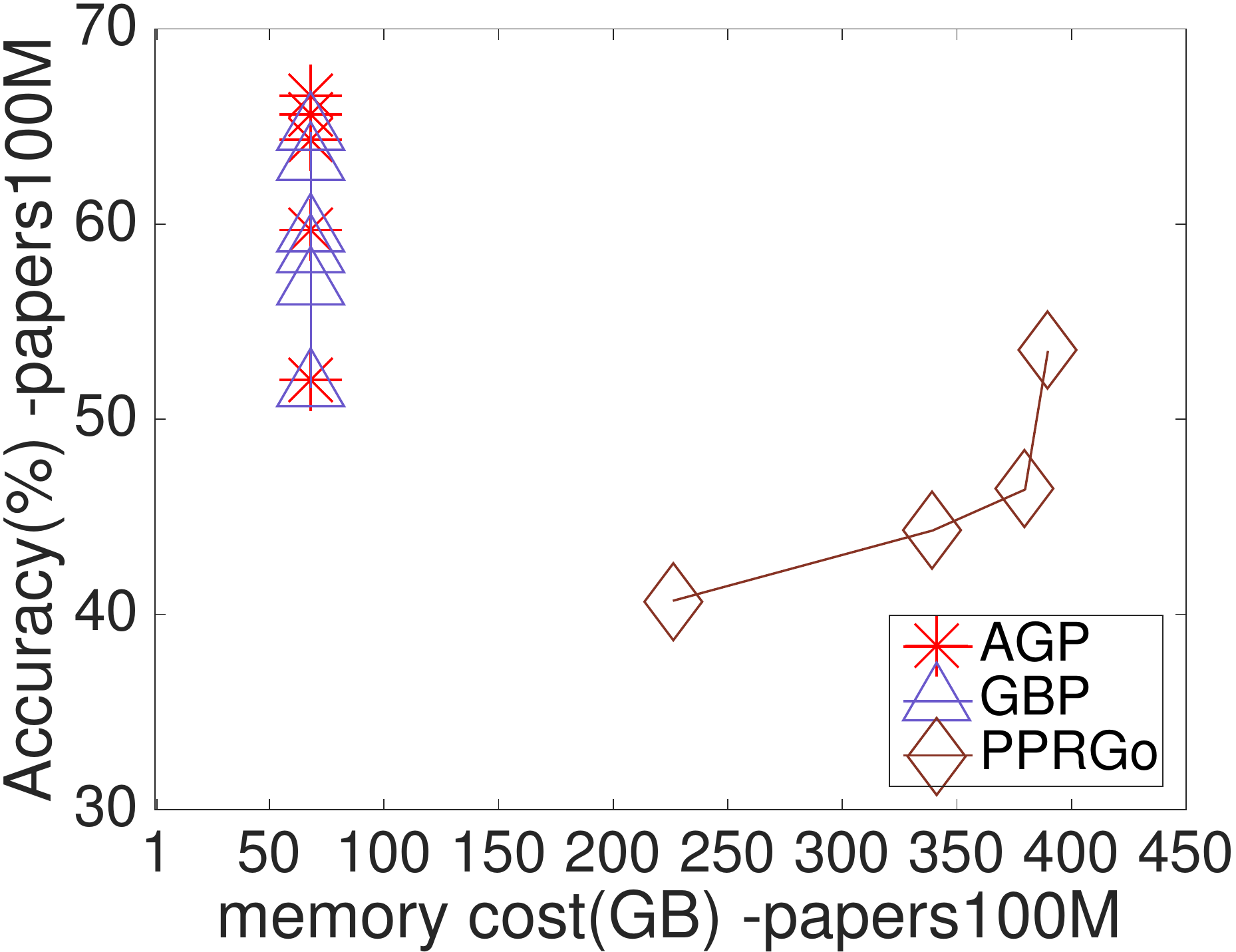} 
	\end{tabular}
	\vspace{-5mm}
	\caption{Tradeoffs between {\em Accuracy(\%)} and memory cost in node classification.}
	\label{fig:GNN-accuracy-memory}
	\vspace{-1mm}
	%\end{small}
\end{figure}
%\end{minipage}

Besides, we also compare AGP with three scalable methods: PPRGo~\cite{bojchevski2020scaling}, GBP~\cite{chen2020GBP}, and ClusterGCN~\cite{chiang2019clusterGCN}. Recall that PPRGo is an improvement work of APPNP. It has three main parameters: the number of non-zero PPR values for each training node $k$, the number of hops $L$, and the residue threshold $r_{max}$. We vary the three parameters $(k,L,r_{max})$ from $(32,2,0.1)$ to $(64,10,10^{-5})$. GBP decouples the feature propagation and prediction to achieve high scalability. In the propagation process, GBP has two parameters: the propagation threshold $r_{max}$ and the level $L$. We vary $r_{max}$ from $10^{-4}$ to $10^{-10}$, and set $L=4$ following~\cite{chen2020GBP}. ClusterGCN uses graph sampling method to partition graphs into small parts, and performs the feature propagation on one randomly picked sub-graph in each mini-batch. We vary the partition numbers from $10^4$ to $10^5$, and the propagation layers from $2$ to $4$. %For AGP, we vary $\delta$ from $10^{-5}$ to $10^{-10}$, and tune $a,b,w_i$ for the best performance. 

For each method, we apply a neural network with 4 hidden layers, trained with mini-batch SGD. %The batch size is set to be $10000$ for faster training time and better generalization ability.
We employ initial residual connection~\cite{He2016ResNet} across the hidden layers to facilitate training. We use the trained model to predict each testing node's labels and take the mean accuracy after five runs. %For GDC, APPNP, and SGC, we divide the computation time into two parts: the {\em preprocessing time} for computing $\mathbf{Z}$, and the {\em training time} for performing mini-batch SGD on $\mathbf{Z}$ until convergence. 
All the experiments in this section are conducted on a machine with an NVIDIA RTX8000 GPU (48GB memory), Intel Xeon CPU (2.20 GHz) with 40 cores, and 512 GB of RAM.

\header{\bf Detailed setups. }
In Table~\ref{tbl:parameters}, we summarize the hyper-parameters of GDC, APPNP, SGC and the corresponding AGP models. Note that the parameter $t$ is for GDC and GDC-AGP, $\alpha$ is for APPNP and APPNP-AGP and $L$ is for SGC and SGC-AGP. We set $a=b=\frac{1}{2}$. In Table~\ref{tbl:para-GBP}, Table~\ref{tbl:para-PPRGo} and Table~\ref{tbl:para-ClusterGCN}, we summarize the hyper-parameters of GBP, PPRGo and ClusterGNN. For ClusterGCN, we set the number of clusters per batch as $10$ on Amazon, following~\cite{chiang2019clusterGCN}. For AGP, we set $a=0.8, b=0.2, w_i=\alpha(1-\alpha)^i$ on Amazon, and $a=b=0.5, w_i=e^{-t}\cdot \frac{t^i}{i!}$ on Papers100M. The other hyper-parameters of AGP are the same as those in Table~\ref{tbl:parameters}. In Table~\ref{tbl:url}, we list the available URL of each  method. 
%In Table~\ref{tbl:parameters}, Table~\ref{tbl:para-GBP}, Table~\ref{tbl:para-PPRGo} and Table~\ref{tbl:para-ClusterGCN}, we summarize the hyper-parameters of AGP, GBP, PPRGo and ClusterGNN used in Figure~\ref{fig:GNN-accuracy-query} and Figure~\ref{fig:GBP}. 
%Note that for GDC and APPNP, we set $r_{max}=0$ to obtain the exact propagation results, because the original implementation of GDC and APPNP don't support the billion-edge graphs.

%\clearpage
\section{Additional experimental results} \label{sec:appendix_old}

\subsection{Local clustering with HKPR}
%\header{\bf Comparison of overhead memory.} 
%Apart from the experiments shown in Section~\ref{subsec:clustering}, we further explore the memory cost of each method. 
Apart from using {\em MaxError} to measure the approximation quality, we further explore the trade-off curves between {\em Precision@k} and query time. 
%We use {\em Precision@k} as our metric to measure the approximation quality of each method. 
Let $V_k$ denote the set of $k$ nodes with highest normalized HKPR values, and $\hat{V}_k$ denote the estimated top-$k$ node set returned by an approximate method. {\em Normalized Precision@k} is defined as the percentage of nodes in $\hat{V}_k$ that coincides with the actual top-$k$ results $V_k$. {\em Precision@k} can evaluate the accuracy of the relative node order of each method. Similarly, we use the Basic Propagation Algorithm~\ref{alg:AGP-deter} with $L=50$ to obtain the ground truths of the normalized HKPR. %On directed graph, $d_v$ is substituted by the out-degree $d_{out}(v)$. 
Figure~\ref{fig:HKPR-precision-query} plots the trade-off curve between {\em Precision@50} and the query time for each method. We omit TEA and TEA+ on Twitter as they cannot handle directed graphs. We observe that AGP achieves the highest precision among the five approximate algorithms on all four datasets under the same query time. 
%Figure~\ref{fig:HKPR-conductance-mem} plots the trade-off lines between conductance and overhead memory. We exclude the space cost by the input graph in overhead memory. We can observe that the memory cost of AGP is relative small among these competitors. On {\em Friendster}, AGP can save $10\times$ space than TEA and TEA+. Moreover, we find that the overhead memory of each method is merely unchanged, which means the error bound only has a limited influence on memory cost. 

Besides, we also conduct experiments to present the influence of heat kernel parameter $t$ on the experimental performances. Figure~\ref{fig:conductance-query-OL} plots the conductance and query time trade-offs on {\em Orkut}, with $t$ varying in $\{5,10,20,40\}$. Recall that $t$ is the average length of the heat kernel random walk. Hence, the query time of each method increases as $t$ varying from 5 to 40.  We observe that AGP consistently achieves the lowest conductance with the same amount of query time. Furthermore, as $t$ increases from $5$ to $40$, AGP's query time only increases by $7\times$, while TEA and TEA+ increase by $10\times-100\times$, which demonstrates the scalability of AGP. 

%Besides, we also present Figure~\ref{fig:HKPR-conductance-query-YT} to show the conductance trade-offs with varying heat kernel parameter $t$ on {\em Youtube}. Observe that AGP can consistently outperform other competitors with varying $t$, which reflects the stability of AGP with different parameters and datasets.

%\vspace{-1mm}
\subsection{Evaluation of Katz index}

%\begin{minipage}[t]{0.49\textwidth}
\begin{table}[t]
	%\vspace{-1mm}
	\centering
	\tblcapup
	\caption{preprocessing and training time}
	\vspace{-4mm}
	\tblcapdown
	\begin{small}
		\begin{tabular}{|c|r|r|r|r|} %p{0.8in}|}
			%\begin{tabular}{|p{1.8cm}|p{1.5cm}|p{1.3cm}|p{1.3cm}|p{0.7cm}|}
			\hline
			{\bf } & {\bf APPNP-AGP }\hspace{-1mm} & {\bf GDC-AGP} & {\bf SGC-AGP}& \makecell[c]{{\bf PPRGo}} \\ \hline
			\hspace{-1mm} \makecell[c]{preprocessing\\ time (s)} \hspace{-2mm} & 9253.85 \hspace{-1mm} &   6807.17& 1437.95 & 7639.72  \\ \hline
			\hspace{-1mm}training time (s)  \hspace{-2mm}  & 166.19  \hspace{-1mm}   & 175.62   & 120.07 & 140.67 \\
			\hline
		\end{tabular}
	\end{small}
	\label{tbl:propagation-training}
	%\tbldown
	%\vspace{-3mm}
\end{table}
%\end{minipage}

%\begin{comment}
%\header{\bf Katz index.}
We also evaluate the performance of AGP to compute the Katz index. Recall that Katz index numerates the paths of all lengths between a pair of nodes, which can be expressed as that $\vec{\pi}=\sum_{i=0}^\infty \beta^i \cdot \mathbf{A}^{i} \cdot \vec{e}_s$. 
In our experiments, $\beta$ is set as $\frac{0.85}{\lambda_1}$ to guarantee convergence, where $\lambda_1$ denotes the largest eigenvalue of the adjacent matrix $\bm{A}$. We compare the performance of AGP with the basic propagation algorithm given in Algorithm~\ref{alg:AGP-deter} (denoted as basic in Figure~\ref{fig:Katz-MaxError-query}, ~\ref{fig:Katz-precision-query}). We treat the results computed by the basic propagation algorithm with $L=50$ as the ground truths. Varying $\delta$ from $0.01$ to $10^{8}$, Figure~\ref{fig:Katz-MaxError-query} shows the trade-offs between {\em MaxError} and query time. Here we define $MaxError =\max_{v \in V}\left|\vec{\pi}(v)-\vec{\epi}(v)\right|$. We issue 50 query nodes and return the average {\em MaxError} of all query nodes same as before. We can observe that AGP costs less time than basic propagation algorithm to achieve the same error. Especially when $\delta$ is large, such as $MaxError=10^{-5}$ on the dataset {\em Orkut}, AGP has a $10\times-100\times$ speed up than the basic propagation algorithm. Figure~\ref{fig:Katz-precision-query} plots the trade-off lines between {\em Precision@50} and query time. The definition of {\em Precision@k} is the same as that in Figure~\ref{fig:HKPR-precision-query}, which equals the percentage of nodes in the estimated top-k set that coincides with the real top-k nodes. 
Note that the basic propagation algorithm can always achieve precision 1 even with large {\em MaxError}. This is because the propagation results derived by the basic propagation algorithm are always smaller than the ground truths. The biased results may present large error and high precision simultaneously by maintaining the relative order of top-k nodes. While AGP is not a biased method, the precision will increase with the decreasing of {\em MaxError}. 

%\end{comment}

%large error, but the relative order of top-k nodes may be accurate. 
%This concurs with our analysis that the basic propagation algorithm can return high accuracy results with expensive time cost. 

\vspace{-1mm}
\subsection{Node classification with GNN}
\header{\bf Comparison of clock time.}
To eliminate the effect of parallelism, we plot the trade-offs between clock time and classification accuracy in Figure~\ref{fig:GNN-accuracy-clock-time}. We can observe that AGP still achieves a $10\times$ speedup on each dataset, which concurs with the analysis for propagation time. Besides, note that every method presents a nearly $30\times$ speedup after parallelism, which reflects the effectiveness of parallelism.

%\begin{figure}[h]
%	\begin{small}
%		\centering
%		\vspace{-4mm}
%		%    \begin{footnotesize}
%		\begin{tabular}{cccc}
%			%\multicolumn{4}{c}{\hspace{-4mm} %\includegraphics[height=5mm]{./Figs/legend_large.eps}} \vspace{-1mm} \\
%			%\hspace{-3mm} %\includegraphics[height=25mm]{./Figs/HKPR-conductance-query-DB.eps} &
%			\hspace{-4mm} \includegraphics[height=45mm]{./Figs/bbb.eps} &
%		\end{tabular}
%		\vspace{-3mm}
%		\caption{Propagation time v.s. training time}
%		\label{fig:GNN-propagation-training}
%		\vspace{-4mm}
%	\end{small}
%\end{figure}

\header{\bf Comparison of preprocessing time and training time. }
Table~\ref{tbl:propagation-training} shows the comparison between the training time and preprocessing time on Papers100M. Due to the large batch size ($10,000$), the training process is generally significantly faster than the feature propagation process. Hence, we recognize feature propagation as the bottleneck for scaling GNN models on large graphs, motivating our study on approximate graph propagation.

\header{\bf Comparison of memory cost. }
Figure~\ref{fig:GNN-accuracy-memory} shows the memory overhead of AGP, GBP, and PPRGo. Recall that the AGP Algorithm~\ref{alg:AGP-RQ} only maintains two $n$ dimension vectors: the residue vector $\vec{r}$ and the reserve vector $\vec{q}$. Consequently, AGP only takes a fixed memory, which can be ignored compared to the graph's size and the feature matrix. Such property is ideal for scaling GNN models on massive graphs. %On the other hand, PPRGo requires a large memory size to store the intermediate local push results from each training node.

\balance
\section{Proofs} \label{sec:proofs}

%\subsection{Chernoff Bound} \label{sec:chernoff}
%\begin{lemma}[Chernoff Bound \cite{ChungL06}] \label{lmm:chernoff}
%	For a set $\{x_i\}$ ($i \in [1, n_r]$) of i.i.d.\ random variables with mean $\mu$ and $x_i \in [0, 1]$,
%	$$\Pr\left[\left|{1\over n_r}\sum_{i=1}^{n_r} x_i - \mu\right| \geq \e\right] \leq \exp\left(-\dfrac{n_r \cdot \e^2}{\frac{2}{3}\e + 2\mu}\right).$$
%	\end{lemma}

%\subsection{Bernstein Inequality} \label{sec:Bernstein}
%\begin{lemma}[Bernstein inequality~\cite{ChungL06}]\label{lem:conc}
%	Let $X_1, \cdots, X_R$ be independent  random variables with
%	$|X_i| <  b$ for $i=1,\ldots, R$. Let
%	$X=\frac{1}{R}\cdot\sum^R_{i=1}X_i$, we have
%	\begin{equation}\nonumber
%	\Pr[|X-\E[X]|\ge \lambda] \le 2\cdot \exp\left(-\frac{\lambda^2\cdot
%		R}{2R\cdot \Var[X] + 2b\lambda/3}\right),
%	\end{equation}
%	where $\Var[X]$ is the variance of $X$.
%\end{lemma}

\subsection{Chebyshev's Inequality} \label{sec:chebyshev}
%\vspace{-1mm}
\begin{lemma}[Chebyshev's inequality] \label{lmm:chebysev}
	Let $X$ be a random variable, then $\Pr\left[\left| X -\E[X]\right| \geq \e\right] \le {\Var[X] \over \e^2 }. $
\end{lemma}

\subsection{Further Explanations on Assumption~\ref{asm:L}}
Recall that in Section~\ref{sec:BPA}, we introduce Assumption~\ref{asm:L} to guarantee we only need to compute the prefix sum $\vec{\pi}=\sum_{i=0}^L w_i \cdot \left(\D^{-a}\A \D^{-b} \right)^i$ to achieve the relative error in Definition~\ref{def:pro-relative}, where $L=O\left(\log{\frac{1}{\delta}}\right)$. The following theorem offers a formal proof of this property. 
%The following lemma shows under these assumptions, we only need to compute the prefix sum $\hat{\vec{\pi}}=\sum_{i=0}^L  w_i \cdot \hspace{-0.5mm} \left(\D^{-a}\A \D^{-b} \right)^i \hspace{-1.5mm}\cdot \vec{x}$ to achieve the relative error in Definition~\ref{def:pro-relative}. 
\begin{theorem}\label{thm:prefix}
According to the assumptions on $w_i$ and $\vec{x}$ in Section~\ref{sec:BPA}, to achieve the relative error in Definition~\ref{def:pro-relative}, we only need to approximate the prefix sum $\vec{\pi}_L=\sum_{i=0}^L  w_i \cdot \hspace{-0.5mm} \left(\D^{-a}\A \D^{-b} \right)^i \hspace{-1.5mm}\cdot \vec{x}$, such that for any $v\in V$ with $\vec{\pi}_L(v)>\frac{18}{19}\cdot \delta$, we have $\left|\vec{\pi}_L(v)-\hat{\vec{\pi}}(v)\right|\le \frac{1}{20} \vec{\pi}_L(v)$ holds with high probability. 
\end{theorem}

\begin{proof}
We first show according to Assumption~\ref{asm:L} and the assumption $\|\vec{x}\|_1=1$, by setting $L=O\left(\log{\frac{1}{\delta}}\right)$, we have: 
\begin{align}\label{eqn:partsum}
  \left\|\sum_{i=L+1}^\infty \hspace{-0.5mm} w_i \cdot \hspace{-0.5mm} \left(\D^{-a}\A \D^{-b} \right)^i \hspace{-1.5mm}\cdot \vec{x}\right\|_2 \le \frac{\delta}{19}. 
\end{align}

Recall that in Assumption~\ref{asm:L}, we assume $w_i \cdot \lambda_{max}^i$ is upper bounded by $\lambda^i$ when $i\ge L_0$ and $L_0\ge 1$ is a constant. $\lambda_{max}$ denotes the maximum eigenvalue of the transition probability matrix $\D^{-a}\A \D^{-b}$ and $\lambda<1$ is a constant. Thus, for $\forall L \ge L_0$, we have: 
\begin{align*}
    \left\|\sum_{i=L+1}^\infty \hspace{-1.5mm} w_i \hspace{-0.5mm} \cdot \hspace{-0.5mm} \left(\D^{-a}\A \D^{-b}\right)^i \hspace{-1.5mm}\cdot \hspace{-0.5mm}\vec{x} \right\|_2 \hspace{-1.5mm}\le \left\|\sum_{i=L+1}^\infty \hspace{-1.5mm} w_i \hspace{-0.5mm} \cdot \hspace{-0.5mm} \lambda_{max}^i \hspace{-1mm} \cdot \hspace{-0.5mm} \vec{x} \right\|_2 \hspace{-1.5mm} \le \hspace{-1.5mm}\sum_{i=L+1}^\infty \hspace{-1mm} \lambda^i \hspace{-0.5mm} \cdot \hspace{-0.5mm} \left\|\vec{x} \right\|_2 \hspace{-1mm}\le \hspace{-2mm}\sum_{i=L+1}^\infty \hspace{-1.5mm} \lambda^i. 
\end{align*}
In the last inequality, we use the fact $\|\vec{x}\|_2\le \|\vec{x}\|_1=1$. By setting $L=\max\left\{L_0,O\left(\log_{\lambda}{\frac{(1-\lambda)\cdot \delta}{19}}\right)\right\}=O(\log{\frac{1}{\delta}})$, we can derive: 
\begin{align}\nonumber %\label{eqn:deltabound}
\left\|\sum_{i=L+1}^\infty \hspace{-0.5mm} w_i \hspace{-0.5mm} \cdot \hspace{-0.5mm} \left(\D^{-a}\A \D^{-b}\right)^i \hspace{-1.5mm}\cdot \hspace{-0.5mm}\vec{x} \right\|_2 \le \sum_{i=L+1}^\infty \hspace{-1mm} \lambda^i=\frac{\lambda^{L+1}}{1-\lambda} \le \frac{\delta}{19}, 
\end{align}
which follows the inequality~\eqref{eqn:partsum}. 

Then we show that according to the assumption on the non-negativity of $\vec{x}$ and the bound $\left\|\sum_{i=L+1}^\infty \hspace{-0.5mm} w_i \cdot \hspace{-0.5mm} \left(\D^{-a}\A \D^{-b} \right)^i \hspace{-1.5mm}\cdot \vec{x}\right\|_2 \le \frac{\delta}{19}$, to achieve the relative error in Definition~\ref{def:pro-relative}, we only need to approximate the prefix sum $\vec{\pi}_L=\sum_{i=0}^L  w_i \cdot \hspace{-0.5mm} \left(\D^{-a}\A \D^{-b} \right)^i \hspace{-1.5mm}\cdot \vec{x}$. Specifically, we only need to return the vector $\hat{\vec{\pi}}$ as the estimator of $\vec{\pi}_L$ such that for any $v\in V$ with $\vec{\pi}_L(v) \ge \frac{18}{19}\cdot \delta$, we have
\begin{align}\label{eqn:apprgoal}
\left|\vec{\pi}_L(v)-\hat{\vec{\pi}}(v)\right|\le \frac{1}{20} \vec{\pi}_L(v)
\end{align} 
holds with high probability. 

Let $\vec{\pi}_L$ denote the prefix sum that $\vec{\pi}_L\hspace{-1mm}=\hspace{-1mm} \sum_{i=0}^L \hspace{-1mm} w_i\cdot \left(\D^{-a}\A \D^{-b} \right)^i \hspace{-1.5mm}\cdot \vec{x}$ and $\vec{\bar{\pi}}_L$ denote the remaining sum that $\vec{\bar{\pi}}_L=\sum_{i=L+1}^\infty w_i\cdot \left(\D^{-a}\A \D^{-b} \right)^i \hspace{-1.5mm}\cdot \vec{x}$. Hence the real propagation vector $\vec{\pi}=\vec{\pi}_L+\vec{\bar{\pi}}_L=\sum_{i=0}^\infty w_i\cdot \left(\D^{-a}\A \D^{-b} \right)^i \hspace{-1.5mm}\cdot \vec{x}$. For each node $v\in V$ with $\vec{\pi}_L(v)>\frac{18}{19}\cdot \delta$, we have: 
\begin{align*}
&\left|\vec{\pi}(v)-\hat{\vec{\pi}}(v)\right|\le \left|\vec{\pi}_L(v)-\hat{\vec{\pi}}(v)\right|+\vec{\bar{\pi}}(v)\\
&\le \frac{1}{20}\vec{\pi}_L(v)+\vec{\bar{\pi}}(v)=\frac{1}{20}\vec{\pi}(v)+\frac{19}{20}\vec{\bar{\pi}}(v). 
\end{align*}
In the first inequality, we use the fact that $\vec{\pi}=\vec{\pi}_L+\vec{\bar{\pi}}_L$ and $\left|\vec{\pi}_L(v)-\hat{\vec{\pi}}(v)+\vec{\bar{\pi}}(v)\right|\le \left|\vec{\pi}_L(v)-\hat{\vec{\pi}}(v)\right|+\vec{\bar{\pi}}(v)$. In the second inequality, we apply inequality~\eqref{eqn:apprgoal} and the assumption that $\vec{x}$ is non-negative. By the bound that $\vec{\bar{\pi}}(v)\le \left\|\sum_{i=L+1}^\infty \hspace{-0.5mm} w_i \hspace{-0.5mm} \cdot \hspace{-0.5mm} \left(\D^{-a}\A \D^{-b}\right)^i \hspace{-1.5mm}\cdot \hspace{-0.5mm}\vec{x} \right\|_2 \le \frac{1}{19}\delta$, we can derive: 
\begin{align*}
\left|\vec{\pi}(v)-\hat{\vec{\pi}}(v)\right|\le \frac{1}{20}\vec{\pi}(v)+\frac{1}{20}\delta \le \frac{1}{10}\vec{\pi}(v). 
\end{align*}
In the last inequality, we apply the error threshold in Definition~\ref{def:pro-relative} that $\vec{\pi}(v)\ge \delta$, and the theorem follows. 
\end{proof}

Even though we derive Theorem~\ref{thm:prefix} based on Assumption~\ref{asm:L}, the following lemma shows that without Assumption~\ref{asm:L}, the property of Theorem~\ref{thm:prefix} is also possessed by all proximity measures discussed in this paper. 
\begin{lemma}
In the proximity models of PageRank, PPR, HKPR, transition probability and Katz, Theorem~\ref{thm:prefix} holds without Assumption~\ref{asm:L}. 
\end{lemma}

\begin{proof}
We first show in the proximity models of PageRank, PPR, HKPR, transition probability and Katz, by setting $L=O\left(\log{\frac{1}{\delta}}\right)$, we can bound
\begin{align}\label{eqn:boundL}
    \left\|\sum_{i=L+1}^\infty \hspace{-0.5mm} w_i \cdot \hspace{-0.5mm} \left(\D^{-a}\A \D^{-b} \right)^i \hspace{-1.5mm}\cdot \vec{x}\right\|_2 \le \frac{\delta}{19}, 
\end{align}
only based on the assumptions that $\vec{x}$ is non-negative and $\|\vec{x}\|_1=1$, without Assumption~\ref{asm:L}.

%Recall that in Section~\ref{sec:BPA}, we assume $\vec{x}$ is non-negative. Hence, $\left\|\sum\limits_{i=L+1}^\infty \hspace{-0.5mm} w_i \cdot \hspace{-0.5mm} \left(\D^{-a}\A \D^{-b} \right)^i \hspace{-1.5mm}\cdot \vec{x}\right\|_1=\sum\limits_{i=L+1}^\infty \hspace{-0.5mm} w_i \cdot  \left\|\hspace{-0.5mm} \left(\D^{-a}\A \D^{-b} \right)^i \hspace{-1.5mm}\cdot \vec{x}\right\|_1$. We first show in the proximity models of PageRank, PPR, HKPR and transition probability, $\sum\limits_{i=L+1}^\infty \hspace{-0.5mm} w_i \cdot  \left\|\hspace{-0.5mm} \left(\D^{-a}\A \D^{-b} \right)^i \hspace{-1.5mm}\cdot \vec{x}\right\|_1$ can be bounded by $\delta$. Then we show that this still holds for Katz.

In the proximity model of PageRank, PPR, HKPR and transition probability, we set $a=b=1$ and the transition probability matrix is $\D^{-a}\A \D^{-b}=\A \D^{-1}$. Thus, the left side of inequality~\eqref{eqn:boundL} becomes $\left\|\sum_{i=L+1}^\infty \hspace{-0.5mm} w_i \cdot \hspace{-0.5mm} \left(\A \D^{-1} \right)^i \hspace{-1.5mm}\cdot \vec{x}\right\|_2=\sum_{i=L+1}^\infty \hspace{-0.5mm} w_i \cdot \left\| \left(\A \D^{-1} \right)^i \hspace{-1.5mm}\cdot \vec{x}\right\|_2$, where we apply the assumption on the non-negativity of $\vec{x}$. Because the maximum eigenvalue of the matrix $\A \D^{-1}$ is $1$, we have $\left\|\left(\A \D^{-1} \right)^i \hspace{-1.5mm}\cdot \vec{x}\right\|_2 \le \|\vec{x}\|_2$, following: 
\begin{align*}
   \sum_{i=L+1}^\infty \hspace{-0.5mm} w_i \cdot \hspace{-0.5mm} \left\|\left(\A \D^{-1} \right)^i \hspace{-1.5mm}\cdot \vec{x}\right\|_2 \le \sum_{i=L+1}^\infty \hspace{-1mm}w_i \cdot \|\vec{x} \|_2 \le \hspace{-1mm} \sum_{i=L+1}^\infty \hspace{-1mm}w_i \cdot \|\vec{x} \|_1 =\hspace{-1mm}\sum_{i=L+1}^\infty \hspace{-1mm} w_i. 
\end{align*}
In the last equality, we apply the assumption that $\|\vec{x}\|_1=1$. Hence, for PageRank, PPR, HKPR and transition probability, we only need to show $\sum_{i=L+1}^\infty w_i \le \frac{\delta}{19}$ holds for $L=O\left(\log{\frac{1}{\delta}}\right)$ . Specifically, for PageRank and PPR, $w_i=\alpha(1-\alpha)^i$. Hence, $\sum_{i=L+1}^\infty w_i =\sum_{i=L+1}^\infty \alpha(1-\alpha)^i=(1-\alpha)^{L+1}$. By setting $L=\log_{1-\alpha}\frac{\delta}{19}=O\left(\log{\frac{1}{\delta}}\right)$, we can bound $\sum_{i=L+1}^\infty w_i$ by $\frac{\delta}{19}$. For HKPR, we set $w_i=e^{-t}\cdot \frac{t^i}{i!}$, where $t>1$ is a constant. According to the Stirling formula~\cite{cam1935stirling} that $i!\ge 3\sqrt{i}\left(\frac{i}{e}\right)^i$, we have $w_i =e^{-t}\cdot \frac{t^i}{i!}<e^{-t}\cdot \left(\frac{et}{i}\right)^i<\left( \frac{1}{2}\right)^i$ for any $i\ge 2et$. Hence, $\sum_{i=L+1}^\infty w_i\le \sum_{i=L+1}^\infty \left( \frac{1}{2}\right)^i=\left( \frac{1}{2}\right)^L$. By setting $L=\max\{2et, O\left(\log{\frac{1}{\delta}}\right)\}=O\left(\log{\frac{1}{\delta}}\right)$, we can derive the bound that $\sum_{i=L+1}^\infty w_i\le \frac{\delta}{19}$. For transition probability that $w_L=1$ and $w_i=0$ if $i\neq L$, we have $\sum_{i=L+1}^\infty w_i =0 \le \frac{\delta}{19}$. Consequently, for PageRank, PPR, HKPR and transition probability, $\sum_{i=L+1}^\infty w_i \le \frac{\delta}{19}$ holds for $L=O\left(\log{\frac{1}{\delta}}\right)$. 

In the proximity model of Katz, we set $w_i=\beta^i$, where $\beta$ is a constant and set to be smaller than $\frac{1}{\lambda_{1}}$ to guarantee convergence. Here $\lambda_{1}$ is the maximum eigenvalue of the adjacent matrix $\A$. The probability transition probability becomes $\D^{-a}\A \D^{-b}=\A$ with $a=b=0$. It follows: 
\begin{align*}
\left\|\sum_{i=L+1}^\infty \hspace{-1.5mm} \beta^i  \cdot \A^i \cdot \vec{x} \right\|_2 \hspace{-1.5mm} =\hspace{-1.5mm} \sum_{i=L+1}^\infty \hspace{-1.5mm} \beta^i  \cdot \left\|\A ^i \cdot \vec{x} \right\|_2 \hspace{-0.5mm} \le \hspace{-2mm} \sum_{i=L+1}^\infty \hspace{-1.5mm} \beta^i \cdot \lambda_1^i\cdot \|\vec{x}\|_2 \le \hspace{-2mm} \sum_{i=L+1}^\infty \hspace{-2mm} \left(\beta \cdot \lambda_1\right)^i. 
\end{align*}
In the first equality, we apply the assumption that $\vec{x}$ is non-negative. In the first inequality, we use the fact that $\lambda_1$ is the maximum eigenvalue of matrix $\A$. And in the last inequality, we apply the assumption that $\|\vec{x}\|_1=1$ and the fact that $\|\vec{x}\|_2\le \|\vec{x}\|_1=1$. Note that $\beta\cdot \lambda_1$ is a constant and $\beta\cdot \lambda_1<1$, following $\sum_{i=L+1}^\infty \left(\beta \cdot \lambda_1\right)^i=\frac{\left(\beta \cdot \lambda_1\right)^{L+1}}{1-\beta\cdot \lambda_1}$. Hence, by setting $L=\log_{\beta \lambda_1}\left((1-\beta \lambda_1)\cdot \frac{\delta}{19}\right)=O\left(\log{\frac{1}{\delta}}\right)$, we have $\sum_{i=L+1}^\infty \left(\beta \cdot \lambda_1\right)^i \le \frac{\delta}{19}$. Consequently, in the proximity models of PageRank, PPR, HKPR, transition probability and Katz, by setting $L=O(\log{\frac{1}{\delta}})$, we can derive the bound 
\begin{align*}
    \left\|\sum_{i=L+1}^\infty \hspace{-0.5mm} w_i \cdot \hspace{-0.5mm} \left(\D^{-a}\A \D^{-b} \right)^i \hspace{-1.5mm}\cdot \vec{x}\right\|_2 \le \frac{\delta}{19}, 
\end{align*}
without Assumption~\ref{asm:L}. 

Recall that in the proof of Theorem~\ref{thm:prefix}, we show that according to the assumption on the non-negativity of $\vec{x}$ and the bound given in inequality~\eqref{eqn:boundL}, we only need to approximate the prefix sum $\vec{\pi}_L=\sum_{i=0}^L \hspace{-0.5mm} w_i \cdot \hspace{-0.5mm} \left(\D^{-a}\A \D^{-b} \right)^i \hspace{-1.5mm}\cdot \vec{x}$ such that for any $v\in V$ with $\vec{\pi}_L(v)>\frac{18}{19}\cdot \delta$, we have $\left|\vec{\pi}_L(v)-\hat{\vec{\pi}}(v)\right|\le \frac{1}{20} \vec{\pi}_L(v)$ holds with high probability. Hence, Theorem~\ref{thm:prefix} holds for all the proximity models discussed in this paper without Assumption~\ref{asm:L}, and this lemma follows. 

\end{proof}

\subsection{Proof of Lemma~\ref{lem:unbiasedness}}\label{sec:unbiasedness}
We first prove the unbiasedness of the estimated residue vector $\hat{\r}^{(\ell)}$ for each level $\ell \in [0,L]$. Let $X^{(\ell)}(u,v)$ denote the increment of $\hat{\vec{r}}^{(\ell)}(v)$ in the propagation from node $u$ at level $\ell-1$ to node $v \in N_u$ at level $\ell$. According to Algorithm~\ref{alg:AGP-RQ}, $X^{(\ell)}(u,v)=\frac{Y_{\ell}}{Y_{\ell-1}}\cdot \frac{\hat{\vec{r}}^{(\ell-1)}(u)}{d_v^a\cdot d_u^b}$ if $\frac{Y_{\ell}}{Y_{\ell-1}}\cdot \frac{\hat{\vec{r}}^{(\ell-1)}(u)}{d_v^a\cdot d_u^b}\ge \e$; otherwise, $X^{(\ell)}(u,v)=\e$ with the probability $\frac{Y_{\ell}}{\e \cdot Y_{\ell-1}}\hspace{-0.5mm}\cdot \hspace{-0.5mm} \frac{\hat{\vec{r}}^{(\ell-1)}(u)}{d_v^a\cdot d_u^b}$, or $0$ with the probability $1\hspace{-0.5mm}-\hspace{-0.5mm}\frac{Y_{\ell}}{\e \cdot Y_{\ell-1}}\cdot \frac{\hat{\vec{r}}^{(\ell-1)}(u)}{d_v^a\cdot d_u^b}$. Hence, the conditional expectation of $X^{(\ell)}(u,v)$ based on the obtained vector $\hat{\vec{r}}^{(\ell-1)}$ can be expressed as 
%Hence, we can derive the expectation of $X^{(\ell)}(u,v)$ conditioned on the vector of estimated residue $\hat{\vec{r}}^{(\ell-1)}$ as below. 
\begin{equation}\nonumber
%\label{eqn:incre2_expectation}
%\vspace{-2mm}
\begin{aligned}
\hspace{-1mm} \E \left[ X^{(\ell)}(u,v) \mid \hat{\vec{r}}^{(\ell-1)} \right]
&=\left\{
\begin{array}{ll}
\hspace{-2mm} \frac{Y_{\ell}}{Y_{\ell-1}}\cdot \frac{\hat{\vec{r}}^{(\ell-1)}(u)}{d_v^a\cdot d_u^b}, \quad if \frac{Y_{\ell}}{Y_{\ell-1}}\cdot \frac{\hat{\vec{r}}^{(\ell-1)}(u)}{d_v^a\cdot d_u^b}\ge \e\\
\hspace{-1mm} \e\cdot \frac{1}{\e}\cdot \frac{Y_{\ell}}{Y_{\ell-1}}\cdot \frac{\hat{\vec{r}}^{(\ell-1)}(u)}{d_v^a\cdot d_u^b}, \quad otherwise
\end{array} 
\right.\\
&= \frac{Y_{\ell}}{Y_{\ell-1}}\cdot \frac{\hat{\vec{r}}^{(\ell-1)}(u)}{d_v^a\cdot d_u^b}.
\end{aligned}
\end{equation}
%Note that we have $\hat{\vec{r}}^{(\ell)}(v)=\sum_{u \in N_v} X^{(\ell)}(u,v)$, following
Because $\hat{\vec{r}}^{(\ell)}(v)=\sum_{u \in N_v} X^{(\ell)}(u,v)$, it follows: 
\begin{equation}\label{eqn:con-exp}
\begin{aligned}
&\E \hspace{-0.5mm} \left[ \hat{\vec{r}}^{(\ell)}(v)\mid \hat{\vec{r}}^{(\ell-1)} \right]%=\E \left[ \sum_{u \in N(v)} X^{(i+1)}(u,v) \mid \hat{\vec{r}}^{(i)} \right]\\
\hspace{-1mm}=\hspace{-2mm}\sum_{u \in N_v} \hspace{-1mm}\E \hspace{-0.5mm}\left[ X^{(\ell)}(u,v) \mid \hat{\vec{r}}^{(\ell-1)} \right]\hspace{-1mm} =\hspace{-2mm}\sum_{u \in N_v} \hspace{-1.5mm}\frac{Y_{\ell}}{Y_{\ell-1}}\hspace{-0.5mm}\cdot \hspace{-0.5mm} \frac{\hat{\vec{r}}^{(\ell-1)}(u)}{d_v^a\cdot d_u^b}. 
\end{aligned}
\end{equation}
In the first equation. we use the linearity of conditional expectation. Furthermore, by the fact: $\E \hspace{-0.5mm}\left[\hat{\vec{r}}^{(\ell)}(v) \right]\hspace{-0.5mm}=\hspace{-0.5mm}\E \hspace{-0.5mm}\left[\E \hspace{-0.5mm}\left[ \hat{\vec{r}}^{(\ell)}(v)\mid \hat{\vec{r}}^{(\ell-1)}\hspace{-0.5mm} \right]\hspace{-0.5mm}\right]$, we have: 
\begin{equation}\label{eqn:proof-expectation}
\begin{aligned}
\E \left[ \hat{\vec{r}}^{(\ell)}(v) \right]=\E \left[\E \left[ \hat{\vec{r}}^{(\ell)}(v)\mid \hat{\vec{r}}^{(\ell-1)} \right]\right]=\hspace{-1mm} \sum_{u \in N_v} \frac{Y_{\ell}}{Y_{\ell-1}} \cdot \frac{\E \left[\hat{\vec{r}}^{(\ell-1)}(u)\right]}{d_v^a\cdot d_u^b}. 
\end{aligned}
\end{equation} 
%The first equality uses the fact that $\E \hspace{-0.5mm}\left[\hat{\vec{r}}^{(\ell)}(v) \right]\hspace{-0.5mm}=\hspace{-0.5mm}\E \hspace{-0.5mm}\left[\E \hspace{-0.5mm}\left[ \hat{\vec{r}}^{(\ell)}(v)\mid \hat{\vec{r}}^{(\ell-1)}\hspace{-0.5mm} \right]\hspace{-0.5mm}\right]$. 
%By equation~\eqref{eqn:proof-expectation}, 
Based on Equation~\eqref{eqn:proof-expectation}, we can prove the unbiasedness of $\hat{\vec{r}}^{(\ell)}(v)$ by induction. 
%Then we can derive the expectation of $\hat{\vec{r}}^{(\ell)}(v)$ by induction. 
Initially, we set $\hat{\vec{r}}^{(0)}=\vec{x}$. Recall that in Definition~\ref{def:RQ-relation}, the residue vector at level $\ell$ is defined as $\vec{r}^{(\ell)}=Y_\ell \cdot \left(\D^{-a}\A\D^{-b} \right)^\ell \cdot \vec{x}$. Hence, $\vec{r}^{(0)}=Y_0 \cdot \vec{x}=\vec{x}$. In the last equality, we use the property of $Y_0$ that $Y_0=\sum_{\ell=0}^\infty w_0=1$. Thus, $\hat{\vec{r}}^{(0)}=\vec{r}^{(0)}=\vec{x}$ and for $\forall u \in V$, $\E \left[ \hat{\vec{r}}^{(0)}(u) \right]=\vec{r}^{(0)}(u)$ holds in the initial stage. Assuming the residue vector is unbiased at the first $\ell-1$ level that $\E \left[ \hat{\vec{r}}^{(i)} \right]=\vec{r}^{(i)}$ for $\forall i \in [0,\ell-1]$, we want to prove the unbiasedness of $\hat{\r}^{(\ell)}$ at level $\ell$. Plugging the assumption $\E\left[ \hat{\r}^{(\ell-1)}\right]=\r^{(\ell-1)}$ into Equation~\eqref{eqn:proof-expectation}, we can derive: 
%Hence, the expectation of $\hat{\bm{r}}^{(\ell)}(v)$ can be derived as %we can derive the unbiasedness of residues that
\begin{equation}\nonumber
%\vspace{-2mm}
\begin{aligned}
\E \left[ \hat{\vec{r}}^{(\ell)}(v) \right]\hspace{-1mm}=\hspace{-2mm}\sum_{u \in N_v}\hspace{-1.5mm} \frac{Y_{\ell}}{Y_{\ell-1}} \hspace{-0.5mm} \cdot \hspace{-0.5mm} \frac{\E \left[\hat{\vec{r}}^{(\ell-1)}(u)\right]}{d_v^a\cdot d_u^b}\hspace{-1mm}=\hspace{-2mm}\sum_{u \in N_v}\hspace{-1mm} \frac{Y_{\ell}}{Y_{\ell-1}}\hspace{-0.5mm} \cdot \hspace{-0.5mm}\frac{\vec{r}^{(\ell-1)}(u)}{d_v^a\cdot d_u^b}=\vec{r}^{(\ell)}(v). 
\end{aligned}
\end{equation} 
In the last equality, we use the recursive formula of the residue vector that $\vec{r}^{(\ell)}\hspace{-0.5mm} =\hspace{-0.5mm} \frac{Y_{\ell}}{Y_{\ell-1}}\hspace{-0.5mm} \cdot \hspace{-0.5mm}\left(\mathbf{D}^{-a}\mathbf{A}\mathbf{D}^{-b} \right) \cdot \vec{r}^{(\ell-1)}$. Thus, the unbiasedness of the estimated residue vector follows. 
%\begin{equation}\nonumber
%	\begin{aligned}
%	\hspace{-0.5mm} \vec{r}^{(i+1)} 
	%=Y_{i+1}  \cdot \left(\mathbf{D}^{-a}\hspace{-1mm}\cdot \mathbf{A}^\top \hspace{-1.5mm} \cdot \mathbf{D}^{-b}\right)^{i+1} \hspace{-2.5mm} \cdot \vec{x}
%	= \frac{Y_{i+1}}{Y_i}\hspace{-0.5mm} \cdot \hspace{-0.5mm}\left(\mathbf{D}^{-a} \cdot \mathbf{A} \cdot \mathbf{D}^{-b} \right) \cdot \vec{r}^{(i)}. 
%	\end{aligned}
%\end{equation} 
%For $\forall v \in V$ and $\forall u \in N(v)$, it follows that
%\begin{equation}\nonumber
%	\begin{aligned}
%	\vec{r}^{(i+1)}(v) =\sum_{u \in N(v)} \frac{Y_{i+1}}{Y_i} \cdot \frac{\vec{r}^{(i)}(u)}{d_v^a\cdot d_u^b}. 
%	\end{aligned}
%\end{equation} 
	
Next we show that the estimated reserve vector is also unbiased at each level. % the expectation analysis of the estimated reserve vector $\hat{\vec{q}}^{(\ell)}$ for $\forall \ell\in \{0,1,...,L\}$. 
Recall that for $\forall v \in V$ and $\ell \ge 0$,  $\hat{\vec{q}}^{(\ell)}(v)=\frac{w_{\ell}}{Y_{\ell}}\cdot \hat{\vec{r}}^{(\ell)}(v)$. Hence, the expectation of $\hat{\vec{q}}^{(\ell)}(v)$ satisfies: 
%\vspace{-2mm}
\begin{equation}\nonumber%\label{eqn:unbiasedness-q}
%\vspace{-1mm}
\begin{aligned}
\E \left[ \hat{\vec{q}}^{(\ell)}(v) \right]=\E \left[\frac{w_{\ell}}{Y_{\ell}}\cdot \hat{\vec{r}}^{(\ell)}(v) \right]=\frac{w_{\ell}}{Y_{\ell}}\cdot \vec{r}^{(\ell)}(v)=\vec{q}^{(\ell)}(v). 
\end{aligned}
\end{equation} 
%The first equality uses the fact that $\vec{q}^{(\ell)}(v)=\frac{w_{\ell}}{Y_{\ell}} \cdot \vec{r}^{(\ell)}(v)$, and the lemma follows. 
Thus, the estimated reserve vector $\hat{\q}^{\ell}$ at each level $\ell \in [0,L]$ is unbiased and Lemma~\ref{lem:unbiasedness} follows.

\subsection{Proof of Lemma~\ref{lem:variance}}
%Recall that in Algorithm~\ref{alg:AGP-RQ}, the propagation from node $u$ at level $i$ will be conducted deterministically if $\left(1-\frac{w_i}{Y_i} \right) \cdot \frac{\hat{R}^{(i)}(u)}{d^{1-r}_v\cdot d^r_u}\ge \delta$. Otherwise, the randomness exists. Let $X^{(i+1)}(u,v)$ denote the increment of residue $\hat{R}^{(i+1)}(v)$ in the propagation from node $u$ at level $i$ to its neighbor $v \in N(u)$ at level $i+1$. When $\left(1-\frac{w_i}{Y_i} \right) \cdot \frac{\hat{R}^{(i)}(u)}{d^{1-r}_v\cdot d^r_u}< \delta$, 
%	\begin{equation}\nonumber
%	\begin{aligned}
%	X^{(i+1)}(u,v)
%	&=\hspace{-1mm} \left\{
%	\begin{array}{ll}
%	\delta, \quad w.p. \quad \frac{1}{\delta}\cdot \left(1-\frac{w_i}{Y_i} \right) \cdot \frac{\hat{R}^{(i)}(u)}{d^{1-r}_v\cdot d^r_u};\\
%	0, \quad w.p. \quad 1-\frac{1}{\delta}\cdot \left(1-\frac{w_i}{Y_i} \right) \cdot \frac{\hat{R}^{(i)}(u)}{d^{1-r}_v\cdot d^r_u},
%	\end{array} 
%	\right.\\
%	\end{aligned}
%	\end{equation}
%where the notation $w.p.$ means "with the probability". Thus, the variance of $X^{(i+1)}(u,v)$ conditioned on the vector of estimated residue $\hat{R}^{(i)}$ can be derived that
%\begin{equation}
%	\begin{aligned}
%		&\Var \left[ X^{(i+1)}(u,v) \mid \hat{R}^{(i)} \right] \le \E \left[ \left( X^{(i+1)}(u,v) \right)^2 \mid \hat{R}^{(i)} \right] \\
%		& = \delta^2 \cdot \frac{1}{\delta} \cdot \left(1-\frac{w_i}{Y_i} \right) \cdot \frac{\hat{R}^{(i)}(u)}{d^{1-r}_v\cdot d^r_u} = \delta \cdot \left(1-\frac{w_i}{Y_i} \right) \cdot \frac{\hat{R}^{(i)}(u)}{d^{1-r}_v\cdot d^r_u}.
%	\end{aligned}
%\end{equation}
%It follows that	

For any $v \in V$, we first prove that $\Var \left[ \vec{\epi}(v) \right]$ can expressed as: 
%\vspace{-2mm}
\begin{equation}\label{eqn:eq1}
%\vspace{-2mm}
\Var \left[ \vec{\epi}(v) \right]= \sum_{\ell=1}^{L} \E \left[ \Var \left[\vec{z}^{(\ell)}(v)\mid \hat{\vec{r}}^{(0)},...,\hat{\vec{r}}^{(\ell-1)}\right]\right],
\end{equation}
where 
%\vspace{-2mm}
\begin{equation}\label{eqn:def-z}
%\vspace{-2mm}
\vec{z}^{(\ell)}(v)=\sum_{i=0}^{\ell-1} \frac{w_i}{Y_i}\hat{\vec{r}}^{(i)}(v)+ \sum_{i=\ell}^{L}\sum_{u \in V} \frac{w_i}{Y_{\ell}} \cdot \hat{\vec{r}}^{(\ell)}(u)\cdot p_{i-\ell}(u,v).     
\end{equation}
In Equation~\eqref{eqn:def-z}, we use $p_{i-\ell}(u,v)$ to denote the $(i-\ell)$-th (normalized) transition probability from node $u$ to node $v$ that $p_{i-\ell}(u,v)=\vec{e}_v^\top \cdot \left(\D^{-a} \A \D^{-b}\right)^{i-\ell}\cdot \vec{e}_u$. % that the random walk starting from $u$ reaches node $v$ at the $(i-\ell)$-th step.  
%And $p_{i-\ell}(u,t)$ equals the $t_{th}$ entry of $\left(\mathbf{D}^{-(1-r)} \cdot \mathbf{A} \cdot \mathbf{D}^{-r} \right)^j \cdot \vec{e}_u$.
Based on Equation~\eqref{eqn:eq1}, we further show:  
%\vspace{-2mm}
\begin{equation}\label{eqn:eq2}
%\vspace{-1mm}
\E \left[ \Var \left[\vec{z}^{(\ell)}(v)\mid \hat{\vec{r}}^{(0)},...,\hat{\vec{r}}^{(\ell-1)}\right]\right] \hspace{-1mm} \le  (L-\ell+1)\cdot \e \vec{\pi}(v),     
\end{equation}
following $\Var\left[\epi(v)\right]\le \sum_{\ell=1}^L \e \vec{\pi}(v)=\frac{L(L+1)}{2} \cdot \e \vpi(v)$.  

\header{\bf Proof of Equation~\eqref{eqn:eq1}. } 
For $\forall v\in V$, we can rewrite $\z^{(L)}(v)$ as
%\vspace{-2mm}
\begin{equation}\nonumber
%\vspace{-1mm}
    \vec{z}^{(L)}(v)=\sum_{i=0}^{L-1} \frac{w_i}{Y_i}\hat{\vec{r}}^{(i)}(v)+ \sum_{u \in V} \frac{w_L}{Y_{L}} \cdot \hat{\vec{r}}^{(L)}(u)\cdot p_{0}(u,v). 
\end{equation}
Note that the $0$-th transition probability satisfies: $p_0(v,v)=1$ and $p_0(u,v)=0$ for each $u \neq v$, following: 
%\vspace{-2mm}
\begin{equation}\nonumber
%\vspace{-1mm}
\z^{(L)}(v)=\sum_{i=0}^L \frac{w_i}{Y_i}\er^{(i)}(v)=\sum_{i=0}^L \eq^{(i)}(v)=\epi(v),       
\end{equation}
where we use the fact:  $\frac{w_i}{Y_i}\er^{(i)}(v)\hspace{-0.5mm}=\hspace{-0.5mm}\eq^{(i)}(v)$ and $\sum_{i=0}^L \eq^{(i)}(v)\hspace{-0.8mm}=\hspace{-0.8mm}\epi(v)$. 
%Thus, $z^{(L)}(v)=\sum_{i=0}^L \hat{\vec{r}}(v)$ beacuse $\sum_{u \in V} \frac{w_L}{Y_{L}} \cdot \hat{\vec{r}}^{(L)}(u)\cdot p_{0}(u,v)= \frac{w_L}{Y_{L}} \cdot \hat{\vec{r}}^{(L)}(v)$. Recall that in Algorithm~\ref{alg:AGP-RQ}, we return $\sum_{i=0}^L \hat{\vec{q}}^{(i)}$ as the estimator of $\pi$. Thus, 
%\begin{equation}\label{eqn:variance-sum}
%	\begin{aligned}
%		\Var \left[ \vec{\epi}(v) \right]\hspace{-0.5mm}= \Var \left[ \sum_{i=0}^{L} \hat{\vec{q}}^{(\ell)}(v)\right] = \Var \left[ \sum_{i=0}^{L} \frac{w_i}{Y_i}\hat{\vec{r}}^{(i)}(v)\right]
%		=\Var \left[\vec{z}^{(L)}(v) \right]. 
%	\end{aligned}
%\end{equation}
    %The last equation uses the fact that $p_0(v,v)=1$ and $\sum_{u \in V} \frac{w_i}{Y_{\ell}} \cdot \hat{\vec{r}}^{(\ell)}(u)\cdot p_{i-\ell}(u,v)=\frac{w_i}{Y_{\ell}}\cdot \hat{\vec{r}}^{(\ell)}(v)$. 
Thus, $\Var[\z^{(L)}(v)]\hspace{-0.8mm}=\hspace{-0.8mm}\Var[\epi(v)]$. The goal to prove Equation~\eqref{eqn:eq1} is equivalent to show: 
%\vspace{-2mm}
\begin{equation}\label{eqn:sameEq1}
%\vspace{-1mm}
\Var[\z^{(L)}(v)]= \sum_{\ell=1}^{L} \E \left[ \Var \left[\vec{z}^{(\ell)}(v)\mid \hat{\vec{r}}^{(0)},...,\hat{\vec{r}}^{(\ell-1)}\right]\right]. 
\end{equation}
In the following, we prove Equation~\eqref{eqn:sameEq1} holds for each node $v\in V$. 
By the total variance law, $\Var[\z^{(L)}(v)]$ can be expressed as
%\vspace{-2mm}
\begin{equation}\label{eqn:var-1}
%\vspace{-2mm}
\begin{aligned}
\Var \left[ \vec{z}^{(L)}(v)\right]
&= \E \left[ \Var \left[\vec{z}^{(L)}(v) \mid \hat{\vec{r}}^{(0)},\hat{\vec{r}}^{(1)},...,\hat{\vec{r}}^{(L-1)}\right]\right]\\
&+\Var \left[ \E \left[\vec{z}^{(L)}(v) \mid \hat{\vec{r}}^{(0)},\hat{\vec{r}}^{(1)},...,\hat{\vec{r}}^{(L-1)}\right] \right]. 
\end{aligned}
\end{equation}
We note that the first term $\E \hspace{-0.5mm}\left[\hspace{-0.5mm} \Var \hspace{-0.5mm}\left[\hspace{-0.5mm}\vec{z}^{(L)}(v) \hspace{-1mm}\mid \hspace{-1mm} \hat{\vec{r}}^{(0)}\hspace{-1.5mm},...,\hat{\vec{r}}^{(L-1)}\hspace{-1mm}\right]\hspace{-0.5mm}\right]$ belongs to the final summation in Equation~\eqref{eqn:sameEq1}. And the second term $\Var \left[ \E \left[\vec{z}^{(L)}(v) \mid \hat{\vec{r}}^{(0)},...,\hat{\vec{r}}^{(L-1)}\right] \right]$ can be further decomposed as a summation of multiple terms in the form of $\E[\Var[.]]$. Specifically, for $\forall \ell \in \{0,1,...,L\}$, 
%Applying the linearity of conditional expectation, we have %for each $\ell \in \{0,1, 2, ... , L\}$, 
\begin{equation}\label{eqn:var-total-1}
\begin{aligned}
&\E \left[\vec{z}^{(\ell)}(v) \mid \hat{\vec{r}}^{(0)},...,\hat{\vec{r}}^{(\ell-1)}\right]\\
&\hspace{-1mm}=\hspace{-1mm}\E \left[\left( \sum_{i=0}^{\ell-1} \hspace{-0.5mm}\frac{w_i}{Y_i}\hat{\vec{r}}^{(i)}(v)\hspace{-0.5mm}+\hspace{-0.5mm} \sum_{i=\ell}^{L}\sum_{u \in V}\hspace{-0.5mm} \frac{w_i}{Y_{\ell}}\hspace{-0.5mm} \cdot \hspace{-0.5mm} \hat{\vec{r}}^{(\ell)}(u)\hspace{-0.5mm}\cdot \hspace{-0.5mm} p_{i-\ell}(u,v)\right) \mid \hat{\vec{r}}^{(0)}\hspace{-1mm},...,\hat{\vec{r}}^{(\ell-1)}\right]\\
&\hspace{-1mm}=\hspace{-1.5mm}\sum_{i=0}^{\ell-1} \frac{w_i}{Y_i}\hat{\vec{r}}^{(i)}(v)+ \hspace{-1mm}\sum_{i=\ell}^{L}\sum_{u \in V} \frac{w_i}{Y_{\ell}} \cdot p_{i-\ell}(u,v) \cdot \E \left[\hat{\vec{r}}^{(\ell)}(u) \mid \hat{\vec{r}}^{(0)},...,\hat{\vec{r}}^{(\ell-1)}\right]. 
\end{aligned}
\end{equation}
In the first equality, we plug into the definition formula of $\z^{\ell}(v)$ given in Equation~\eqref{eqn:def-z}. And in the second equality, we use the fact $\E \left[\sum_{i=0}^{\ell-1} \frac{w_i}{Y_i}\hat{\vec{r}}^{(i)}(v) \mid \hat{\vec{r}}^{(0)},...,\hat{\vec{r}}^{(\ell-1)}\right]= \sum_{i=0}^{\ell-1} \frac{w_i}{Y_i}\hat{\vec{r}}^{(i)}(v)$ and the linearity of conditional expectation. Recall that in the proof of Lemma~\ref{lem:unbiasedness}, we have $\E \left[\er^{(\ell)}(u) \mid \er^{(\ell-1)}\right]=\sum_{w \in N_u} \frac{Y_{\ell}}{Y_{\ell-1}} \cdot \frac{\hat{\vec{r}}^{(\ell-1)}(w)}{d_u^a\cdot d_w^b}$ given in Equation~\eqref{eqn:con-exp}. Hence, we can derive: 
\vspace{-1mm}
\begin{equation}\nonumber
\vspace{-1mm}
\begin{aligned}
&\E \left[\er^{(\ell)}(u) \mid \er^{(0)},...,\er^{(\ell-1)}\right]=\E \left[\er^{(\ell)}(u) \mid \er^{(\ell-1)}\right]\\
&=\sum_{w \in N_u} \frac{Y_{\ell}}{Y_{\ell-1}} \cdot \frac{\hat{\vec{r}}^{(\ell-1)}(w)}{d_u^a\cdot d_w^b}
=\sum_{w \in N_u}\frac{Y_{\ell}}{Y_{\ell-1}} \cdot \hat{\vec{r}}^{(\ell-1)}(w)\cdot p_1(w,u),      
\end{aligned}
%\vspace{-1mm}
\end{equation}
In the last equality, we use the definition of the $1$-th transition probability: $p_1(w,u)\hspace{-0.5mm}=\hspace{-1mm}\frac{1}{d_u^a\cdot d_w^b}$. 
Plugging into Equation~\eqref{eqn:var-total-1}, we have: 
%\vspace{-2mm}
\begin{equation}\label{eqn:r-l1}
\begin{aligned}
&\E \left[\vec{z}^{(\ell)}(v) \mid \hat{\vec{r}}^{(0)},...,\hat{\vec{r}}^{(\ell-1)}\right]\\
%&=\sum_{i=0}^{\ell-1} \frac{w_i}{Y_i}\hat{\vec{r}}^{(i)}(v)+\sum_{i=\ell}^{L}\sum_{u \in V}\sum_{w \in N_u} \frac{w_i}{Y_{\ell}} \cdot \frac{Y_{\ell}}{Y_{\ell-1}} \cdot \hat{\vec{r}}^{(\ell-1)}(w)\cdot p_{i-\ell}(u,v) \cdot p_1(w,u)\\ 
&=\sum_{i=0}^{\ell-1} \frac{w_i}{Y_i}\hat{\vec{r}}^{(i)}(v)+\sum_{i=\ell}^{L}\sum_{w \in V} \frac{w_i}{Y_{\ell-1}}\hat{\vec{r}}^{(\ell-1)}(w)\cdot p_{i-\ell+1}(w,v), 
\end{aligned}
\end{equation}
where we also use the property of the transition probability that $\sum_{u\in N_w} p_{i-\ell}(u,v)\cdot p_{1}(w,u)=p_{i-\ell+1}(w,v)$. More precisely, 
%\vspace{-2mm}
\begin{equation}\label{eqn:r-l1}
\begin{aligned}
%&\sum_{i=\ell}^{L}\sum_{u \in V} \frac{w_i}{Y_{\ell}} \cdot p_{i-\ell}(u,v) \cdot \E \left[\hat{\vec{r}}^{(\ell)}(u) \mid \hat{\vec{r}}^{(0)},...,\hat{\vec{r}}^{(\ell-1)}\right]\\
&\sum_{i=\ell}^{L}\sum_{u \in V}\sum_{w \in N_u} \frac{w_i}{Y_{\ell}} \cdot \frac{Y_{\ell}}{Y_{\ell-1}} \cdot \hat{\vec{r}}^{(\ell-1)}(w)\cdot p_{i-\ell}(u,v) \cdot p_1(w,u)\\
&=\sum_{i=\ell}^{L}\sum_{w \in V} \frac{w_i}{Y_{\ell-1}}\hat{\vec{r}}^{(\ell-1)}(w)\cdot p_{i-\ell+1}(w,v). 
\end{aligned}
\end{equation}
%In the last equality, we use the property of transition probability that $\sum_{u \in N_w}p_{i-\ell}(u,v)\cdot p_1(w,u)=p_{i-\ell+1}(w,v)$. By plugging Equation~\eqref{eqn:r-l1} into Equation~\eqref{eqn:var-total-1}, we have
Furthermore, we can derive: 
%\vspace{-1mm}
\begin{equation}\label{eqn:var-total-2}
\vspace{-1mm}
\begin{aligned}
&\E \left[\vec{z}^{(\ell)}(v) \mid \hat{\vec{r}}^{(0)},...,\hat{\vec{r}}^{(\ell-1)}\right]\\
&=\sum_{i=0}^{\ell-1} \frac{w_i}{Y_i}\hat{\vec{r}}^{(i)}(v)+\sum_{i=\ell}^{L}\sum_{w \in V} \frac{w_i}{Y_{\ell-1}}\hat{\vec{r}}^{(\ell-1)}(w)\cdot p_{i-\ell+1}(w,v)\\
&=\hspace{-1mm}\sum_{i=0}^{\ell-2} \frac{w_i}{Y_i}\hat{\vec{r}}^{(i)}(v)+ \hspace{-1mm}\sum_{i=\ell-1}^{L}\sum_{u \in V} \frac{w_i}{Y_{\ell}} \hat{\vec{r}}^{(\ell-1)}(u) \cdot p_{i-\ell+1}(u,v)
=\vec{z}^{(\ell-1)}(v), 
\end{aligned}
\end{equation}
%where we use the fact that
%\begin{equation}\nonumber
%\begin{aligned}
%&\sum_{i=\ell}^{L}\sum_{w \in V} \frac{w_i}{Y_{\ell-1}}\hat{\vec{r}}^{(\ell-1)}(w)\cdot p_{i-\ell+1}(w,v)\\
%&=\hspace{-2mm}\sum_{i=\ell-1}^{L}\sum_{u \in V} \frac{w_i}{Y_{\ell-1}}\hat{\vec{r}}^{(\ell-1)}(u)\cdot p_{i-\ell+1}(u,v)\hspace{-0.5mm}-\hspace{-2mm}\sum_{u\in V}\hspace{-1mm}\frac{w_{\ell-1}}{Y_{\ell-1}}\er^{(\ell-1)}(u)\cdot p_0(u,v). 
%\end{aligned}
%\end{equation}
In the last equality we use the fact: $\frac{w_i}{Y_i}\hat{\vec{r}}^{(\ell-1)}(v)=\sum_{u\in V} \frac{w_i}{Y_i}\hat{\vec{r}}^{(i)}(v) \cdot p_0(u,v)$ because $p_0(v,v)=1$, and  $p_0(u,v)=0$ if $u \neq v$. 
%The last equality also uses the properties of $0$-hop transition probability:  $p_0(v,v)=1$, and  $p_0(u,v)=0$ if $u \neq v$. 
Plugging Equation\eqref{eqn:var-total-2} into Equation~\eqref{eqn:var-1}, we have
\begin{equation}\nonumber%\label{eqn:var-2}
	\begin{aligned}
		\Var \left[ \vec{z}^{(L)}(v)\right]
		= \E \left[ \Var \left[\vec{z}^{(L)}(v) \mid \hat{\vec{r}}^{(0)},...,\hat{\vec{r}}^{(L-1)}\right]\right]
		\hspace{-0.5mm}+\hspace{-0.5mm}\Var \left[ \vec{z}^{(L-1)}(v) \right]. 
	\end{aligned}
\end{equation}
%We can repeat applying the total variance law to $\Var \left[ \vec{z}^{(L-1)}(v) \right]$ conditioned on the set $\{\hat{\vec{r}}^{(0)},\hat{\vec{r}}^{(1)},...,\hat{\vec{r}}^{(L-2)}\}$. Then $\Var \left[\vec{z}^{(L-1)}(v)\right]$ can be similarly expressed as
%\begin{equation}\nonumber
%\begin{aligned}
%\hspace{-1mm}\Var \left[\vec{z}^{(L-1)}(v)\right]\hspace{-1mm}=\hspace{-1mm}\E \hspace{-0.5mm} \left[\hspace{-0.5mm} \Var \hspace{-0.5mm}\left[\hspace{-0.5mm} \vec{z}^{(L-1)}(v) \mid \hat{\vec{r}}^{(0)},...,\hat{\vec{r}}^{(L-2)}\hspace{-0.5mm}\right]\hspace{-0.5mm}\right]\hspace{-1mm}+\hspace{-1mm}\Var \hspace{-0.5mm}\left[\hspace{-0.5mm} \vec{z}^{(L-2)}(v)\hspace{-0.5mm}\right], 
%\end{aligned}
%\end{equation}
%following $\Var \left[ \vec{z}^{(L)}(v)\right]\hspace{-1mm}=\hspace{-1mm}\sum_{\ell=L-1}^L \E \left[\Var \left[\vec{z}^{(\ell)}(v)\mid \hat{\vec{r}}^{(0)},...,\hat{\vec{r}}^{(\ell-1)}\right]\right]+\Var [\z^{L-2}(v)]$. 
Iteratively applying the above equation $L$ times, we can  derive: 
%By iterating the above process, we can finally express $\Var \left[ \vec{z}^{(L)}(v)\right]$ as: 
\vspace{-2mm}
\begin{equation}\nonumber
\vspace{-2mm}
\Var \left[ \vec{z}^{(L)}(v)\right]
\hspace{-1mm}=\hspace{-1mm}\sum_{\ell=1}^{L}\E \left[\Var \left[\vec{z}^{(\ell)}(v)\mid \hat{\vec{r}}^{(0)},...,\hat{\vec{r}}^{(\ell-1)}\right]\right]+\Var [\z^{(0)}(v)]. 
%&+\Var\left[ \sum_{i=0}^{L}\sum_{u \in V} \frac{w_i}{Y_{0}} \cdot \hat{\vec{r}}^{(0)}(u)\cdot p_{i}(u,v) \right]. 
\end{equation}
Note that $\Var [\z^{(0)}(v)]\hspace{-1mm}=\hspace{-1mm}\Var\left[ \sum_{i=0}^{L}\sum_{u \in V} \frac{w_i}{Y_{0}} \cdot \hat{\vec{r}}^{(0)}(u)\cdot p_{i}(u,v) \right]\hspace{-1mm}=\hspace{-1mm}0$ 
%And $\Var\left[ \sum_{i=0}^{L}\sum_{u \in V} \frac{w_i}{Y_{0}} \cdot \hat{\vec{r}}^{(0)}(u)\cdot p_{i}(u,v) \right]\hspace{-1mm}=\hspace{-1mm}0$ 
because we initialize $\hat{\vec{r}}^{(0)}\hspace{-0.5mm}=\hspace{-0.5mm}\vec{r}^{(0)}\hspace{-0.5mm}=\hspace{-0.5mm}\vec{x}$ deterministically. Consequently, 
\vspace{-3mm}
\begin{equation}\nonumber
%\vspace{-2mm}
\Var \left[ \vec{z}^{(L)}(v)\right]=\sum_{\ell=1}^{L}\E \left[\Var \left[\vec{z}^{(\ell)}(v)\mid \hat{\vec{r}}^{(0)},...,\hat{\vec{r}}^{(\ell-1)}\right]\right],     
\end{equation}
which follows Equation~\eqref{eqn:sameEq1}, and  Equation~\eqref{eqn:eq1} equivalently. 

\vspace{+1mm}
\header{\bf Proof of Equation~\eqref{eqn:eq2}. } 
In this part, we prove: 
\vspace{-1mm}
\begin{equation}\label{eqn:final-goal}
\vspace{-1mm}
\E \left[\Var \left[\vec{z}^{(\ell)}(v)\mid \hat{\vec{r}}^{(0)},...,\hat{\vec{r}}^{(\ell-1)}\right]\right]\le \e \vec{\pi}(v), 
\end{equation}
holds for $\ell \hspace{-0.8mm}\in \hspace{-0.8mm} \{1,...,L\}$ and $\forall v \hspace{-0.5mm}\in \hspace{-0.5mm}V$. Recall that $\vec{z}^{(\ell)}(v)$ is defined as $ \vec{z}^{(\ell)}(v)\hspace{-0.8mm}=\hspace{-0.8mm}\sum_{i=0}^{\ell-1}\hspace{-0.5mm} \frac{w_i}{Y_i}\hat{\vec{r}}^{(i)}(v)+ \sum_{i=\ell}^{L}\sum_{u \in V} \frac{w_i}{Y_{\ell}} \cdot \hat{\vec{r}}^{(\ell)}(u)\cdot p_{i-\ell}(u,v)$. Thus, we have
%\vspace{-1mm}
\begin{equation}\label{eqn:goal1}
\begin{aligned}
\vspace{-1mm}
&\E \left[\Var \left[\vec{z}^{(\ell)}(v)\mid \hat{\vec{r}}^{(0)},...,\hat{\vec{r}}^{(\ell-1)}\right]\right]=\\
&\E \left[\Var\hspace{-0.5mm} \left[\left(\sum_{i=0}^{\ell-1} \hspace{-1mm}\frac{w_i}{Y_i}\hat{\vec{r}}^{(i)}(v)\hspace{-0.5mm}+\hspace{-1mm}\sum_{i=\ell}^{L}\hspace{-0.5mm}\sum_{u \in V}\hspace{-0.5mm} \frac{w_i}{Y_{\ell}}\hspace{-0.5mm} \cdot \hspace{-0.5mm} \hat{\vec{r}}^{(\ell)}(u)\hspace{-0.5mm}\cdot \hspace{-0.5mm} p_{i-\ell}(u,v)\right)\hspace{-0.5mm} \mid \hspace{-0.5mm} \hat{\vec{r}}^{(0)}\hspace{-1mm},...,\hat{\vec{r}}^{(\ell-1)}\right]\right].    
\end{aligned}
\end{equation}
%Recall that in the proof of Lemma~\ref{lem:unbiasedness}, we use  $X^{(\ell)}(w,u)$ to denote the increment of $\hat{\vec{r}}^{(\ell)}(u)$ in the propagation from node $w$ at level $\ell-1$ to $u \in N(w)$ at level $\ell$. According to Algorithm~\ref{alg:AGP-RQ}, given the obtained $\{\hat{\vec{r}}^{(0)}\hspace{-1mm},...,\hat{\vec{r}}^{(\ell-1)}\}$, we introduce subset sampling to guarantee the independence among all $X^{(\ell)}(w,u)$ at level $\ell$ for $\forall w,u \in V$. Furthermore, based on the obtained $\{\hat{\vec{r}}^{(0)}\hspace{-1mm},...,\hat{\vec{r}}^{(\ell-1)}\}$, $\hat{\vec{r}}^{(\ell)}(u)=\sum_{w \in N(u)} X^{(\ell)}(w,u)$ is also independent to every $\er^{(i)}(v)$ for $\forall i \in \{0,...,\ell-1\}$.
Recall that in Algorithm~\ref{alg:AGP-RQ}, we introduce subset sampling to guarantee the independence of each   propagation from level $\ell-1$ to level $\ell$. Hence, after the propagation at the first $\ell-1$ level that the estimated residue vector $ \hat{\vec{r}}^{(0)}\hspace{-1mm},...,\hat{\vec{r}}^{(\ell-1)}$ are determined, $X^{(\ell)}(w,u)$ is independent of each $w,u \in V$. Here we use $X^{(\ell)}(w,u)$ to denote the increment of $\hat{\vec{r}}^{(\ell)}(u)$ in the propagation from node $w$ at level $\ell-1$ to $u \in N(w)$ at level $\ell$. Furthermore, with the obtained $\{\hat{\vec{r}}^{(0)}\hspace{-1mm},...,\hat{\vec{r}}^{(\ell-1)}\}$, $\hat{\vec{r}}^{(\ell)}(u)$ is independent of $\forall u \in V$ because $\hat{\vec{r}}^{(\ell)}(u)=\sum_{w \in N(u)} X^{(\ell)}(w,u)$. Thus, Equation~\eqref{eqn:goal1} can be rewritten as: 
\vspace{-1mm}
\begin{equation}\nonumber
\vspace{-1mm}
\begin{aligned}
&\E\left[\Var \left[\sum_{i=0}^{\ell-1} \frac{w_i}{Y_i}\hat{\vec{r}}^{(i)}(v)\mid  \hat{\vec{r}}^{(0)},...,\hat{\vec{r}}^{(\ell-1)}\right]\right]\\
&+\E \left[\Var \left[\sum_{i=\ell}^{L}\sum_{u \in V} \frac{w_i}{Y_{\ell}}\cdot  \hat{\vec{r}}^{(\ell)}(u)\cdot  p_{i-\ell}(u,v)\mid  \hat{\vec{r}}^{(0)},...,\hat{\vec{r}}^{(\ell-1)}\right]\right], 
\end{aligned}
\end{equation}
Note that $\E\left[\Var \left[\sum_{i=0}^{\ell-1} \frac{w_i}{Y_i}\hat{\vec{r}}^{(i)}(v)\mid  \hat{\vec{r}}^{(0)},...,\hat{\vec{r}}^{(\ell-1)}\right]\right]=0$. Thus, we can derive: 
\vspace{-1mm}
\begin{equation}\label{eqn:goal2}
\begin{aligned}
\vspace{-1mm}
&\E \left[\Var \left[\vec{z}^{(\ell)}(v)\mid \hat{\vec{r}}^{(0)},...,\hat{\vec{r}}^{(\ell-1)}\right]\right]\\
&=\E \left[\Var \left[\sum_{i=\ell}^{L}\sum_{u \in V} \frac{w_i}{Y_{\ell}}\cdot  \hat{\vec{r}}^{(\ell)}(u)\cdot  p_{i-\ell}(u,v)\mid  \hat{\vec{r}}^{(0)},...,\hat{\vec{r}}^{(\ell-1)}\right]\right], 
%=\sum_{i=\ell}^{L} \sum_{u \in V}\left( \frac{w_i}{Y_{\ell}} \cdot p_{i-\ell}(u,v) \right)^2\Var \left[ \hat{\vec{r}}^{(\ell)}(u)\mid \hat{\vec{r}}^{(\ell-1)}\right].
\end{aligned}
\end{equation}
%because $\Var \left[ \sum_{i=0}^{\ell-1}\frac{w_i}{Y_i}\hat{\vec{r}}^{(i)}(v)\mid \hat{\vec{r}}^{(0)},...,\hat{\vec{r}}^{(\ell-1)}\right]=0$.  %we can further express the goal as the bound of the expectation of Equation~\eqref{eqn:goal1} as: 
Furthermore, we utilize the fact $\hat{\vec{r}}^{(\ell)}(u)=\sum_{w \in N(u)} X^{(\ell)}(w,u)$ to rewrite Equation~\eqref{eqn:goal2} as: 
%the goal to bound Equation~\eqref{eqn:goal1} equivalently changes to give the bound of the following variance: 
%\vspace{-1mm}
\begin{equation}\label{eqn:goal3}
\vspace{-1mm}
\begin{aligned}
&\E \left[ \Var \left[\sum_{i=\ell}^{L}\sum_{u \in V} \sum_{w \in N_u}\hspace{-2mm}\frac{w_i}{Y_{\ell}}\cdot p_{i-\ell}(u,v) \cdot X^{(\ell)}(w,u)\mid  \hat{\vec{r}}^{(0)},...,\hat{\vec{r}}^{(\ell-1)}\right]\right]\\
&\hspace{-2mm}=\E \left[\sum_{u \in V}\hspace{-1mm} \sum_{w \in N(u)}\hspace{-1mm}\left(\sum_{i=\ell}^{L}\frac{w_i}{Y_{\ell}}\cdot p_{i-\ell}(u,v)\right)^2  \hspace{-2mm} \cdot\Var \left[X^{(\ell)}(w,u)\mid  \hat{\vec{r}}^{(0)}\hspace{-1mm},...,\hat{\vec{r}}^{(\ell-1)}\right]\right]. 
\end{aligned}
\end{equation}
%where we use the independence of $X^{(\ell)}(w,u)$ for $\forall w,u\in V$, based on the obtained $\{\hat{\vec{r}}^{(0)}\hspace{-1mm},...,\hat{\vec{r}}^{(\ell-1)}\}$. 
%Recall that in this part, our goal is to bound the expected variance given in Equation~\eqref{eqn:goal1}. Hence, this is equivalent to give the bound of the expectation of Equation~\eqref{eqn:goal2}. 
According to Algorithm~\ref{alg:AGP-RQ}, if $\frac{Y_{\ell}}{Y_{\ell-1}}\cdot \frac{\hat{\vec{r}}^{(\ell-1)}(w)}{d_u^a\cdot d_w^b}< \e$, $X^{(\ell)}(w,u)$ is increased by $\e$ with the probability $\frac{Y_{\ell}}{\e \cdot Y_{\ell-1}}\cdot \frac{\hat{\vec{r}}^{(\ell-1)}(w)}{d_u^a\cdot d_w^b}$, or $0$ otherwise. Thus, the variance of $X^{(\ell)}(w,u)$ conditioned on the obtained $\hat{\vec{r}}^{(\ell-1)}$ can be bounded as: 
%\vspace{-1mm}
\begin{equation}\label{eqn:tmp1}
%\vspace{-2mm}
\begin{aligned}
&\Var \left[X^{(\ell)}(w,u) \mid \hat{\vec{r}}^{(\ell-1)} \right] 
\le \E \left[ \left(X^{(\ell)}(w,u)  \right)^2 \mid \hat{\vec{r}}^{(\ell-1)} \right]\\
&\hspace{-1mm}=\e^2 \cdot \frac{1}{\e} \cdot \frac{Y_{\ell}}{Y_{\ell-1}} \cdot \frac{\hat{\vec{r}}^{(\ell-1)}(w)}{d_u^a\cdot d_w^b}
%=\frac{\delta \cdot Y_{i+1}}{Y_i} \cdot \frac{\hat{\vec{r}}^{(i)}(u)}{d_v^a\cdot d_u^b}
=\e \cdot \frac{Y_{\ell}}{Y_{\ell-1}} \cdot \hat{\vec{r}}^{(\ell-1)}(w) \cdot p_1(w,u), 
\end{aligned}
\end{equation}
where $p_1(w,u)=\frac{1}{d_u^a\cdot d_w^b}$ denotes the 1-hop transition probability. By plugging into Equation~\eqref{eqn:goal3}, we have: 
%\vspace{-1mm}
\begin{equation}\label{eqn:goal4}
%\vspace{-2mm}
\begin{aligned}
&\E \left[\Var \left[\vec{z}^{(\ell)}(v)\mid \hat{\vec{r}}^{(0)},...,\hat{\vec{r}}^{(\ell-1)}\right]\right]\\
&=\E \left[\sum_{u \in V} \sum_{w \in N_u}\hspace{-1mm}\left(\sum_{i=\ell}^{L}\frac{w_i}{Y_{\ell}}\hspace{-1mm} \cdot p_{i-\ell}(u,v) \right)^2\hspace{-2mm}\cdot \e \cdot \frac{Y_{\ell}}{Y_{\ell-1}} \cdot \hat{\vec{r}}^{(\ell-1)}(w) \cdot p_1(w,u)\right]\\
%&\le \E \left[ (L-\ell+1)\cdot \sum_{i=\ell}^{L}\sum_{w \in V} \frac{\e\cdot w_i}{Y_{\ell-1}}\cdot \hat{\vec{r}}^{(\ell-1)}(v)\cdot p_{i-\ell+1}(w,v)\right].  
\end{aligned}
\end{equation}
Using the fact: $\sum_{i=\ell}^L \frac{w_i}{Y_{\ell}}\hspace{-1mm} \cdot p_{i-\ell}(u,v) \le (L-\ell+1) $ and 
%\vspace{-1mm}
\begin{equation}\nonumber
%\vspace{-2mm}
\begin{aligned}
&\sum_{u\in V}\sum_{w\in N_u}  p_1(w,u)\cdot p_{i-\ell}(u,v)\\
&=\sum_{w\in V}\sum_{u\in N_w} p_1(w,u)\cdot p_{i-\ell}(u,v)=\sum_{w\in V}p_{i-\ell+1}(w,v),       
\end{aligned}
\end{equation}
we can further derive: 
%\vspace{-1mm}
\begin{equation}\label{eqn:goal4}
%\vspace{-2mm}
\begin{aligned}
&\E \left[\Var \left[\vec{z}^{(\ell)}(v)\mid \hat{\vec{r}}^{(0)},...,\hat{\vec{r}}^{(\ell-1)}\right]\right]\\
&\le \E \left[ (L-\ell+1)\cdot \sum_{i=\ell}^{L}\sum_{w \in V} \frac{\e\cdot w_i}{Y_{\ell-1}}\cdot \hat{\vec{r}}^{(\ell-1)}(w)\cdot p_{i-\ell+1}(w,v)\right],   
\end{aligned}
\end{equation}
It follows: 
%\vspace{-2mm}
%\vspace{-1mm}
\begin{equation}\label{eqn:goal5}
%\vspace{-1mm}
\begin{aligned}
&\E \left[\Var \left[\vec{z}^{(\ell)}(v)\mid \hat{\vec{r}}^{(0)},...,\hat{\vec{r}}^{(\ell-1)}\right]\right]\\
%&\E \left[\sum_{i=\ell}^{L}\sum_{w \in V} \frac{\e\cdot w_i}{Y_{\ell-1}}\cdot \hat{\vec{r}}^{(\ell-1)}(w) \cdot p_{i-\ell+1}(w,v)\right]\\
&\le (L-\ell+1)\cdot \sum_{i=\ell}^{L}\sum_{w \in V}\hspace{-1mm}\frac{\e\cdot w_i}{Y_{\ell-1}}\cdot \r^{(\ell-1)}(w)\cdot p_{i-\ell+1}(w,v). 
%=(L-\ell+1)\cdot\hspace{-1mm} \sum_{i=\ell}^{L}  \e \cdot \q^{(i)}(v). %\le \e \vpi(v).  
\end{aligned}
\end{equation}
by applying the linearity of expectation and the unbiasedness of $\hat{\vec{r}}^{(\ell-1)}(w)$ proved in Lemma~\ref{lem:unbiasedness}. 
Recall that in Definition~\ref{def:RQ-relation}, we define $\r^{(i)}=Y_i \left(\D^{-a} \A \D^{-b} \right)^i\cdot \vec{x}$. Hence, we can derive: 
%\vspace{-1mm}
\begin{align}\nonumber
%\vspace{-1mm}
\sum_{w\in V}\frac{1}{Y_{\ell-1}}\r^{(\ell-1)}(w)\cdot p_{i-\ell+1}(w,v)=\frac{1}{Y_i}\r^{(i)}(v),     
\end{align}
where we also use the definition of the $(i-\ell+1)$-th transition probability $p_{i-\ell+1}(w,v)=\vec{e}_v^\top \cdot \left(\D^{-a} \A \D^{-b} \right)^{i-\ell+1}\cdot \vec{e}_w$. 
Consequently, 
\begin{align}\nonumber
   \E \left[\Var \left[\vec{z}^{(\ell)}(v)\mid \hat{\vec{r}}^{(0)},...,\hat{\vec{r}}^{(\ell-1)}\right]\right] \le (L-\ell+1) \sum_{i=\ell}^L \frac{\e \cdot w_i}{Y_i}\r^{(i)}(v).  
\end{align}
Because $\q^{(i)}=\frac{w_i}{Y_i}\cdot r^{(i)}$ and $\sum_{i=\ell}^L \q^{(i)} \le \vec{\pi}$, we have: 
\begin{equation}\nonumber
    \E \left[\Var \left[\vec{z}^{(\ell)}(v)\mid \hat{\vec{r}}^{(0)},...,\hat{\vec{r}}^{(\ell-1)}\right]\right]\le \e(L-\ell+1) \cdot  \vec{\pi}(v). 
\end{equation}
Hence, Equation~\eqref{eqn:goal1} holds for $\forall \ell \in [0,L]$ and Lemma~\ref{lem:variance} follows. 

\subsection{Proof of Theorem~\ref{thm:RP-error}}
We first show the expected cost of Algorithm~\ref{alg:AGP-RQ} can be bounded as: 
%\vspace{-2mm}
\begin{equation}\nonumber
\vspace{-1mm}
\E \left[ C_{total} \right] \le \frac{1}{\e}\cdot \sum_{i=1}^{L} \left\| Y_i \cdot \left(\mathbf{D}^{-a}\mathbf{A}\mathbf{D}^{-b} \right)^i \cdot \vec{x} \right\|_1. 
\end{equation}
Then, by setting $\e=O\left(\frac{\delta}{L^2}\right)$, the theorem follows. 
 
For $\forall i \in \{1,...,L\}$ and $\forall u,v \in V$, let $C^{(i)}(u,v)$ denote the cost of the propagation from node $u$ at level $i-1$ to $v \in N(u)$ at level $i$. According to Algorithm~\ref{alg:AGP-RQ}, $C^{(i)}(u,v)=1$ deterministically if $\frac{Y_{i}}{Y_{i-1}} \cdot \frac{\hat{\r}^{(i-1)}(u)}{d_v^a\cdot d_u^b}\ge \e$. Otherwise, $C^{(i)}(u,v)=1$ with the probability $\frac{1}{\e}\cdot \frac{Y_{i}}{Y_{i-1}} \cdot \frac{\hat{\r}^{(i-1)}(u)}{d_v^a\cdot d_u^b}$, following 
%\vspace{-2mm}
\begin{equation}\nonumber
%\vspace{-1mm}
\begin{aligned}
\hspace{-1mm} \E \left[ C^{(i)}(u,v)\mid \hat{\vec{r}}^{(i-1)} \right]
&=\hspace{-1mm} \left\{
\begin{array}{ll}
1, \quad if \quad \frac{Y_{i}}{Y_{i-1}}\cdot \frac{\hat{\vec{r}}^{(i-1)}(u)}{d_v^a\cdot d_u^b}\ge  \e\\
1 \cdot \frac{1}{\e}\cdot \frac{Y_{i}}{Y_{i-1}} \cdot \frac{\hat{\vec{r}}^{(i-1)}(u)}{d_v^a\cdot d_u^b}, \quad otherwise
\end{array} 
\right.\\
&\le \frac{1}{\e}\cdot \frac{Y_{i}}{Y_{i-1}} \cdot \frac{\hat{\vec{r}}^{(i-1)}(u)}{d_v^a\cdot d_u^b}.
\end{aligned}
\end{equation}
Because $\E \left[ C^{(i)}(u,v) \right] \le \E \left[\E \left[ C^{(i)}(u,v)\mid \hat{\vec{r}}^{(i-1)} \right]\right]$, we have
%\vspace{-2mm}
\begin{equation}\nonumber
%\vspace{-1mm}
	\begin{aligned}
	\E \left[ C^{(i)}(u,v) \right]=\frac{1}{\e}\cdot \frac{Y_{i}}{Y_{i-1}}  \cdot \frac{\E \left[ \hat{\vec{r}}^{(i-1)}(u)\right]}{d_v^a\cdot d_u^b}
	=\frac{1}{\e}\cdot \frac{Y_{i}}{Y_{i-1}} \cdot \frac{\vec{r}^{(i-1)}(u)}{d_v^a\cdot d_u^b}, 
	\end{aligned}
\end{equation} 
where we use the unbiasedness of $\hat{\vec{r}}^{(i)}(u)$ shown in Lemma~\ref{lem:unbiasedness}. Let $C_{total}$ denotes the total time cost of Algorithm~\ref{alg:AGP-RQ} that $C_{total}=\sum_{i=1}^{L} \sum_{v\in V} \sum_{u \in N(v)} C^{(i)}(u,v)$. It follows: 
%\vspace{-2mm}
\begin{equation}\nonumber
%\vspace{-1mm}
	\begin{aligned}
	&\E \left[ C_{total} \right]=\sum_{i=1}^{L} \sum_{v\in V} \sum_{u \in N(v)} \E \left[C^{(i)}(u,v)\right]\\
	&\le\sum_{i=1}^{L} \sum_{v\in V} \sum_{u \in N(v)}\frac{1}{\e}\cdot \frac{Y_{i}}{Y_{i-1}} \cdot \frac{\vec{r}^{(i-1)}(u)}{d_v^a\cdot d_u^b}
	=\sum_{i=1}^{L} \sum_{v\in V} \frac{1}{\e}\cdot \vec{r}^{(i)}(v).
	\end{aligned}
\end{equation}
By Definition~\ref{def:RQ-relation}, we have $\vec{r}^{(i)}=Y_i \cdot \left(\mathbf{D}^{-a}\mathbf{A}\mathbf{D}^{-b} \right)^i \cdot \vec{x}$, following  
\begin{equation}\label{eqn:cost1}
\E \left[ C_{total} \right] \le \frac{1}{\e}\cdot \sum_{i=1}^{L} \left\| Y_i \cdot \left(\mathbf{D}^{-a} \mathbf{A} \mathbf{D}^{-b} \right)^i \cdot \vec{x} \right\|_1.
\end{equation}
%which follows the Lemma.
Recall that in Lemma~\ref{lem:variance}, we prove that the variance $\Var\left[\vec{\epi}(v) \right]$ can be bounded as: $\Var\left[\vec{\epi}(v) \right]\le \frac{L(L+1)}{2}\cdot \e \vec{\pi}(v)$. According to the Chebyshev's Inequality shown in Section~\ref{sec:chebyshev}, we have: 
\begin{equation}\label{eqn:chebypr}
	\begin{aligned}
		\Pr \{ \left|\vec{\pi}(v)-\vec{\epi}(v)\right| \ge \frac{1}{20}\cdot \vec{\pi}(v)\} \le \frac{L(L+1)\cdot \e \vec{\pi}(v)}{\frac{1}{200} \cdot \vec{\pi}^2(v)}=\frac{200L(L+1)\cdot \e}{\vec{\pi}(v)}. 
	\end{aligned}
\end{equation}
%Denote $\tilde{O}$ as the Big-Oh notation ignoring the log factors. 
For any node $v$ with $\vec{\pi}(v)> \frac{18}{19}\cdot \delta$, when we set $\e=\frac{0.01\cdot \delta}{200L(L+1)}=O \left( \frac{\delta}{L^2}\right)$, Equation~\eqref{eqn:chebypr} can be further expressed as: 
%\vspace{-2mm}
\begin{equation}\nonumber
%\vspace{-1mm}
	\begin{aligned}
		\Pr \{ \left|\vec{\pi}(v)-\vec{\epi}(v)\right| \ge \frac{1}{20}\cdot \vec{\pi}(v)\} \le \frac{0.01\cdot \delta}{\vec{\pi}(v)} < 0.01.  
	\end{aligned}
\end{equation}
Hence, for any node $v$ with $\pi(v)>\frac{18}{19}\cdot \delta$, $\Pr \{ \left|\vec{\pi}(v)-\vec{\epi}(v)\right| \ge \frac{1}{20}\cdot \vec{\pi}(v)\} $ holds with a constant probability ($99\%$), and the relative error in Definition~\ref{def:pro-relative} is also achieved according to Theorem~\ref{thm:prefix}. Combining with Equation~\eqref{eqn:cost1}, the expected cost of Algorithm~\ref{alg:AGP-RQ} satisfies 
%\vspace{-2mm}
\begin{equation}\nonumber
%\vspace{-1mm}
	\begin{aligned}
		\E \left[ C_{total} \right] 
		&\le \frac{1}{\e}\cdot \sum_{i=1}^{L} \left\| Y_i \cdot \left(\mathbf{D}^{-a}\mathbf{A} \mathbf{D}^{-b} \right)^i \cdot \vec{x} \right\|_1\\
		&=O\left(\frac{L^2}{\delta}\cdot \sum_{i=1}^{L} \left\| Y_i \cdot \left(\mathbf{D}^{-a}\mathbf{A} \mathbf{D}^{-b} \right)^i \cdot \vec{x} \right\|_1\right), 
	\end{aligned}
\end{equation}
which follows the theorem.

\subsection{Further Explanations on Theorem~\ref{thm:RP-error}}
According to Theorem~\ref{thm:RP-error}, the expected time cost of Algorithm~\ref{alg:AGP-RQ} is bounded as: 
\begin{equation*}
%\vspace{-1mm}
E[Cost]=\tilde{O}\left(\frac{1}{\delta}\cdot \sum_{i=1}^{L} \left\| Y_i \cdot \left(\mathbf{D}^{-a}\mathbf{A}\mathbf{D}^{-b} \right)^i \cdot \vec{x} \right\|_1\right),
\end{equation*}
where $\tilde{O}$ denotes the Big-Oh notation ignoring the log factors. 
%We can also show that in some ``bad" cases, the lower bound to return approximate propagation results with relative error $\delta$ is $\Omega{\frac{1}{\delta}}$. 
%We mention that the output size of the propagation can reach $\frac{1}{\delta}\cdot \sum_{i=1}^{L} \left\| w_i \cdot \left(\mathbf{D}^{-a}\mathbf{A}\mathbf{D}^{-b} \right)^i \cdot \vec{x} \right\|_1$. 
%Recall that by Pigeonhole principle, the output size of the propagation is $O\left(\frac{1}{\delta}\cdot \sum_{i=1}^{L} \left\| w_i \cdot \left(\mathbf{D}^{-a}\mathbf{A}\mathbf{D}^{-b} \right)^i \cdot \vec{x} \right\|_1\right)$. 
The following lemma shows when 
\begin{align*}
\sum_{i=1}^{L} \left\| Y_i \cdot \left(\mathbf{D}^{-a}\mathbf{A}\mathbf{D}^{-b} \right)^i \cdot \vec{x} \right\|_1\hspace{-1mm}=\tilde{O}\left(\sum_{i=0}^{L} \left\| w_i \cdot \left(\mathbf{D}^{-a}\mathbf{A}\mathbf{D}^{-b} \right)^i \cdot \vec{x} \right\|_1\right),
\end{align*}
the expect time cost of Algorithm~\ref{alg:AGP-RQ} is optimal up to log factors. 

\begin{lemma}\label{lem:nearoptimal}
%The output size of the propagation is lower bounded by $\Omega \left(\frac{1}{\delta}\cdot \sum_{i=0}^{L} \left\| w_i \cdot \left(\mathbf{D}^{-a}\mathbf{A}\mathbf{D}^{-b} \right)^i \cdot \vec{x} \right\|_1\right)$. 
When we ignore the log factors, the expected time cost of Algorithm~\ref{alg:AGP-RQ} is asymptotically the same as the lower bound of the output size of the graph propagation process if 
\begin{align}\nonumber %\label{eqn:label2}
\sum_{i=1}^{L} \left\| Y_i \cdot \left(\mathbf{D}^{-a}\mathbf{A}\mathbf{D}^{-b} \right)^i \cdot \vec{x} \right\|_1\hspace{-1mm}=\tilde{O}\left(\sum_{i=1}^{L} \left\| w_i \cdot \left(\mathbf{D}^{-a}\mathbf{A}\mathbf{D}^{-b} \right)^i \cdot \vec{x} \right\|_1\right), 
\end{align}
where $L=O\left(\log{\frac{1}{\delta}}\right)$. 
\end{lemma}

\begin{proof}
Let $C^*$ denote the output size of the propagation. We first prove that in some ``bad'' cases, the lower bound of $C^*$ is $\Omega \left(\frac{1}{\delta}\cdot \sum_{i=0}^{L} \left\| w_i \cdot \left(\mathbf{D}^{-a}\mathbf{A}\mathbf{D}^{-b} \right)^i \cdot \vec{x} \right\|_1\right)$. According to Theorem~\ref{thm:prefix}, to achieve the relative error in Definition~\ref{def:pro-relative}, we only need to compute the prefix sum $\sum_{i=0}^L w_i \left(\D^{-a} \A \D^{-b}\right)^i \cdot \vec{x}$ with constant relative error and $O(\delta)$ error threshold, where $L=O\left(\log \frac{1}{\delta} \right)$. Hence, by the Pigeonhole principle, the number of node $u$ with $\vec{\pi}(u)=O(\delta)$ can reach $\frac{1}{\delta}\cdot \hspace{-0.5mm}\left\| \sum_{i=0}^{\infty} \hspace{-0.5mm} w_i \hspace{-0.5mm}\cdot \hspace{-0.5mm} \left(\mathbf{D}^{-a}\mathbf{A}\mathbf{D}^{-b}\right)\right\|_1=\frac{1}{\delta}\cdot \hspace{-0.5mm}\sum_{i=0}^{\infty} \hspace{-0.5mm} w_i \hspace{-0.5mm}\cdot \hspace{-0.5mm} \left\| \left(\mathbf{D}^{-a}\mathbf{A}\mathbf{D}^{-b}\right)\right\|_1$, where apply the assumption on the non-negativity of $\vec{x}$. It follows the lower bound of $C^*$ as $\Omega\left(\frac{1}{\delta}\cdot \sum_{i=0}^{L} \left\| w_i \cdot \left(\mathbf{D}^{-a}\mathbf{A}\mathbf{D}^{-b} \right)^i \hspace{-2mm}  \cdot \vec{x} \right\|_1\right)$. Applying the assumptions that $\|\vec{x}\|=1$ and $\sum_{i=0}^\infty w_i=1$ given in Secition~\ref{asm:L}, we have $\left\|w_0 \cdot \vec{x}\right\|_1\le \left\|\vec{x}\right\|_1=1$, and the lower bound of $C^*$ becomes $\Omega\left(\frac{1}{\delta}\cdot \sum_{i=1}^{L} \left\| w_i \cdot \left(\mathbf{D}^{-a}\mathbf{A}\mathbf{D}^{-b} \right)^i \hspace{-2mm}  \cdot \vec{x} \right\|_1\right)$. When $\sum_{i=1}^{L} \hspace{-0.5mm} \left\| Y_i \cdot \left(\mathbf{D}^{-a}\mathbf{A}\mathbf{D}^{-b} \right)^i \cdot \vec{x} \right\|_1\hspace{-2mm}=\hspace{-0.5mm}\tilde{O}\left(\sum_{i=1}^{L} \left\| w_i \cdot \left(\mathbf{D}^{-a}\mathbf{A}\mathbf{D}^{-b} \right)^i \hspace{-2mm} \cdot \vec{x} \right\|_1\right)$, the expected time cost of Algorithm~\ref{alg:AGP-RQ} is asymptotically the same as the lower bound of the output size $C^*$ ignoring the log factors, which follows the lemma. 
\end{proof}

In particular, for all proximity models discussed in this paper, the optimal condition in Lemma~\ref{lem:nearoptimal}:  
\begin{align}\nonumber %\label{eqn:label2}
\sum_{i=1}^{L} \left\| Y_i \cdot \left(\mathbf{D}^{-a}\mathbf{A}\mathbf{D}^{-b} \right)^i \cdot \vec{x} \right\|_1\hspace{-1mm}=\tilde{O}\left(\sum_{i=1}^{L} \left\| w_i \cdot \left(\mathbf{D}^{-a}\mathbf{A}\mathbf{D}^{-b} \right)^i \cdot \vec{x} \right\|_1\right). 
\end{align}
is satisfied. Specifically, for PageRank and PPR, we set $w_i=\alpha \cdot (1-\alpha)^i$ and $Y_i=(1-\alpha)^i$, where $\alpha$ is a constant in $(0,1)$. %Hence, we can derive $Y_i=O(w_i)$. %Hence, $\E [Cost]$ is asymptotically the same as the lower bound of $C^*$ when we ignore log factors. 
Hence,  
\begin{align*}
    \frac{1}{\delta}\cdot \hspace{-1mm} \sum_{i=1}^{L} \left\| Y_i \cdot \left(\mathbf{D}^{-a}\mathbf{A}\mathbf{D}^{-b} \right)^i \cdot \vec{x} \right\|_1 \hspace{-2mm}=\hspace{-0.5mm}O\left(\frac{1}{\delta}\cdot \hspace{-0.5mm} \sum_{i=1}^{L} \left\| w_i \cdot \left(\mathbf{D}^{-a}\mathbf{A}\mathbf{D}^{-b} \right)^i \cdot \vec{x} \right\|_1\right). 
\end{align*}
For HKPR, $w_i=e^{-t}\cdot \frac{t^i}{i!}$ and $Y_i=\sum_{\ell=i}^\infty e^{-t}\cdot \frac{t^i}{i!}=e^{-t}\cdot \frac{(et)^i}{i^i}$. According to the Stirling's formula~\cite{cam1935stirling} that $\left(\frac{e}{i}\right)^i\le \frac{e\sqrt{i}}{i!}\le \frac{e\sqrt{L}}{i!}$, we can derive: $Y_i=O\left(\left( e \sqrt{\log \frac{1}{\delta}}\right)\cdot w_i\right)$ as $L=O(\log \frac{1}{\delta})$, following: 
\begin{align*}
    \frac{1}{\delta}\cdot \hspace{-1mm} \sum_{i=1}^{L} \left\| Y_i \cdot \left(\mathbf{D}^{-a}\mathbf{A}\mathbf{D}^{-b} \right)^i \cdot \vec{x} \right\|_1 \hspace{-2mm}=\hspace{-0.5mm}\tilde{O}\left(\frac{1}{\delta}\cdot \hspace{-0.5mm} \sum_{i=1}^{L} \left\| w_i \cdot \left(\mathbf{D}^{-a}\mathbf{A}\mathbf{D}^{-b} \right)^i \cdot \vec{x} \right\|_1\right). 
\end{align*}
For transition probability, $w_L=1$ and $w_i=0$ if $i \neq L$. Thus, $Y_i=1$ for $\forall i\le L$. Hence, we have: 
\begin{align*}
    &\frac{1}{\delta}\cdot \sum_{i=1}^{L} \left\| Y_i \cdot \left(\mathbf{D}^{-a}\mathbf{A}\mathbf{D}^{-b} \right)^i \cdot \vec{x} \right\|_1 \hspace{-2mm}= \frac{L}{\delta}\cdot \sum_{i=1}^{L} \left\| w_i \cdot \left(\mathbf{D}^{-a}\mathbf{A}\mathbf{D}^{-b} \right)^i \cdot \vec{x} \right\|_1\\
    &=\hspace{-0.5mm}\tilde{O}\left(\frac{1}{\delta}\cdot \hspace{-0.5mm} \sum_{i=1}^{L} \left\| w_i \cdot \left(\mathbf{D}^{-a}\mathbf{A}\mathbf{D}^{-b} \right)^i \cdot \vec{x} \right\|_1\right). 
\end{align*}
In the last equality, we apply the fact that $L=O\left(\log \frac{1}{\delta}\right)$. For Katz, $w_i=\beta^i$ and $Y_i=\frac{\beta^i}{1-\beta}$, where $\beta$ is a constant and is set to be smaller than the reciprocal of the largest eigenvalue of the adjacent matrix $\A$. Similarly, we can derive: 
\begin{align*}
    \frac{1}{\delta}\cdot \hspace{-1mm} \sum_{i=1}^{L} \left\| Y_i \cdot \left(\mathbf{D}^{-a}\mathbf{A}\mathbf{D}^{-b} \right)^i \cdot \vec{x} \right\|_1 \hspace{-2mm}=\hspace{-0.5mm}O\left(\frac{1}{\delta}\cdot \hspace{-0.5mm} \sum_{i=1}^{L} \left\| w_i \cdot \left(\mathbf{D}^{-a}\mathbf{A}\mathbf{D}^{-b} \right)^i \cdot \vec{x} \right\|_1\right). 
\end{align*}
%Note that the expected time cost of Algorithm~\ref{alg:AGP-RQ} is bounded as  $E[Cost]=\tilde{O}\left(\frac{1}{\delta}\cdot \sum_{i=1}^{L} \left\| Y_i \cdot \left(\mathbf{D}^{-a}\mathbf{A}\mathbf{D}^{-b} \right)^i \cdot \vec{x} \right\|_1\right)$, and the output size of the propagation $C^*=\Omega\frac{1}{\delta}\cdot \hspace{-0.5mm} \sum_{i=0}^{L} \left\| w_i \cdot \left(\mathbf{D}^{-a}\mathbf{A}\mathbf{D}^{-b} \right)^i \cdot \vec{x} \right\|_1$. 
Consequently, in the proximity models of PageRank, PPR, HKPR, transition probability and Katz, by ignoring log factors, the expected time cost of Algorithm~\ref{alg:AGP-RQ} is asymptotically the same as the lower bound of the output size $C^*$, and thus is near optimal up to log factors.

\end{document}